\documentclass[11pt]{amsart} 

\newcommand{\wasisdas}{{\sc Canonical Bose Ensemble}}

\newcommand{\intremark}[1]{\par\bigskip\noin\hrulefill\par\medskip \noin{\bf Internal Remark:} #1 \par\medskip\noin\hrulefill\par\bigskip}
\renewcommand{\intremark}[1]{}

\usepackage[papersize={210mm,297mm},left=30mm,right=50mm,top=20mm,bottom=20mm]{geometry}
\usepackage{fancyhdr}
\usepackage{amssymb}
\usepackage{amsthm}

\usepackage{xcolor}
\definecolor{slategray}{RGB}{112,138,144}
\definecolor{hintergrundblau}{RGB}{0 ,43 ,54}
\definecolor{wichtigrot}{RGB}{220,50,47}
\definecolor{textgrau}{RGB}{131,148,150}
\definecolor{ultramarinblau}{RGB}{32,33,79}
\definecolor{kobaltblau}{RGB}{35,45,83}
\definecolor{navy}{RGB}{0 ,0,128}
\definecolor{leuchtgruen}{RGB}{0,181,26}
\definecolor{leuchtrot}{RGB}{255,77,6} 

\definecolor{structuresgray}{RGB}{162,173,191}

\definecolor{cp1}{rgb}{94, 130, 182} 
\definecolor{cp2}{rgb}{226, 156, 36} 
\definecolor{cp3}{rgb}{143, 177, 50} 
\definecolor{cp4}{rgb}{236, 99, 54} 
\definecolor{cp5}{rgb}{135, 121, 179} 
\definecolor{cp6}{rgb}{198, 111, 26} 
\definecolor{cp7}{rgb}{93, 158, 200} 

\definecolor{firebrick}{RGB}{178,34,34}
\definecolor{darkorange}{RGB}{255,140,0}
\definecolor{darkgreen}{RGB}{0,100,0}
\definecolor{seagreen}{RGB}{46,139,87}
\definecolor{lightseagreen}{RGB}{32,178,170}
\definecolor{forestgreen}{RGB}{34,139,34}
\definecolor{midnightblue}{RGB}{25,25,112}
\definecolor{navyblue}{RGB}{0,0,128}
\definecolor{cornflowerblue}{RGB}{100,149,237}
\definecolor{mediumblue}{RGB}{0,0,205}
\definecolor{dimgray}{RGB}{105,105,105}
\definecolor{slategray}{RGB}{112,138,144}
\definecolor{Korrektur}{RGB}{0,0,205}

\usepackage{mathrsfs}

\usepackage[colorlinks=true]{hyperref}
\hypersetup{linkcolor=slategray,urlcolor=slategray,citecolor=slategray}

\newtheorem{satz}{Theorem}

\newtheorem{lemma}[satz]{Lemma}

\newcommand{\beq}{\begin{equation}}
\newcommand{\eeq}{\end{equation}}
\newcommand{\neeq}{\nonumber\end{equation}}
\newcommand{\bea}{\begin{eqnarray}}
\newcommand{\eea}{\end{eqnarray}}
\newcommand{\beast}{\begin{eqnarray*}}
\newcommand{\eeast}{\end{eqnarray*}}
\newcommand{\beal}{\begin{align}}
\newcommand{\eeal}{\end{align}}
\newcommand{\bsp}{\begin{split}}
\newcommand{\esp}{\end{split}}

\newcommand{\bbbone}{{\mathchoice {\rm 1\mskip-4mu l} {\rm 1\mskip-4mu l}    {\rm 1\mskip-4.5mu l} {\rm 1\mskip-5mu l}}}

\newcommand{\bC}{{\mathbb C}}

\newcommand{\bL}{{\mathbb L}}

\newcommand{\bN}{{\mathbb N}}

\newcommand{\bQ}{{\mathbb Q}}
\newcommand{\bR}{{\mathbb R}}
\newcommand{\bS}{{\mathbb S}}
\newcommand{\bT}{{\mathbb T}}

\newcommand{\bV}{{\mathbb V}}

\newcommand{\bX}{{\mathbb X}}

\newcommand{\bZ}{{\mathbb Z}}
\newcommand{\al}{\alpha}

\newcommand{\de}{\delta}

\newcommand{\vphi}{\varphi}

\newcommand{\veps}{\varepsilon}
\newcommand{\bGa}{{{\rm I}\kern-.16em \Gamma}}
\newcommand{\cA}{{\mathcal A}}
\newcommand{\cB}{{\mathcal B}}
\newcommand{\cC}{{\mathcal C}}
\newcommand{\cD}{{\mathcal D}}
\newcommand{\cE}{{\mathcal E}}
\newcommand{\cF}{{\mathcal F}}
\newcommand{\cG}{{\mathcal G}}
\newcommand{\cH}{{\mathcal H}}

\newcommand{\cK}{{\mathcal K}}

\newcommand{\cQ}{{\mathcal Q}}

\newcommand{\cS}{{\mathcal S}}
\newcommand{\cT}{{\mathcal T}}
\newcommand{\cU}{{\mathcal U}}
\newcommand{\cV}{{\mathcal V}}

\newcommand{\cY}{{\mathcal Y}}

\newcommand{\scB}{{\mathscr B}}

\newcommand{\scF}{{\mathscr F}}

\newcommand{\scP}{{\mathscr P}}

\newcommand{\scS}{{\mathscr S}}

\newcommand{\tr}{\mbox{ tr }}
\newcommand{\del}{\partial}
\newcommand{\gtoas}[1]{{\;\mathop{\longrightarrow}\limits_{#1}\;}}

\newcommand{\abs}[1]{{\left\vert #1 \right\vert}}
\newcommand{\norm}[1]{{\left\Vert #1 \right\Vert}}

\newcommand{\Perm}[1]{{\scS}_{#1}}
\newcommand{\Ref}[1]{$(\ref{#1})$}
\newcommand{\True}[1]{\; \bbbone\left( #1 \right) \;}
\newcommand{\sfrac}[2]{{\textstyle \frac{#1}{#2}}}

\newcommand{\const}{\hbox{ \rm const }}

\newcommand{\ili}{\int\limits}
\newcommand{\sli}{\sum\limits}
\newcommand{\pli}{\prod\limits}
\newcommand{\lli}{\lim\limits}
\newcommand{\noin}{\noindent}

\newcommand{\E}{{\rm e}}
\newcommand{\I}{{\rm i}}
\newcommand{\dd}{{\rm d}}

\usepackage{graphicx}

\renewcommand\Re{\operatorname{Re}}
\renewcommand\Im{\operatorname{Im}}

\newcommand{\Tr}{{\rm Tr}}

\newcommand{\ketbra}[2]{|#1\rangle\langle#2|}
\newcommand{\bracket}[2]{\langle#1\mid #2\rangle}
\renewcommand{\bbbone}{1}

\newcommand{\yywalk}{\yy}
\newcommand{\ywalk}[1]{{\yy}^{(#1)}}
\newcommand{\ywalkel}[2]{{\yy}^{(#1)}_{#2}}
\newcommand{\ddens}{\nu}
\newcommand{\dens}[1]{\nu^{(#1)}}
\newcommand{\Pamp}[1]{P(#1)}

\newcounter{Liste}

\newcommand{\ntau}{\ell} 
\newcommand{\itau}{\imath} 
\newcommand{\jtau}{\jmath} 
\newcommand{\eptau}{\epsilon}
\newcommand{\Xspace}{{\rm X}} 
\newcommand{\XTspace}{\bX} 
\newcommand{\speps}{\eta}
\newcommand{\hebeps}{\veps}
\newcommand{\Quad}{\cQ}
\newcommand{\Queps}{\cQ^{(\eptau)}}
\newcommand{\Cov}{\cC}
\newcommand{\boQ}{\bQ}
\newcommand{\boQloc}{\bL}
\newcommand{\tQuad}{\cQ_1}
\newcommand{\ttQuad}{\cQ_2}
\newcommand{\tttQuad}{\cQ_3}
\newcommand{\tttaQuad}{\!\cQ_4}
\newcommand{\tCov}{\Cov_1}
\newcommand{\ttCov}{\Cov_2}
\newcommand{\tttCov}{\Cov_3}
\newcommand{\tttaCov}{\Cov_4}
\newcommand{\tttaS}{\cS}
\newcommand{\tttaA}{\cA}
\newcommand{\tttaZ}[1]{
{\ZZ{#1}}}
\newcommand{\tttaY}[1]{\cY_{#1}}
\newcommand{\Kuad}{\cK}
\newcommand{\Gov}{\cG}
\newcommand{\Proj}{{\opP}}

\newcommand{\Ham}{{\rm H}}

\newcommand{\Rad}{R}
\newcommand{\WW}{{\rm v}}
\newcommand{\WWW}{{\bV}}
\newcommand{\aphi}{a}

\newcommand{\bphi}{b}
\newcommand{\cphi}{c}
\newcommand{\bpht}[1]{b(#1)}
\newcommand{\cpht}[1]{c(#1)}
\newcommand{\bphj}[1]{b_{#1}}
\newcommand{\cphj}[1]{c_{#1}}

\usepackage{extarrows}

\newcommand{\acre}{{\opa}^\dagger}
\newcommand{\aann}{{\opa}^{\phantom{\dagger}}}
\newcommand{\aan}{{\opa}}

\newcommand{\FockB}{\scF_B}

\newcommand{\FckoB}[1]{\scF_B^{(#1)}}

\newcommand{\opId}{{\tt 1}}
\newcommand{\opA}{{\tt A}}
\newcommand{\opB}{{\tt B}}

\newcommand{\opD}{{\tt D}}

\newcommand{\opF}{{\tt F}}

\newcommand{\opH}{{\tt H}}

\newcommand{\opN}{{\tt N}}
\newcommand{\opP}{{\tt P}}
\newcommand{\opQ}{{\tt Q}}
\newcommand{\opR}{{\tt R}}

\newcommand{\opV}{{\tt V}}
\newcommand{\opW}{{\tt W}}
\newcommand{\opa}{{\tt a}}

\newcommand{\opn}{{\tt n}}

\newcommand{\unnex}[1]{\left[ #1\right]}

\newcommand{\ee}{{\rm e}}

\newcommand{\ww}{{\rm w}}
\newcommand{\yy}{{\tt y}}

\newcommand{\vvv}{{v}}
\newcommand{\vv}{{\tt v}}
\newcommand{\bogv}{{\tt w}}
\newcommand{\coh}[1]{\vv_{#1}}
\newcommand{\klamma}[2]{\langle\mkern-5mu\langle #1 |\mkern-3mu| #2\rangle\mkern-5mu\rangle}
\newcommand{\bili}[2]{(#1\mid#2)}
\newcommand{\KIP}[2]{\klamma{#1}{#2}}

\renewcommand{\Re}{\mbox{Re}\;}
\renewcommand{\Im}{\mbox{Im}\;}

\newcommand{\xx}{{\rm x}}
\newcommand{\xy}{{\rm y}}
\newcommand{\xz}{{\rm z}}
\newcommand{\xp}{{\rm p}}

\newcommand{\unicoh}[1]{\kappa_{#1}}
\newcommand{\DD}{\cD}
\newcommand{\ZZ}[1]{\mathfrak{Z}_{#1}}
\newcommand{\tilZ}{\canZ}
\newcommand{\ttilZ}{Z_2}
\newcommand{\Kin}{\cE}
\newcommand{\tKin}{\tilde \cE}

\newcommand{\Tube}{\cT}

\newcommand{\canZ}{Z_c}
\newcommand{\ganZ}{Z_g}

\renewcommand{\True}[1]{\bbbone_{#1}} 

\newcommand{\Diff}{\Delta}
\newcommand{\oh}{{\tt h}}
\newcommand{\hexpec}[1]{\left\langle #1\right \rangle_h}

\newcommand{\Bana}{\scB}

\newcommand{\Prm}[2]{\scP_{#1,#2}}
\newcommand{\uu}{{\rm u}}
\newcommand{\mm}{{\rm m}}

\newcommand{\Action}{\cS}

\begin{document}

\title[]{Functional Integral and Stochastic Representations for Ensembles of Identical Bosons on a Lattice}
\author{Manfred Salmhofer} 
\address{ 
Institut f\" ur Theoretische Physik, Universit\" at Heidelberg,
Philosophenweg~19, 69120 Heidelberg, Germany}

\email{salmhofer@uni-heidelberg.de}

\date{\today}

\begin{abstract}
\noindent
Regularized coherent-state functional integrals are derived for ensembles of identical bosons on a lattice, the regularization being a discretization of Euclidian time. Convergence of the time-continuum limit is shown for various discretized actions. The focus is on the integral representation for the partition function and expectation values in the canonical ensemble. The connection to the grand-canonical integral, and a number of differences, are discussed. Uniform bounds for covariances are proven, which simplify the analysis of the time-continuum limit and can also be used to analyze the thermodynamic limit. The relation to a stochastic representation by an ensemble of interacting random walks is made explicit, and its modifications in presence of a condensate are discussed. 
\end{abstract}

\maketitle

%

\section{Introduction}
The operator-algebraic (`second-quantized') formulation of quantum statistical mechanics \cite{BR} provides a precise mathematical framework for the study of quantum many-body models. For the analysis of equilibrium states and dynamics at (small) positive temperatures, functional integral representations (FR) and stochastic representations, in particular interacting Brownian motions (IBM), have proven very useful alternative approaches. In theoretical physics, formal (i.e.\ nonrigorous) FIR have almost become the method of choice; stochastic representations also have a long history, going back to Feynman \cite{Feynman1,Feynman2}. In mathematical physics, FIR have also been studied  already for a long time. Clearly, there are subtleties involved when dealing mathematically with infinite-dimensional integrals. A traditional way of obtaining natural regularizations that also provide a strict relation to the algebraic formulation proceeds via Suzuki-Trotter product formulas and coherent-state integrals, to get a discrete-time functional integral, the time-continuum limit of which can then be analyzed. For fermionic systems, this has (in combination with multiscale techniques) led to mathematical results that have not yet been proven by any other method, see for instance \cite{BGPS,DR1,DR2,FKTfl,FST4,GM,GMP,crg}. Bosonic systems pose the additional technical problem of dealing with the unboundedness of operators, which in the FIR translates to the unboundedness of fields. This has been dealt with using decompositions into large and small fields. Using these has led to proofs of existence of massless scalar field theories with a fixed short-distance regulator \cite{GK,FMRS}. More recently, the existence of the time-continuum limit of the Bose system (again with a fixed short-distance regularization) has been proven \cite{BFKT3}. 

IBM representations start from the Feynman-Kac formula. Given this formula for a single particle, it is straightforward to generalize it to $N$ particles. Unlike the coherent-state formalism, where (anti)symmetrization is built in, Bose (resp. Fermi) statistics then needs to be imposed by an explicit (anti)symmetrization operation, given by a (signed) sum over permutations of the labels of the $N$ particles. The IBM allows to bring large-deviation techniques to bear on the problem \cite{AdamsBruKoenig,AdamsCollevecchioKoenig,AdamsKoenig}. It is naturally suited for the canonical ensemble (but can also be used in the grand canonical one). 

For Bose systems, it is a major outstanding problem in mathematical physics to prove the occurrence of Bose-Einstein condensation (BEC) in the thermodynamic limit at fixed nonvanishing density. The case of the Gross-Pitaevskii limit for $N$ bosons in a trapping potential of fixed extension has been treated in \cite{LSY,LS,LSSY}, where BEC was proven in the sense that the reduced density matrix has an eigenvalue of order $N$, associated to a condensate state.
In the FIR, BEC shows up in the behaviour of the `zero mode' of the field (the constant part of the field in a spatially homogeneous system). At small enough temperatures and high enough density, the most likely value for the zero mode is nonvanishing, and condensation takes place when it acquires a nonzero expectation value, which breaks the global $U(1)$ symmetry of the action spontaneously.
In the IBM representation and its subsequent large-deviation analysis, BEC is characterised as the occurrence of infinite cycles (as $N \to \infty$) in the permutations dominating the  above-mentioned symmetrization sum. 

Since both FIR and IBM represent the same mathematical objects (the standard traces for partition functions and expectation values over Fock space), it is clear that they are related. Indeed, this relation is essentially Symanzik's random-walk representation of quantum field theory, the rigorous version of which has led to breakthrough results in mathematical quantum field theory \cite{Aizenman,Froehlichphi4}. (A detailed exposition is in \cite{FFS}.) It is interesting to make this relation explicit for specific models, to see how far techniques can be combined. (For the systems considered in this paper, this is done here in Section \ref{RWsec}.) Another motivation for this is as follows.The IBM (and in the lattice setting, its discrete version, the interacting random walk (IRW)) corresponds to a one-particle kinetic energy given by the Laplacian, i.e.\ $\frac{p^2}{2m}$, where $p$ is the momentum of the particle.  In presence of a condensate, Bogoliubov \cite{Bog} predicted a spectrum of low-lying states above the ground state which is linear in $|p|$ at small $|p|$, namely $E_B(p) \approx c |p|$, where $c$ is the velocity of sound in the condensate: the excitations are sound waves. It is certainly an interesting question how this spectrum arises in a collection of Brownian motions when an infinite cycle forms. In the FIR, there is a short answer: once a condensate has formed, the fluctuation fields describing deviations from the condensate have a matrix-valued covariance, the eigenvalues of which exhibit the Bogoliubov spectrum (see, e.g.\ \cite{Benfatto}). It is thus natural, given that a relation between FIR and IRW representation can be derived without reference to a condensate, to look if one can obtain a stochastic representation also in presence of a condensate. First steps in this direction are described in Section \ref{physsec}. It is shown there that, at least in the time-continuum limit and for a local interaction, the kinetic term does generate a stochastic process, albeit one with long-range jumps (corresponding to $c|p|$). A main difference is in the interaction term, which differs from the case without condensate by terms causing branching and coalescence of random walks. 

The question about the relation of FIR and IRW also motivates taking a closer look at the FIR for the canonical ensemble of bosons, and this is the main focus of this paper. The methods developed here, and some results, also apply to the grand canonical ensemble. The latter has been analysed mathematically in a series of papers \cite{BFKT1,BFKT2,BFKT3,BFKT4}, in which essential foundations for an analysis of Bose condensation using methods of constructive quantum field theory, in particular mathematically rigorous Wilsonian renormalization group methods, have been laid, and the flow has been studied into the symmetry-breaking regime. Further papers in that series can be obtained from \cite{feldmanhomepage}. More recently, the FIR for boson systems and its relation to IBM representations have been studied in \cite{Froehlichetal1, Froehlichetal2}, with a focus on the grand-canonical ensemble. 

In this paper, I give a detailed discussion of the FIR for the canonical ensemble of identical bosons on a lattice, with a careful derivation from the operator formulation and a detailed discussion of the time continuum limit and the freedom of using different discrete-time actions that may be suitable for different aspects of the problem.  The main points and results are (similarly to \cite{Froehlichetal1, Froehlichetal2}) based on using auxiliary (Hubbard-Stratonovitch) fields, in terms of which the integral can be rewritten in terms of a positive Gaussian measure, with an interaction given by the inverse of a determinant (grand-canonical case) or a large permanent (canonical case) of a boson covariance in the background field given by the auxiliary field. The permanent implies the explicit symmetrization mentioned above. In this representation, there is no oscillatory part of the Gaussian integral. The auxiliary fields are ultralocal, i.e. a lattice-regularized version of white noise.  

Some details of the grand-canonical and the canonical integrals are actually quite different, both concerning the covariance in a background field and the factorization properties of inverse determinant vs. permanent  (see Section \ref{grandsec}). A main technical result shown here is that the fully regularized model has properties that allow for bounds that are uniform in the auxiliary fields, very much in the spirit of the loop vertex expansion \cite{MagRiv}. This allows for completely elementary and technically very simple proofs of the existence ot the time-continuum limit for a variety of different actions. They do not require any mutliscale analysis and in that sense, they are as simple as the corresponding ones in the fermionic case \cite{FKT,PSUV}. The uniformity of the bounds also allows to analyze the infinite-volume limit and the properties of correlation functions in that limit, in cases where renormalization is not required. This will be done in another paper \cite{toappear}. Although technically comparatively simple, the hypotheses and proofs contain a few subtle points (to be discussed below, see in particular Sections \ref{grandsec} and \ref{conclusec}), and it remains to be seen whether a similar uniformity also holds in genuinely multiscale situations (for a multiscale version of the loop vertex expansion, see \cite{multiscaleLVE}). 

Since one aim of this paper is to provide a technically simple approach, I have added appendices to make it essentially self-contained, given a knowledge of the operator algebraic formulation of quantum statistical mechanics (as exposed, e.g., in Section 5.2 of \cite{BR}).

\section{Identical Bosons on a Lattice}

This section contains the detailed setup of the class of models considered here, namely lattice models of identical bosons with a short-range density-density interaction, in the second-quantized formulation. 

\subsection{Kinematics} 

\subsubsection{Space}
To avoid technical issues in the setup, I start from a lattice system. Let $\speps > 0$,  $L$ be a 
very large integer multiple of $\speps$, and
\beq
\Xspace = \speps \bZ^d/L\bZ^d \; .
\eeq
$\Xspace$ is a lattice with spacing $\speps>0$ and periodic boundary conditions, i.e.\ a discrete torus. The limit $\speps \to 0$ at fixed $L$ (i.e. $L/\speps \to \infty$)  is usually called the continuum or ultraviolet limit, and the limit $L \to \infty$ is called the thermodynamic or infrared limit. In principle, one would like to take both limits, but this paper is mainly concerned with the thermodynamic limit at fixed $\speps > 0$. It is expected that in this limit (and with an appropriate choice of kinetic term and interaction) Bose-Einstein condensation happens at low enough temperatures, even in presence of an ultraviolet regulator. 
For maps $\xx \mapsto f_\xx$ and $\xx \mapsto g_\xx$ from $\Xspace $ to some $\bC$-algebra $\cA $, let
\beq\label{bilfgdef}
\bili{f}{g}_\Xspace
=
\speps^d \sum_{\xx \in \Xspace}
f_\xx \; g_\xx
\eeq
For $f$ and $g$ arising as restrictions of continuous functions on $\bR^d/L\bZ^d$, this is a Riemann sum approximation to an integral that arises in the limit $\speps \to 0$, which motivates the alternative notation $\bili{f}{g}_\Xspace = \int_{\xx \in \Xspace} f_\xx \; g_\xx=\int_\Xspace f\; g$, as well as the notation 
$\int_\xx F_\xx  = \eta^d \sum_{\xx \in \Xspace} F_\xx$. This notation will also be used at fixed $\speps > 0$, e.g.\ at $\speps  = 1$. 

\subsubsection{Hilbert spaces for bosons and the CCR algebra}

The quantum mechanical Hilbert space for a spinless particle on $\Xspace $ is $\cH = L^2 (\Xspace, \bC)$. In our setting it is simply the finite-dimensional space $\bC^\Xspace$ of all maps $f$ from $\Xspace$ to $\bC$, with inner product $\bracket{f}{g} = \bili{\bar f}{g}_\Xspace$. The $n$-boson space is $\FckoB{n} = \bS_n \bigotimes^n \cH$, where $\bS_n$ is the symmetrization operator, $(\bS_n f)(\xx_1, \ldots, \xx_n) = \frac{1}{n!} \sum_{\pi \in \Perm{n}} f(\xx_{\pi(1)}, \ldots, \xx_{\pi(n)})$.  Here $\Perm{n}$ denotes the set of all permutations on $\{ 1, \ldots , n\}$. The bosonic Fock space is defined as
\beq
\FockB
=
\bigoplus_{n=0}^\infty \FckoB{n}
\eeq
(where $\FckoB{0}=\bC$). 

For $N \in \bN_0$ let $\Proj_N$ denote the projector onto the $N$-particle subspace $\FckoB{N}$ of $\FockB$, $\Proj_{\le N } = \sum_{n=0}^N \Proj_N$, and $\FckoB{\le N} = \Proj_{\le N }\FockB$. 

\intremark{Let $n \in \bN_0$. For $f \in \cH$, the bosonic creation operator $\acre (f): \FckoB{n} \to \FckoB{n+1}$ is defined by $\acre (f) \bS_n (g_1 \otimes \ldots \otimes g_n) = \sqrt{n+1} \; \bS_{n+1} (f \otimes  g_1 \otimes \ldots \otimes g_n) $ and linearity. The bosonic annihilation operator $\aan (f): \FckoB{n+1} \to \FckoB{n}$ is the adjoint of $\acre(f)$. Define $\acre_\xx = \acre(e_\xx)$, for $e_\xx \in \cH$ given by $(e_\xx)_\xy = \eta^{-d/2} \delta_{\xx,\xy}$. }
Recall the standard $C^*$ algebra formulation of quantum many-particle systems given, e.g., in \cite{BR}. The algebra for boson systems is generated by bosonic annihilation operators $\aan$ and creation operators $\acre$ satisfying the canonical commutation relations (CCR) 
\beq
\aann_\xx \aann_\xy - \aann_\xy \aann_\xx = 0  
\quad
\mbox{ and }
\quad
\aann_\xx \acre_\xy - \acre_\xy \aann_\xx = \speps^{-d} \delta_{\xx,\xy} 
\quad
\mbox{ for all } \xx,\xy \in \Xspace \;.
\eeq
Let 
\beq\label{vacdef}
\Omega  = (1, 0, 0, \ldots ) \in \FockB
\eeq 
denote the vacuum vector. Then an orthonormal basis for $\FckoB{N}$ is given by the vectors 
\beq\label{FockONB}
e_\nu = \pli_{\xx \in \Xspace} (\nu_\xx!)^{-\frac12}\; {\acre_\xx}^{\nu_\xx} \Omega\; ,
\eeq
where the sequences $\nu = (\nu_\xx)_{\xx \in \Xspace} \in \bN_0^\Xspace$ satisfy the condition $\sum_{\xx \in \Xspace} \nu_\xx = N$. Thus the $N$-boson space $\FckoB{N}$ is finite-dimensional, with dimension
\beq
{\rm dim } \; \FckoB{N} 
=
(-1)^N \; {-|\Xspace| \choose N}
=
{|\Xspace| + N -1 \choose |\Xspace| -1} \; .
\eeq
$\FockB$ is infinite-dimensional even if $\Xspace$ is finite. 
\intremark{Using $\delta_{mn} = \frac{1}{2\pi\I} \int_{|z|=r} \frac{\dd z}{z} \; z^{m-n}$
(and choosing $r < 1$ to ensure convergence of the geometric series)
\beq
\begin{split}
{\rm dim } \; \FckoB{N}
&=
\sum_{\nu \in \bN_0^\Xspace} \delta_{\sum_{\xx \in \Xspace} \nu_\xx ,N}
\\
&=
\frac{1}{2\pi\I} \int_{|z|=r} \frac{\dd z}{z^{N+1}} \;  
(1-z)^{-|\Xspace|}
\\
&=
(-1)^N \; {-|\Xspace| \choose N}
\\
&=
{|\Xspace| + N -1 \choose |\Xspace| -1}
\end{split}
\neeq
}

Define the local density operators and the number operator $\opN$ by  
\beq\label{densidef}
\opn_\xx = \acre_\xx \aann_\xx
\qquad
\opN
=
\int_\xx \opn_\xx\; .
\eeq
Every $\opn_\xx$ maps the $N$-particle space $\FckoB{N}$ to itself, and on this space, it has operator norm $\norm{\opn_\xx} = N$. The restriction of $\opN$ to $\FckoB{N}$ is $N$ times the identity operator.  
The $\opn_\xx$ all commute:
\beq\label{densicom}
\opn_\xx \opn_\xy = \opn_\xy \opn_\xx
\quad
\mbox{ for all } \xx, \xy \in \Xspace \; .
\eeq

\subsection{Hamiltonians} 

For a single particle, the Hamiltonian is given by a self-adjoint matrix 
$\Kin$, acting as 
$(\Kin f)_\xx = \int_\xy \Kin_{\xx,\xy} f_\xy$, i.e.\ a hopping matrix. Typical choices will be the discrete Laplacian $-\Delta$, or $-\Delta + m^2$, with $m^2 > 0$,  or, for particles in an external potential $W$, $-\Delta + W(\xx)$. This operator will be assumed to be nonnegative. The particles interact by a two-body potential, i.e. $\WW_{\xx,\xy} \in \bR$ is the interaction energy contributed by a pair of particles at sites $\xx$ and $\xy$. The interaction is symmetric: $\WW_{\xx,\xy}= \WW_{\xy,\xx}$. Thus $\WW$ can be regarded as a self-adjoint operator on $\bC^\Xspace$. $\WW$ is called translation-invariant if $\WW_{\xx+\xz,\xy+\xz} = \WW_{\xx,\xy}$ for all $\xz$, and in this case I use the notation  $\WW_{\xx,\xy} = \WW (\xx-\xy)$. 
 
The Hamiltonian on the bosonic $N$--particle Hilbert space  $\FckoB{N}$ is
\beq\label{HamNdef}
\Ham_N 
=
\sum_{n=1}^N \Kin_n 
+
\frac12
\sum_{m,n=1}^N \WW_{\xx_m,\xx_n}
\eeq
where 
\beq
(\Kin_n f) (\xx_1, \ldots, \xx_N)
=
\sum_{n=1}^N 
\int_\xy  \Kin_{\xx_n,\xy}\; 
f(\xx_1, \ldots, \xx_{n-1}, \xy, \xx_{n+1}, \ldots, \xx_N)
\eeq
and where $\WW_{\xx_m,\xx_n}$ acts as a multiplication operator. 
In second-quantized form, it can be written as 
\begin{equation}\label{Hamdef}
\opH
=
\int_{\xx,\xy} \Kin_{\xx,\xy} \; \acre_\xx \aann_\xy
+
\frac12
\int_{\xx,\xy} \WW_{\xx,\xy}
\; 
\opn_\xx \opn_\xy \; .
\end{equation}
Inserting the notation (\ref{bilfgdef}), gives $\opH = \opH_0 + \opV$ with
\beq\label{H0Vdef}
\opH_0 = \bili{\acre}{\Kin\, \aan}_\Xspace
\quad\mbox{ and }\quad
\opV = \sfrac12\bili{\opn}{\WW\, \opn}_\Xspace\; 
\eeq
The case $\WW = 0$ describes free (i.e.\ noninteracting bosons).

\intremark{For all $\xx \in \Xspace: [\opV, \opn_\xx]=0$. Therefore, if $\lambda: \Xspace \to \bC$ and we define a new decomposition of $\opH_+$ by 
\beq
\begin{split}
\opH_+
&=
\opH_0 + \opV 
\\
&=
\opH_0 + \sum_\xx \lambda_\xx \opn_\xx
+
\opV -  \sum_\xx \lambda_\xx \opn_\xx
\\
&=
\tilde \opH_0 + \tilde \opV
\end{split}
\eeq 
then $[\tilde\opH_0,\tilde \opV] = [\tilde \opH_0, \opV] = [\opH_0,\opV]$. For all operators to be self-adjoint, the $\lambda_\xx$ must of course be real. 
}

\subsection{Ensembles}
In the following, the Hamiltonian $\opH$ is assumed to be nonnegative; in the above setup this is the case if both operators $\Kin \ge 0$ and $\WW \ge 0$. 

\subsubsection{Canonical ensemble}
For $\beta > 0$ define
\beq\label{uncan}
\unnex{\opA}^{(N,\beta,\Xspace)}
=
\Tr_{\FockB} \left[
\E^{-\beta \opH}\;
\opA\; \Proj_N
\right] \; ,
\eeq
$\canZ^{(N,\beta,\Xspace)} = [\opId]^{(N,\beta,\Xspace)}$, and 
\beq
\langle \opA\rangle^{(N,\beta,\Xspace)} 
= 
\frac{1}{\canZ^{(N,\beta,\Xspace)}}\, [\opA]^{(N,\beta,\Xspace)} \; .
\eeq 
Then $\canZ^{(N,\beta,\Xspace)}$ is the canonical partition function and $\langle A\rangle^{(N,\beta,\Xspace)} $ is the normalized canonical expectation value of $\opA$, for a system of $N$ bosons at inverse temperature $\beta$ on configuration space $\Xspace$. The range of the projection operator $\Proj_N$ is finite-dimensional, so $\Proj_N$ is trace class on $\FockB$, and the trace in (\ref{uncan}) exists whenever $\E^{-\beta \opH}\; \opA$ is a bounded operator. Because $\Proj_N^2 = \Proj_N$ and the trace is cyclic,
\beq
[A]^{(N,\beta,\Xspace)}
=
\Tr_{\FockB} \left[\Proj_N
\E^{-\beta \opH}\;
\opA\; \Proj_N
\right] \; .
\eeq
Because $\opH$ conserves particle number, 
\beq
[A]^{(N,\beta,\Xspace)}
=
\Tr_{\FockB} \left[\Proj_N
\E^{-\beta \opH}\;
\opA_N
\right] 
=
\Tr_{\FockB} \left[
\E^{-\beta \opH}\;
\opA_N
\right] 
\eeq
with
\beq\label{ANdef}
\opA_N = \Proj_N \; \opA\; \Proj_N \; .
\eeq
Thus it suffices to consider expectations of operators $\opF$ that map $\FckoB{N}$ to itself. 
It is convenient to introduce a generating function with insertion $\opF$ as
\beq\label{FcanZ}
\canZ^{(N,\beta,\Xspace)} (\opH, \opF) 
=
\Tr_{\FckoB{N}} \left[ \E^{-\beta \opH} \; (\opId + \opF)\right]
=
\Tr_{\FockB} \left[ \E^{-\beta \opH} \; (\opId + \opF) \; \Proj_N \right] \; .
\eeq
Normalized expectation values can then be obtained by differentiation as
\beq
\langle \opF \rangle^{(N,\beta,\Xspace)}_\opH
=
\frac{\partial}{\partial \lambda} \ln \canZ^{(N,\beta,\Xspace)} (\opH, \lambda \opF) \mid_{\lambda = 0}
\eeq
The dependence on $\opH$ has been made explicit in the notation. In general, for Hamiltonians preserving particle number, {\em normalized} canonical expectation values are unchanged under shifts in the energy, as follows. For any $\alpha \in \bC$ and $\psi \in \FockB$, 
\beq
\E^{\alpha \opN} \Proj_N \psi 
=
\E^{\alpha N} \Proj_N \psi
\eeq
As a consequence
\beq\label{fixy}
\canZ^{(N,\beta,\Xspace)} (\opH + \alpha \opN , \; \opF)
=
\E^{-\alpha\beta N} \; \canZ^{(N,\beta,\Xspace)} (\opH, \; \opF)
\eeq
The prefactor drops out in the quotient defining the normalized canonical expectation value, so
\beq\label{shifty}
\langle F \rangle^{(N,\beta,\Xspace)}_{\opH + \alpha \opN}
=
\langle F \rangle^{(N,\beta,\Xspace)}_{\opH}
\eeq

\subsubsection{Grand-canonical ensemble}
The grand canonical partition function at chemical potential $\mu$ is 
\beq\label{grancanZ}
\ganZ^{(\beta,\mu\Xspace)}
=
\sum_{N=0}^\infty \E^{\beta \mu N} \canZ^{(N,\beta,\Xspace)}
=
\Tr_{\FockB} \left[
\E^{-\beta (\opH-\mu\opN)}
\right] \; ,
\eeq
and the grand canonical expectation value is defined similarly. 
Convergence of the sum over $N$ in general requires conditions on $\mu$ and $\opH$ (e.g.\ $\mu < \inf {\rm spec}\;  \opH_0$ when $\opV =0$) which will be specified when appropriate. 

Returning to the shift invariance \Ref{shifty}, it is very important to note that this invariance does {\em not} hold for the partition functions. Instead, the prefactor in \Ref{fixy} implies that in the sum for the grand canoncial partition function, the parameter $\mu$ gets shifted to $\mu - \alpha$. Depending on $\alpha$, this shift may cause the sum over $N$ to diverge. Thus, when one attempts to obtain the usual relation between canonical and grand canonical ensembles, arbitrary shifts are no longer allowed. The convergence condition will show up prominently in the coherent-state integral derived in the next section.

\section{Main Results}\label{resultssec}
Here I state the main integral formulas and convergence theorems. 

By definition, $\Kin$ generates a stochastic process iff for all $\tau \ge 0$ and all $\xx$ and $\xy \in \Xspace$, the $\xx,\xy$ matrix element of $\E^{-\tau \Kin}$ is a nonnegative real number: 
\beq\label{stochcond}
\forall \tau \ge 0: \quad
(\E^{-\tau \Kin})_{\xx,\xy} \ge 0 \; .
\eeq

\begin{satz}\label{a4thm}
Let $\opH$ have kinetic term $\Kin$ and interaction $\WW$, as in (\ref{H0Vdef}). Let $\Kin$ and $\WW$ be translation invariant, $\Kin > 0$ and $\WW \ge 0$ as operators, and $\Kin$ generate a stochastic process. Then the canonical partition function has the integral representation 
\beq
\canZ^{(N,\beta,\Xspace)} 
=
\E^{\frac{\beta}{2} \WW(0) N} \; 
\E^{\frac{\beta}{2} |\Xspace| (\WW(0) - \hat \WW (0))}
\lim_{\ntau \to \infty}
\ZZ{\ntau}^{(N,\beta,\Xspace)}
\eeq
where 
\beq\label{a4integral}
\ZZ{\ntau}^{(N,\beta,\Xspace)}
=
\int \DD \aphi \; 
\E^{- \Action_\XTspace (\aphi)} \; 
\sfrac{1}{N!} 
{\bili{\bar\aphi(\beta)}{\aphi(0)}_\Xspace}^N \; .
\eeq
Here $\XTspace$ includes a discrete time lattice of spacing $\eptau = \frac{\beta}{\ntau}$, i.e.\ $\XTspace = \bT \times \Xspace$, where $\bT = \eptau \{0, \ldots , \ntau\}$\footnote{In the later sections, an index $\jtau \in \{0,\ldots, \ntau\}$, so that $\tau = \eptau \jtau$, will be used instead of $\tau$ to label fields and matrix elements of operators} The integral runs over complex fields $\aphi: \XTspace \to \bC$, $\DD \aphi$ denotes the volume form 
\beq
\DD \aphi 
=
\pli_{\tau \in \bT}
\pli_{\xx \in \Xspace}
\sfrac{\dd \bar\aphi(\tau,\xx) \wedge \aphi (\tau,\xx)}{2\pi\I}
\eeq
on $\bC^\XTspace$, and the action $\Action_\XTspace$ is
\beq\label{a4action}
\begin{split}
\Action_\XTspace (\aphi)
&=
\int_\tau \left[
\bili{\bar\aphi(\tau)}{[(- \del_\tau + \hat\WW(0) + \Kin^{(\eptau)} ) \aphi](\tau)}_\Xspace
+
\left(
|\aphi(\tau)|^2
\right|
\left.
\WW \; |\aphi(\tau)|^2
\right)_\Xspace
\right]
\\
&=
\int_\tau \left[
\int_\xx
\bar\aphi(\tau,\xx) \; 
[(- \del_\tau + \hat\WW(0) + \Kin^{(\eptau)} ) \aphi](\tau,\xx)
+
\int_{\xx,\xy}
|\aphi(\tau,\xx)|^2
\;
\WW (\xx-\xy)\; |\aphi(\tau,\xy)|^2
\right]
\end{split}
\eeq
The integral notation $\int_\tau$ is shorthand for $\eptau \sum_{\tau \in \bT}$, $\del_\tau$ denotes the discrete forward time derivative 
\beq
(\del_\tau \aphi)(\tau,\xx)
=
\frac{1}{\eptau} 
( \aphi(\tau+\eptau, \xx) - \aphi (\tau,\xx))
\eeq
with the boundary condition $\aphi (\beta + \eptau, \xx) =0$, and 
\beq\label{31}
(\Kin^{(\eptau)} \aphi) (\tau,\xx)
=
\frac{1}{\eptau} \int_\xy (1 - \E^{-\eptau \Kin}) (\xx-\xy)\; \aphi (\tau+\eptau,\xy) \; .
\eeq
The limit remains unchanged if $(\Kin^{(\eptau)} \aphi) (\tau,\xx)$ is replaced by $(\Kin \aphi) (\tau+\eptau,\xx)$ or by $ (\Kin \aphi) (\tau,\xx)$ in the action.
\end{satz}

Two important features of this representation are characteristic for the canonical ensemble: (1) the $N$-dependent polynomial factor ${\bili{\aphi(\beta)}{\aphi(0)}_\Xspace}^N$ appears in the integrand, and (2) the time derivative operator in the quadratic part of the action does {\em not} have a periodic boundary condition at $\tau = \beta$. (See also Theorem \ref{HSFIthm}, in particular (\ref{jbouc})). 
The `mass term' $\hat \WW (0)$ may be surprising at first sight. It arises in the proof in an integration by parts step which is necessary to show convergence. Its occurrence is also related to that of the prefactor $\E^{\frac{\beta}{2} \WW(0) N}$, which is necessary to fulfil certain positivity conditions in intermediate steps of the proof. Another variant of the interaction, where the $|a|^4$ term is replaced by a term of at most quadratic growth for $|a| \to \infty$, is discussed in Section \ref{integralsec}.

Theorem \ref{a4thm} is proven in Section \ref{Analysis2}, using the preparations made in Sections \ref{Analysis1}. Translation invariance is assumed here only to simplify the statement. In Section \ref{recoverasec}, a more general statement, Theorem \ref{a4thmj}, where translation invariance is not assumed, is proven. 

Theorem \ref{a4thm} applies in particular to a $\Kin$ of the form $\Kin = \Kin_0 - \mu \opId $, with $\Kin_0 > 0$ and $\mu \le  0$, and it implies that the limit of the integral (\ref{a4integral}) equals (up to the explicit prefactor) the canonical partition function  $\canZ^{(N,\beta,\Xspace)} $. By (\ref{fixy}), $\canZ^{(N,\beta,\Xspace)} $ is analytic in $\mu$ for all $\mu \in \bC$. If $\WW > 0$, the quartic term in (\ref{a4action}) is bounded below by $\delta  \int_{\tau,\xx} |\aphi(\tau,\xx)|^4$ for some $\delta > 0$, so the integral (\ref{a4integral}) converges absolutely for all $\mu \in \bC$ as well, and it defines an analytic function of $\mu$. Moreover, for $\mu_0 < 0$, the convergence as $\ntau \to \infty$ is uniform on compact subsets of $\{\mu \in \bC: \Re \mu < \mu_0\}$, hence the limiting function is analytic in $\mu$ there. Thus it coincides with $\canZ^{(N,\beta,\Xspace)} $ for all $\mu$ by the identity theorem, and hence the integral representation extends to all $\mu$. 

The action (\ref{a4action}) is such that it admits a rewriting by an integral over auxiliary fields. Indeed, the theorem is proven by using such a representation, which is the basis of all the analysis done here, hence more fundamental. It is given as follows.

\begin{satz}\label{h-a-integral}
Assume that $\Kin > 0$ and $\WW \ge 0$, and let $\opF$ be any operator on $\FckoB{N}$. Then 
$\canZ^{(N,\beta,\Xspace)} (\opH,\opF) = \lli_{\ntau\to \infty} \canZ^{(N,\beta,\Xspace,\ntau)}(\opH,\opF)$
where  
\beq
\canZ^{(N,\beta,\Xspace,\ntau)}(\opH,\opF)
=
\int \dd \mu_{\WWW} (h) 
\int_{\bC^{\XTspace}} \DD \aphi \; 
\E^{-\bili{\bar\aphi}{\Quad (h) \; \aphi}_\XTspace}\;
\bracket{\coh{\aphi (\beta)}}{(1+\opF) \; \Proj_N \; \coh{\aphi(0)}}
\eeq
The integration over $h$ runs over $\bR^\XTspace$, and $\mu_\WWW$ is the normalized, centered Gaussian measure with time-local covariance $\WW$, i.e.\
$\dd\mu_\WWW (h) = \pli_{\tau \in \bT} \dd\mu_\WW (h(\tau))$. 
The operator $\Quad (h)$, given in (\ref{Quaddef}) below, has a positive hermitian part, and the integral over $\aphi$ converges absolutely. 
The vectors $\coh{\aphi} \in \FockB$ are unnormalized coherent states, given in 
(\ref{cohstexp}). For $\opF = 0$, the matrix element is
\beq
\bracket{\coh{\aphi(\beta)}}{
\Proj_N \; \coh{\aphi(0)}}
=
\sfrac{1}{N!}\; 
{\bili{\bar \aphi (\beta)}{\aphi (0)}_\Xspace}^N
\eeq
The covariance $\Cov (h) = \Quad (h)^{-1}$ is entire analytic in $h$ on $\bC^\XTspace$. If $\Kin$ generates a stochastic process, then $\Cov(h)$ satisfies the uniform bound
\beq
\forall h \in \bR^\XTspace: \quad
\abs{\Cov(h) (\tau,\xx; \tau'\xx')}
\le
\Cov(0) (\tau,\xx; \tau'\xx') \; .
\eeq
\end{satz}

Theorem \ref{h-a-integral} follows from Theorem \ref{HSFIthm}, Lemma \ref{Quadlemma}, and (\ref{termNeum}). It will be used in further work to study properties of correlation functions. 

\begin{satz}\label{grancanthm1}
Let $\mu \in \bR$ and assume that $\Kin - \mu > 0$. Then the grand canonical partition function (\ref{grancanZ}) has the representation 
$\ganZ^{(\beta,\mu,\Xspace)} = \lim_{\ntau \to \infty} \ganZ^{(\beta,\mu,\Xspace,\ntau)} $ with
\beq
\ganZ^{(\beta,\mu,\Xspace,\ntau)}
=
\int \dd \mu_{\WWW} (h) 
\int_{\bC^{\XTspace}} \DD \aphi \; 
\E^{-\bili{\bar\aphi}{\Kuad (h) \; \aphi}_\XTspace}
\eeq
where $\Kuad (h)$ is given by $\Quad (h)$, but with $\Kin$ replaced by $\Kin-\mu$ and a periodic boundary condition in time (see (\ref{Kuaddef}) below).  
The grand-canonical covariance $\Gov (h) = \Kuad (h)^{-1}$ is analytic in $h$ on a neighbourhood of $\bR^{|\XTspace|}$.
It $\Kin$ generates a stochastic process, then $\Gov (h) $ satisfies the uniform bound 
\beq\label{grand-uniformed-tau}
\forall h \in \bR^\XTspace: \quad
\abs{\Gov (h) (\tau,\xx;\tau',\xx')}
\le
\Gov (0) (\tau,\xx;\tau',\xx') \; .
\eeq
$\Gov(0)$ is the time-ordered Green function for free bosons with kinetic term $\Kin - \mu$.
\end{satz}

Theorem \ref{grancanthm1} follows directly from Theorem \ref{grancanthm2}.

\section{Integral Formulas}\label{integralsec}

\subsection{Coherent states}

Here I summarize some well-known properties of coherent states, in the form that I shall use them in the following. For convenience of the reader, their (elementary) proofs are given in Appendix \ref{cohstateapp}.

Coherent states are eigenvectors of the annihilation operators: for any function $\aphi: \Xspace \to \bC$ there is a nonzero vector $\coh{\aphi} \in \FockB$ such that for all $\xx \in \Xspace$
\begin{equation}\label{coheig}
\aann_\xx \coh{\aphi} = \aphi_\xx \coh{\aphi} .
\end{equation}
The vector $\coh{\aphi}$ is explicitly given by   
\begin{equation}\label{cohstexp}
\coh{\aphi} = \E^{\bili{\aphi}{\acre}_\Xspace} \; \Omega
\end{equation}
where $\Omega$ is the vacuum vector in the bosonic Fock space (see (\ref{vacdef})). The right hand side of (\ref{cohstexp}) is defined by the series expansion of the exponential. Theseries is norm-convergent for all $\aphi$, so the coherent state $\coh{\aphi}$ is entire analytic in $\aphi$.

Coherent states are grand canonical objects, 
since $\aan$ maps $\FckoB{n+1}$ to $\FckoB{n}$. 
Their inner product on Fock space is 
\begin{equation}\label{cohinp}
\bracket{\coh{\aphi}}{\coh{\aphi'}}
=
\E^{\bili{\bar\aphi}{\aphi'}_\Xspace}.
\end{equation}
Thus $\coh{\aphi} \in \cF$ has norm 
\beq\label{cohnorm}
\norm{\coh{\aphi}}
=
\E^{\frac12 \bili{\bar\aphi}{\aphi}_\Xspace}\; .
\eeq
Denote the unit vector 
\beq\label{unicohdef}
\unicoh{\aphi} 
= 
\frac{1}{\norm{\coh{\aphi}}} \coh{\aphi}
=
\E^{- \frac12 \bili{\bar\aphi}{\aphi}_\Xspace} \; \coh{\aphi}\; ,
\eeq
the orthogonal projector on the subspace spanned by $\unicoh{\aphi}$ by $\ketbra{\unicoh{\aphi}}{\unicoh{\aphi}}$,
\begin{equation}
\dd^\Xspace \aphi
=
\pli_{\xx\in\Xspace}
\sfrac{\dd\bar \aphi_\xx \wedge \dd \aphi_\xx }{2\pi \I} \; ,
\end{equation} 
and $\bC_R = \{ z \in \bC: |z| \le R\}$. 
Then, in the sense of strong convergence,
\begin{equation}\label{resid}
1_{\FockB}
=
\lim\limits_{\Rad \to \infty}
\ili_{{\bC_\Rad}^\Xspace}
\dd^\Xspace\aphi \; 
\ketbra{\unicoh{\aphi}}{\unicoh{\aphi}}
\end{equation}
and for trace class operators $\opA$, 
\begin{equation}\label{spurli}
\Tr_{\FockB} \opA 
=
\lim\limits_{\Rad \to \infty}
\ili_{{\bC_\Rad}^\Xspace}
\dd^\Xspace\aphi \; 
\bracket{\unicoh{\aphi}}{\opA  \; \unicoh{\aphi}}
\end{equation}
(see  Theorem 2.26 of \cite{BFKT1}).

Moreover, if $\opH_0 = \bili{\acre}{\Kin \aann}_\Xspace$
with hermitian $\Kin \ge 0$ and if $\tau \ge 0$, then
\begin{equation}\label{expexp}
\left\langle \coh{\aphi} \mid
\E^{-\tau \opH_0}
\; \coh{\aphi'} \right\rangle
=
\E^{\bili{\bar \aphi}{\E^{-\tau\Kin} \, \aphi'}_\Xspace} \; .
\end{equation}

\subsection{The integral for the canonical ensemble}

By (\ref{spurli}), and with the notation (\ref{nota1}),
\beq
\canZ^{(N,\beta,\Xspace)} (\opH,\opF)
=
\lim\limits_{\Rad \to \infty}
\ili_{\bC_\Rad^\Xspace}
\dd^\Xspace \aphi\;
\bracket{\unicoh{\aphi}}{\E^{-\beta \opH} \; (\opId + \opF) \Proj_N \unicoh{\aphi}} \; .
\eeq
By (\ref{resid}), $\canZ^{(N,\beta,\Xspace)} (\opH, \opF) = \lim\limits_{\Rad \to \infty} {}_\Rad\canZ^{(N,\beta,\Xspace)} (\opH,\opF)$ where 
\beq
{}_\Rad\canZ^{(N,\beta,\Xspace)} (\opH,\opF)
=
\ili_{\bC_\Rad^\Xspace}
\dd^\Xspace \aphi\; 
\ili_{\bC_\Rad^\Xspace}
\dd^\Xspace \aphi'\; 
\bracket{\unicoh{\aphi}}{\E^{-\beta \opH} \; (\opId + \opF) \; \unicoh{\aphi'}}\;
\bracket{\unicoh{\aphi'}}{\Proj_N \unicoh{\aphi}} \; .
\eeq
A priori, the application of (\ref{resid}) results in two limits $\Rad \to \infty$ and $\Rad' \to \infty$, but since both exist, the limit can be taken keeping $\Rad = \Rad'$. This will again be done later, when there are many insertions of type (\ref{resid}).

To streamline notation further, I will write 
\beq\label{nota1}
\int_{\aphi_1, \ldots, \aphi_n}^\Rad \ldots 
=
\ili_{{\bC_\Rad}^\Xspace}
\dd^\Xspace\aphi_1 \; 
\ldots
\ili_{{\bC_\Rad}^\Xspace}
\dd^\Xspace\aphi_n \;
\ldots 
\eeq
and
\beq\label{nota2}
\int_{\aphi_1, \ldots, \aphi_n} \ldots 
=
\ili_{{\bC}^\Xspace}
\dd^\Xspace\aphi_1 \; 
\ldots
\ili_{{\bC}^\Xspace}
\dd^\Xspace\aphi_n \;
\ldots \eeq

By (\ref{unicohdef}), and with the notation (\ref{nota1})
\beq
{}_\Rad\canZ^{(N,\beta,\Xspace)} (\opF)
=
\int_{\aphi,\aphi'}^\Rad
 \E^{-\bili{\bar\aphi}{\aphi}_\Xspace - \bili{\bar\aphi'}{\aphi'}_\Xspace}
\bracket{\coh{\aphi}}{\E^{-\beta \opH} \; (\opId + \opF) \; \coh{\aphi'}}\;
\bracket{\coh{\aphi'}}{\Proj_N \coh{\aphi}} \; .
\eeq

\subsubsection{Matrix element of the $N$-particle projection}

\bigskip\begin{lemma}\label{sandprichlemma}
Let $N \in \bN_0$,  $\Proj_N$ be the orthogonal projection from $\FockB$ onto $\FckoB{N}$. Then, for all $\aphi,\aphi' \in \bC^\Xspace$, 
\beq\label{sandprich}
\bracket{\coh{\aphi'}}{\Proj_N \coh{\aphi}}
=
\sfrac{1}{N!} \, \left(\bar\aphi',\aphi\right)_\Xspace^N \, .
\eeq
\end{lemma}

\begin{proof} Let $N \in \bN_0$. Because the expansion (\ref{cohstexp}) for the  coherent state $\coh{\aphi}$ is norm convergent for any $\aphi$ and the projection $\Proj_N$ is continuous, 
\begin{equation}
\Proj_N \coh{\aphi}
=
\sfrac{1}{N!}\, {\bili{\aphi}{\acre}_\Xspace}^N \Omega \; .
\end{equation}
Because $\acre$ is the adjoint of $\aan$ and because $\coh{\aphi'}$ is an eigenvector of $\aan$,  
\beq
\begin{split}
\bracket{\coh{\aphi'}}{\Proj_N \coh{\aphi}}
&=
\sfrac{1}{N!}\;
\bracket{\coh{\aphi'}}{{\bili{\aphi}{\acre}_\Xspace}^N \Omega}
=
\sfrac{1}{N!}\;
\bracket{{\bili{\bar\aphi}{\aan}_\Xspace}^N\coh{\aphi'}}{ \Omega}
\\
&=
\sfrac{1}{N!}\;
\bracket{{\bili{\bar\aphi}{\aphi'}_\Xspace}^N\coh{\aphi'}}{ \Omega}
=
\sfrac{1}{N!}\;
{\bili{\bar\aphi'}{\aphi}_\Xspace}^N \; \bracket{\coh{\aphi'}}{\Omega} \; .
\end{split}
\eeq
The inner product $\bracket{\coh{\aphi'}}{\Omega} = 1$. 
\end{proof}

\bigskip\begin{lemma}
If $\E^{-\beta\opH}$ is a bounded operator on $\FockB$, 
the canonical generating function $\canZ^{(N,\beta,\Xspace)} (\opH,\opF)$, defined in $(\ref{FcanZ})$, is the limit $\Rad \to \infty$ of the integral 
\beq\label{RadZc}
{}_\Rad\canZ^{(N,\beta,\Xspace)} (\opH,\opF)
=
\int_{\aphi,\aphi'}^\Rad
 \E^{-\bili{\bar\aphi}{\aphi}_\Xspace - \bili{\bar\aphi'}{\aphi'}_\Xspace}\;
\bracket{\coh{\aphi}}{\E^{-\beta \opH} \; \coh{\aphi'}}\;
\frac{\left(\bar\aphi',\aphi\right)_\Xspace^N}{N!}
\eeq
and the limit $\Rad\to\infty$ is given by an absolutely convergent integral.
\end{lemma}

\begin{proof}
The integral representation follows by inserting (\ref{sandprich}) into (\ref{FcanZ}). If the norm of $\E^{-\beta\opH}$ is bounded by $M >0$, then 
\beq
\abs{
\bracket{\coh{\aphi}}{\E^{-\beta \opH} \; \coh{\aphi'}}
}
\le
M\; \norm{\coh{\aphi}} \; \norm{\coh{\aphi'}}
\eeq
By (\ref{cohnorm}), the absolute value of the integrand is bounded by 
\beq
M\; 
\E^{-\frac12\bili{\bar\aphi}{\aphi}_\Xspace - \frac12\bili{\bar\aphi'}{\aphi'}_\Xspace}\;
\sfrac{\abs{\left(\bar\aphi',\aphi\right)_\Xspace}^N}{N!}
\eeq
which is integrable in $\aphi$ and $\aphi'$.
\end{proof}

Note that if $\opH$ commutes with $\opN$, then 
\beq
\tr \left[\E^{-\beta \opH} \Proj_N\right]
= 
\tr \left[\Proj_N \E^{-\beta \opH} \Proj_N\right]
=
\tr \left[\E^{-\beta \Proj_N\opH\Proj_N} \Proj_N\right]
\eeq
so for $\opH$ that is unbounded below, a formula similar to (\ref{RadZc}), with $\opH$ replaced by $\opH_N = \Proj_N \opH \Proj_N$, still holds.

By the residue formula 
\beq
\frac{\alpha^N}{N!} 
=
\frac{1}{2\pi\I} \int_{|z| = r} \frac{\dd z}{z^{N+1}} \E^{\alpha z}
\eeq
($r>0$), the canonical partition function then becomes
\beq\label{canoz}
{}_\Rad\canZ^{(N,\beta,\Xspace)}
=
\frac{1}{2\pi\I} \int_{|z| = r} \frac{\dd z}{z^{N+1}} 
\int_{\aphi,\aphi'}^\Rad
 \E^{-\bili{\bar\aphi}{\aphi}_\Xspace - \bili{\bar\aphi'}{\aphi'}_\Xspace + z \bili{\bar\aphi'}{\aphi}_\Xspace}\;
\bracket{\coh{\aphi}}{\E^{-\beta \opH} \; \coh{\aphi'}}\;
\eeq
It is this form which allows to connect conveniently to the other ensembles in the limit of large $N$ by a stationary-phase argument.

\subsubsection{Simple cases for $\opF$}
Let $\opF = \bili{\acre}{F\aan}_\Xspace = \int_{\xx,\xy} F_{\xx,\xy} \; \acre_\xx \aann_\xy$. Then 
\beq
\langle \acre_\xx \aann_\xy \rangle^{(N)}
=
\frac{\partial}{\partial F_{\xx,\xy}} \ln Z_{c}^{(N)} (\opF) \mid_{F=0} \; .
\eeq
A little calculation gives 
\beq\label{fandprich}
\opF {\bili{\aphi}{\acre}_\Xspace}^N \; \Omega
=
N \; \bili{\aan}{F^{\tt T}\acre}_\Xspace \; {\bili{\aphi}{\acre}_\Xspace}^{N-1} \Omega\; ,
\eeq
hence
\beq\label{fandprick}
\bracket{\coh{\aphi'}}{(\opId + \opF) \Proj_N \coh{\aphi}}
=
\frac{1}{(N-1)!} \; {\bili{\bar\aphi'}{\aphi}_\Xspace}^{N-1} \; 
\bili{\bar\aphi'}{ (\sfrac{1}{N} + F) \; \aphi}_\Xspace
\eeq
This is easy to interpret: when $\aann_\xy$ is applied to the $N$-particle state
$\Proj_N \coh{\aphi}$, the result is an $N-1$-particle state, so the power decreases to $N-1$ in that term. 

It is straightforward to generalize this to operators of the form 
\footnote{if $\opH$ conserves particle number, this only contributes to the canonical expectation value when $m=n$}
\beq
\opF = \int_{\xx_1, \ldots, \xx_m} \int_{\xy_1, \ldots , \xy_n} F_{\xx_1, \ldots, \xx_m;\xy_1, \ldots , \xy_n} \pli_{i=1}^m \acre_{\xx_i} \pli_{j=1}^n \aann_{\xy_j}
\neeq

\subsection{Noninteracting bosons}
It is instructive to see how this leads to the standard formulas for the canonical partition function for free bosons. 
In the absence of interaction, $\opH = \opH_0$, so (\ref{expexp}) can be used to rewrite the canonical partition function as
\beq\label{contourli}
\begin{split}
{}_\Rad\canZ^{(N,\beta,\Xspace)}
&=
\frac{1}{2\pi\I} \int_{|z| = r} \frac{\dd z}{z^{N+1}} 
\int_{\aphi,\aphi'}^\Rad
\E^{-\bili{\bar\aphi}{\aphi}_\Xspace - \bili{\bar\aphi'}{\aphi'}_\Xspace + z \bili{\bar\aphi'}{\aphi}_\Xspace + \bili{\bar\aphi}{\E^{-\beta\Kin}\; \aphi'}_\Xspace}
\\
&=
\frac{1}{2\pi\I} \int_{|z| = r} \frac{\dd z}{z^{N+1}} 
\int_{\aphi,\aphi'}^\Rad
\exp \left(
- \left([\bar\aphi,\bar \aphi'] \;\Big\vert\;
\left[
\begin{array}{cc}
\opId & - \E^{-\beta \Kin} \\
- z\opId & \opId
\end{array}
\right]
\left[
\begin{array}{c}
\aphi \\
\aphi'
\end{array}
\right]
\right)_\Xspace
\right)
\end{split}
\eeq
Let $r<1$, then the matrix in the quadratic form in the exponent has a positive definite hermitian part, and therefore the Gaussian integral is absolutely convergent, hence the limit $\Rad \to \infty$ can be taken. The Gaussian integral over $\aphi$ and $\aphi'$ then gives the inverse of its determinant, so that
\beq
\begin{split}
{}_\Rad\canZ^{(N,\beta,\Xspace)}
&=
\frac{1}{2\pi\I} \int_{|z| = r} \frac{\dd z}{z^{N+1}} \;
\det
\left[
\begin{array}{cc}
\opId & - \E^{-\beta \Kin} \\
- z\opId & \opId
\end{array}
\right]^{-1}
\\
&=
\frac{1}{2\pi\I} \int_{|z| = r} \frac{\dd z}{z^{N+1}} \;
\det
\left[
\opId - z  \E^{-\beta \Kin}
\right]^{-1}
\end{split}
\eeq
Evaluated in an eigenbasis of $\Kin$, this determinant factorizes, and the (convergent) expansion in $z$ then gives the sum
\beq
\canZ^{(N,\beta,\Xspace)}
=
\sum_{n \in \bN_0^A}
\E^{-\beta \sum_\al n_\al E_\al}\;
\delta_{\sum_{\alpha \in A} n_\alpha \;,\; N}
\eeq
which follows straightforwardly from the definition of the canonical ensemble as a trace. Here the eigenbasis is indexed by the set $A$, i.e.\ $(E_\alpha)_{\alpha \in A}$ are the eigenvalues of $\Kin $.

Without the contour integral, the canonical partition function can be written as a permanent: 
\beq
\begin{split}
{}_\Rad\canZ^{(N,\beta,\Xspace)}
&= 
\int_{\aphi,\aphi'}^\Rad
\E^{-\bili{\bar\aphi}{\aphi}_\Xspace - \bili{\bar\aphi'}{\aphi'}_\Xspace + \bili{\bar\aphi}{\E^{-\beta\Kin}\; \aphi'}_\Xspace}\;
\sfrac{\left(\bar\aphi',\aphi\right)_\Xspace^N}{N!}
\\
&=
\int_{\aphi,\aphi'}^\Rad
\E^{
- \left([\bar\aphi,\bar \aphi'] \;\Big\vert\;
\left[
\begin{array}{cc}
\opId & - \E^{-\beta \Kin} \\
0 & \opId
\end{array}
\right]
\left[
\begin{array}{c}
\aphi \\
\aphi'
\end{array}
\right]
\right)_\Xspace
} \;
\frac{\left(\bar\aphi',\aphi\right)_\Xspace^N}{N!}
\end{split}
\eeq
Again, the limit $\Rad\to\infty$ can be taken since the real part of the quadratic form is strictly positive. Writing out each factor $\left(\bar\aphi',\aphi\right)_\Xspace =\int_\xx \bar\aphi'_\xx\; \aphi_\xx$ and doing the Gaussian integral (noting that this time, the determinant is $1$ and the inverse obviously $\left[
\begin{array}{cc}
\opId & \E^{-\beta \Kin} \\
0 & \opId
\end{array}
\right] $), this gives 
\beq\label{perm1}
\begin{split}
\canZ^{(N,\beta,\Xspace)}
&= 
\frac{1}{N!}\;
\sum_{\pi \in \Perm{N}}
\int_{\xx_1,\ldots, \xx_N}
\pli_{n=1}^N
\left(
\E^{-\beta\Kin}
\right)_{\xx_n,\xx_{\pi(n)}}
\\
&=
\int_{\xx_1,\ldots, \xx_N}
\sfrac{1}{N!}\;
\mbox{Perm } C (\xx_1, \ldots, \xx_N)
\end{split}
\eeq
where the $N\times N$ matrix $C(\xx_1, \ldots, \xx_N)$ has matrix elements $C(\xx_1, \ldots, \xx_N)_{ij}= \left(\E^{-\beta\Kin}\right)_{\xx_i,\xx_j}$, and Perm denotes its permanent. See Appendix \ref{detpermlapapp}

\subsection{Interacting bosons}

\subsubsection{Time slicing}
Recall that $[\opH,\opN]=0$. Then the traces defining the canonical partition function and expectation values are traces over the finite-dimensional $N$-boson space $\FckoB{N}$, on which $\opH_0$ and $\opV$ are bounded, so the Lie product formula applies, and 
\beq\label{Lieprod}
\tilZ^{(N,\beta,\Xspace)} (\opF) 
=
\lim_{\ntau\to\infty} \tilZ^{(N,\beta,\Xspace,\ntau)} (\opF) \;,
\eeq 
where
\beq
\tilZ^{(N,\beta,\Xspace,\ntau)} (\opF) 
=
\tr 
\left[
\left(
\E^{-\eptau \opH_0} \; \E^{-\eptau \opV}
\right)^\ntau\;
(1+\opF) \; \Proj_N
\right]
\eeq
with 
\beq\label{eptaudef}
\eptau = \frac{\beta}{\ntau}\;.
\eeq 
Again by (\ref{spurli}) and (\ref{resid}), and with the new notation $\aphi_0=\aphi$ and $\aphi_\ntau=\aphi' $, this becomes
\beq
\tilZ^{(N,\beta,\Xspace,\ntau)} (\opF) 
=
\lim_{\Rad \to \infty}
\int_{\aphi_0,\aphi_\ntau}^{\Rad}
\bracket{\unicoh{\aphi_0}}{%
\left(
\E^{-\eptau \opH_0} \; \E^{-\eptau \opV}
\right)^\ntau\; \unicoh{a_\ntau}}\;
\bracket{\unicoh{a_\ntau}}{%
(1+\opF) \; \Proj_N \; \unicoh{\aphi_0}
}
\eeq
By (\ref{resid}), 
\beq\label{schleiss1}
\bracket{\unicoh{\aphi_0}}{%
\left(
\E^{-\eptau \opH_0} \; \E^{-\eptau \opV}
\right)^\ntau\; \unicoh{a_\ntau}}
=
\lli_{\Rad' \to \infty}
\int_{\aphi_1, \ldots, \aphi_{\ntau -1}}^{\Rad'}
\pli_{\jtau=1}^\ntau
\bracket{\unicoh{\aphi_{\jtau-1}}}{\E^{-\eptau \opH_0} \; \E^{-\eptau \opV}\; \unicoh{a_\jtau}}
\eeq
As discussed above, the limit may be taken with $\Rad'=\Rad$.
By (\ref{unicohdef}) this can be rewritten in terms of the unnormalized coherent states as
\beq
\begin{split}
\tilZ^{(N,\beta,\Xspace,\ntau)} (\opF) 
&=
\lim_{\Rad \to \infty}
\ili_{\aphi_0,\aphi_1, \ldots , \aphi_\ntau}^{\Rad}
\E^{-\sum_{\jtau=0}^\ntau \bili{\bar\aphi_\jtau}{\aphi_\jtau}_\Xspace}\; 
\bracket{\coh{a_\ntau}}{%
(1+\opF) \; \Proj_N \; \coh{\aphi_0}}
\\
&\mkern120mu
\pli_{\jtau=1}^\ntau
\bracket{\coh{\aphi_{\jtau-1}}}{\E^{-\eptau \opH_0} \; \E^{-\eptau \opV}\; \coh{a_\jtau}}
\;
\end{split}
\eeq

\subsubsection{Auxiliary fields}
As noted before, the interaction $\WW$ is assumed to be a nonnegative operator, and the $\opn_\xx$ generate an abelian Banach algebra, so the exponential of the interaction can be written as a Gaussian integral
\beq
\E^{-\eptau \opV}
=
\E^{-\eptau \bili{\opn}{\WW \opn}_\Xspace}
=
\int
\dd\mu_\WW (h) \;
\E^{\I \sqrt{\eptau} \bili{h}{\opn}_\Xspace}
\eeq
where $\dd\mu_\WW$ denotes the normalized Gaussian measure on $\bR^\Xspace$ with mean zero and covariance $\WW$. This identity is used for every $\jtau$ in the product  (\ref{schleiss1}).  The Gaussian integral is absolutely convergent, hence can be taken out of the inner product and the integral over the $a$ fields. This, and (\ref{expexp}), then gives
\beq
\begin{split}\label{66}
\bracket{\coh{\aphi_{\jtau-1}}}{\E^{-\eptau \opH_0} \; \E^{-\eptau \opV}\; \coh{a_\jtau}}
&=
\int
\dd\mu_\WW (h_\jtau)\;
\bracket{\coh{\aphi_{\jtau-1}}}{\E^{-\eptau \opH_0} \; 
\E^{\I \sqrt{\eptau} \bili{h_\jtau}{\opn}_\Xspace}\; \coh{a_\jtau}}
\\
&=
\int
\dd\mu_\WW (h_\jtau)\;
\E^{%
\bili{{\aphi_{\jtau-1}}}{\;\E^{-\eptau \Kin} \; 
\E^{\I \sqrt{\eptau} \; h_\jtau}\; {a_\jtau}}_\Xspace
}
\end{split}
\eeq
The collection of additional integration variables $h_{\jtau,\xx}$ indexed by $(\jtau,\xx)\in \{ 1, \ldots, \ntau\} \times \Xspace$, is called the {\em auxiliary field} in the following. In (\ref{66}) and similar expressions, $\E^{\I \sqrt{\eptau} h_\jtau}$ denotes the multiplication operator, $(\E^{\I \sqrt{\eptau} h_\jtau} \, \aphi)_\xx = \E^{\I \sqrt{\eptau} h_{\jtau,\xx}} \aphi_\xx $.

With the notations
\beq\label{DDadef}
\DD a 
=
\pli_{\jtau=0}^\ntau 
\;
\dd^\Xspace \aphi_\jtau
=
\pli_{\jtau=0}^\ntau \;
\pli_{\xx \in \Xspace}
\sfrac{\dd \bar \aphi_{\jtau,\xx} \wedge \aphi_{\jtau,\xx}}{2\pi\I}
\eeq
and 
\beq
\dd \mu_{\WWW} (h) \; 
=
\pli_{\jtau=1}^\ntau 
\dd\mu_{\WW} (h_\jtau)
\eeq
the following statement holds.

\begin{satz}\label{HSFIthm}
Assume that $\WW \ge 0$. Set $h_0 = 0$. For $\jtau, \jtau' \in \{ 0, \ldots , \ntau\}$ let
\beq\label{Quaddef}
\Quad (h) _{\jtau,\jtau'}
=
\delta_{\jtau,\jtau'} \; \opId
-
\delta_{\jtau+1,\jtau'} \; 
\E^{-\eptau \Kin} \; 
\E^{\I \sqrt{\eptau} \; h_{\jtau'}} 
\eeq
so that 
\beq\label{jbouc}
\begin{split}
\bili{\bar\aphi}{\Quad (h) \; \aphi}_\XTspace
&=
\sum_{\jtau,\jtau'=0}^\ntau
\bili{\bar\aphi_\jtau}{\Quad(h)_{\jtau,\jtau'} \aphi_{\jtau'}}_\Xspace
\\
&=
\sum_{\jtau=0}^\ntau \bili{\bar\aphi_\jtau}{\aphi_\jtau}_\Xspace
-
\sum_{\jtau=0}^{\ntau-1} \bili{\bar\aphi_\jtau}{\E^{-\eptau \Kin} \; 
\E^{\I \sqrt{\eptau} \; h_{\jtau+1}} \aphi_{\jtau+1}}_\Xspace\; .
\end{split}
\eeq
Then
\beq\label{ZcN}
\tilZ^{(N,\beta,\Xspace,\ntau)} (\opF)
=
\lli_{\Rad\to\infty}
\int \dd \mu_{\WWW} (h) 
\int^\Rad \DD \aphi \; 
\E^{-\bili{\bar\aphi}{\Quad (h) \; \aphi}_\XTspace}\;
\bracket{\coh{\aphi_\ntau}}{(1+\opF) \; \Proj_N \; \coh{\aphi_0}}
\eeq
If $\Kin > 0$, then the quadratic form $\Re \bili{\bar\aphi}{\Quad (h) \; \aphi}_\XTspace$ is strictly positive, so the Gaussian integral over $\aphi$ converges absolutely and the limit $\Rad \to \infty$ can be taken by dropping the restriction on the integral domain. In particular, for $\opF =0$,
\beq\label{HSFI}
\begin{split}
\tilZ^{(N,\beta,\Xspace,\ntau)} 
&=
\int \dd \mu_{\WWW} (h) 
\int_{\bC^{\XTspace}} \DD \aphi \; 
\E^{-\bili{\bar\aphi}{\Quad (h) \; \aphi}_\XTspace}\;
\sfrac{{\bili{\bar \aphi_{\ntau}}{\aphi_0}_\Xspace}^N}{N!}
\end{split}
\eeq
\end{satz}

\begin{proof}
The explicit formula (\ref{ZcN}) has already been derived; recall also (\ref{sandprich}) and (\ref{fandprich}). The positivity of the real part is equivalent to operator positivity of $\Quad_S(h)
=
\frac12(\Quad (h)+ \Quad(h)^\dagger)$, and it holds because for $\Kin > 0$, the diagonal part of $\Quad_S (h)$ strictly dominates the nondiagonal parts (note that $
\norm{\E^{-\eptau \Kin}\; \E^{\I \sqrt{\eptau} \; h_\jtau}}
=
\norm{\E^{-\eptau \Kin}}
< 1$). For $\Kin =0 $ and $h = 0$, $\Quad_S (0) = -\Delta$, the Laplacian on $\{0, \ldots, \ntau\}$, which is nonnegative. 
In general, set $K=\E^{-\eptau \Kin}$ and $U_\jtau = \E^{\I \sqrt{\eptau} h_\jtau}$. Then $K$ is a positive operator because $\Kin $ is hermitian, and $U_\jtau$ is unitary for all $\jtau$ because $h_\jtau$ is real-valued. Consequently, 
\beq
\KIP{\aphi}{\aphi'} =
\bili{\aphi}{K \; \aphi'} 
\eeq
is a positive definite scalar product on $\bC^\Xspace$, and 
\beq
\begin{split}
\Re \bili{\bar \aphi }{\Quad (h) \; \aphi}_\XTspace
&=
\sum_{\jtau=0}^\ntau
\bili{\bar\aphi_\jtau}{(1-K)\; \aphi_\jtau}_\Xspace
+
\sfrac12 \KIP{\aphi_0}{\aphi_0} 
+
\sfrac12 \KIP{\aphi_\ntau}{\aphi_\ntau} 
\\
&+
\sfrac12 
\sum_{\jtau=1}^\ntau
\KIP{\aphi_{\jtau -1} - U_\jtau \aphi_\jtau}{\aphi_{\jtau -1} - U_\jtau \aphi_\jtau} 
\\
&\ge 
\sum_{\jtau=0}^\ntau
\bili{\bar\aphi_\jtau}{(1-K)\; \aphi_\jtau}_\Xspace
\end{split}
\eeq
Because $0 < K \le 1$ if $\Kin \ge 0$, and $0< K < 1$ if $\Kin > 0$, the lower bound is a nonnegative quadratic form if $\Kin \ge 0$ and strictly positive if $\Kin > 0$. 
Thus the limit $\Rad \to \infty$ can be taken and gives (\ref{HSFI}). 
\end{proof}

\subsection{The covariance}
The operator $\Quad(h)$, given by (\ref{Quaddef}), is a matrix $\Quad(h)\in M_{\ntau+1}(\cB(\FockB))$ that is upper triangular: setting $\opA_\jtau  =  \E^{-\eptau \Kin} \; \E^{\I \sqrt{\eptau} h_\jtau }$
\beq
\Quad (h) 
=
\left[
\begin{array}{cccccc}
\opId & - \opA_1 & 0 & 0 & \ldots & 0 \\
0 & \opId & - \opA_2 & 0 & \ldots & 0\\[3ex]
 &  & \ddots & & \ddots &  \\[3ex]
0 & 0 & 0 & & \opId & -A_\ntau \\
0 & 0 & 0 & \ldots & 0 & \opId
\end{array}
\right]
\eeq
Therefore, $\det \Quad(h) = 1$, and the inverse of $\Quad{h}$, the covariance
\beq
\Cov (h) = \Quad (h)^{-1}
\eeq
exists and is also upper triangular, and given by a terminating Neumann series as
\beq\label{termNeum}
\left( \Cov (h) \right)_{\jtau,\jtau'}
=
\begin{cases}
0 & \jtau > \jtau' \\
\opId & \jtau = \jtau' \\
\pli_{k=\jtau+1}^{\jtau'} \opA_k & \jtau < \jtau'  
\end{cases}
\eeq
(see Appendix \ref{Covarapp}). Because the inverse is given by a finite sum, $\Cov(h)$ is entire analytic in $h$. In fact, for every finite $\ntau$, it is a trigonometric polynomial in $h$.

\subsection{Remarks on the structures arising from \Ref{HSFI}}
Eq.\ \Ref{HSFI} is useful in various respects. It represents the partition function as a Gaussian integral both in the original boson field $\aphi$ and in the auxiliary field $h$. Integrating over the auxiliary field will give a regularized version of the formal functional integral for boson systems; see Subsection \ref{aintsssec}.
Instead integrating over the $\aphi$-field gives a representation in terms of the auxiliary field only. This will be used to derive a random-walk representation which is similar to that of \cite{AdamsKoenig,AdamsBruKoenig,AdamsCollevecchioKoenig}  (see Section \ref{RWsec}) but, being an integral with respect to a real Gaussian measure it also provides a starting point for a field-theoretical analysis by convergent  expansions.

\subsubsection{Stochastic field equation}
One can also derive a complex stochastic equation for $\aphi$ 
by integration by parts in the integral. Let $\langle F(\aphi) \rangle$
denote the normalized canonical expectation value, then using
\begin{equation}
\left(\Quad(h) \aphi\right)_k (x)
=
- \sfrac{\del}{\del \bar\aphi_k (x)}
\;
\E^{- \langle \bar\aphi, \Quad(h) \aphi \rangle}
\end{equation}
and integrating by parts
gives for $k < \ntau$ the `Schwinger-Dyson' equation $\langle\Quad (h) \aphi \rangle=0$, 
which, in a formal continuum limit corresponds to the Gaussian white noise average over $h$ of the equation
\begin{equation}
\left(
\del_\tau - \Delta + \E^{\I \sqrt{\eptau} h(\tau,x)} - 1
\right)
\aphi(\tau,x) 
= 0 .
\end{equation}

\subsubsection{The complex integral obtained after integrating out the auxiliary field}
\label{aintsssec}
For simplicity I consider here the case of an on-site interaction, i.e.\ $\WW_{\xx,\yy} = \vvv \speps^{-d} \delta _{\xx,\yy}$, with $\vvv \in \bR$, $\vvv > 0$. More general $\WW$'s can be treated as well in terms of expansions, but the essential points are visible already in this case. 
With this choice of $\WW$, the $h$-integral factorizes both over the time slice index $\jtau$ and over the space index $\xx$, as follows. The $h$-dependent part of the quadratic form $\bili{\bar\aphi}{\Quad(h) \aphi}_\XTspace$ is
\beq
\sum_{\jtau=1}^\ntau
\bili{\tilde\aphi_{\jtau-1}}{\E^{\I \sqrt{\eptau} h_\jtau} \; \aphi_\jtau}_\Xspace
=
\sum_{\jtau=1}^\ntau
\int_\xx
\E^{\I \sqrt{\eptau} h_{\jtau,\xx}}
\tilde\aphi_{\jtau-1,\xx} \; \aphi_{\jtau,\xx}
\eeq
with $\tilde \aphi_{\jtau-1} = \overline{\E^{-\eptau \Kin} \aphi_{\jtau-1}}$, 
so its exponential factorizes; the assumption that $\WW$ is on-site implies that also the Gaussian integral factorizes. Calling $z = \tilde\aphi_{\jtau-1,\xx} \; \aphi_{\jtau,\xx}$, the integral over $h = h_{\jtau,\xx}$ is then simply
\beq\label{Phidef}
\Phi (z) 
=
\int \dd \mu_\vvv (h) \; 
\E^{z \; \E^{\I \sqrt{\eptau} h}}
=
\sum_{n=0}^\infty \frac{z^n}{n!} \; \E^{-\frac12 \; \eptau \vvv \; n^2} \;  
\eeq
The function $\cV (z) = -\ln \Phi (z)$ is then the interaction term in the action for the $\aphi$-integral. Obviously, $\abs{\Phi (z)} \le \E^{|z|}$ for all $|z|$, so the large-field behaviour is bounded by a quadratic term. Moreover $\cV(z) = z + \sfrac12 \eptau \vvv z^2 $ for small $|z|$, i.e.\ $\cV$ contains a quartic interaction term for small fields.  

\beq\label{aI}
\begin{split}
\tilZ^{(N,\beta,\Xspace,\ntau)} 
&=
\int_{\bC^{\XTspace}} \DD \aphi \; \;
\E^{-\bili{\bar\aphi}{\Quad (0) \; \aphi}_\XTspace
- \sum_\jtau \int_\xx \cV (\tilde\aphi_{\jtau-1,\xx} \; \aphi_{\jtau,\xx})}\;\;
\sfrac{{\bili{\bar \aphi_{\ntau}}{\aphi_0}_\Xspace}^N}{N!}
\end{split}
\eeq
Clearly, $\tilde\aphi_{\jtau-1,\xx} \; \aphi_{\jtau,\xx}$ is complex, and therefore this representation loses the positivity of $\WW$. However, it is possible to obtain a representation with a positive interaction term, in the sense that one can prove that the action can be changed to one with a positive interaction term without changing the limit $\eptau \to 0$. This will be done below in Section \ref{Analysis2}.

\subsubsection{The probabilistic expectation obtained after integrating over the $\aphi$-fields}
The $\aphi$-integral  in \Ref{HSFI} is Gaussian as well. Performing this integral,
the polyomial of degree $N$ in the integrand results in a permanent of order $N$, specifically
\begin{equation}\label{Qhpermanent}
\tilZ^{(N,\beta,\Xspace,\ntau)}
=
\frac{1}{N!} \sum_{\sigma \in \cS_N}
\ili_{\xx_1, . . , \xx_N}
\int \dd\mu_\WWW (h) \; 
\pli_{k=1}^N 
\Cov(h)_{(0,\xx_k),(\ntau,\xx_{\sigma(k)})} \; .
\end{equation}
The integrand is still complex because $\Cov(h)$ is a complex function, but this function will be shown to be bounded. The Gaussian measure itself is real and normalized, i.e.\ a probability measure). This motivates the notation 
\beq\label{hexpecdef}
\hexpec{F(h)}
=
\int F(h) \dd\mu_\WWW (h) \; .
\eeq
Introduce the additional notation for the $\Xspace$-averaged permanent, 
\beq
\Prm{N}{\Xspace} \Cov_{\jtau,\jtau'}
=
\sum_{\pi \in \scS_N} 
\int_{\xx_1, \ldots, \xx_N}
\pli_{n=1}^N
\Cov_{(\jtau,\xx_n),(\jtau', \xx_{\pi(n)})} \; .
\eeq
Then
\beq\label{tilZshort}
\tilZ^{(N,\beta,\Xspace,\ntau)}
=
\frac{1}{N!} \; 
\hexpec{
\Prm{N}{\Xspace} \Cov (h)_{0,\ntau}
}
\eeq
A similar expression holds for the generating function with $\opF$ insertion.
This representation will be the basis for most of the bounds given in this paper. 

\section{The Random-Walk Representation}
\label{RWsec}

By \Ref{termNeum}, the $k^{\rm th}$ factor in the product in (\ref{Qhpermanent}) is given by
\beq
\Cov(h)_{(0,\xx_k),(\ntau,\xx_{\sigma(k)})}
=
\left[
\pli_{j_k=1}^\ntau
\left(
\E^{-\eptau \Kin} \; \E^{-\I \sqrt{\eptau} h_{j_k}}
\right)
\right]_{(0,\xx_k),(\ntau,\xx_{\sigma(k)})} \; .
\eeq
In the basis of position eigenstates, $\E^{-\I \sqrt{\eptau} h_{j_k}}$ is diagonal, and thus writing out all matrix products gives 
\beq\label{QNeumann}
\begin{split}
\Cov(h)_{(0,\xx_k),(\ntau,\xx_{\sigma(k)})}
=
\sum_{\ywalk{k}} 
\pli_{j=1}^\ntau
\left(
\E^{-\eptau \Kin}
\right)_{\ywalkel{k}{j-1},\ywalkel{k}{j}}
\;
\E^{-\I \sqrt{\eptau} h_{j} (\ywalkel{k}{j})}
\end{split}
\eeq
where $\ywalk{k} = (\ywalkel{k}{0}, \ldots \ywalkel{k}{\ntau})$ with 
$\ywalkel{k}{0} = \xx_k$, $ \ywalkel{k}{\ntau} = \xx_{\sigma(k)}$ and 
$ \ywalkel{k}{j} \in \Xspace$ for all $j \in \{ 1, \ldots, \ntau -1\}$. 
It is natural to interpret $\ywalk{k}$ as a walk from $\xx_k$ to $\xx_{\sigma(k)}$, and write $\ywalk{k}: \xx_k \to \xx_{\sigma(k)}$ as well as 
\beq
\pli_{j=1}^\ntau
\left(
\E^{-\eptau \Kin}
\right)_{\ywalkel{k}{j-1},\ywalkel{k}{j}}
=
\Pamp{\ywalk{k}} \; 
\eeq
for the product of transition amplitudes along the walk. 
Let  $\XTspace = \{1,\ldots,\ntau\} \times \Xspace$. 
Define the density associated to $\ywalk{k}$ by 
$\dens{k}:\XTspace \to \bR$,
where, for $(\jtau,\xx) \in \XTspace$, 
$\dens{k}_{\jtau,\xx} = \speps^{-d}\; \delta_{\xx, \ywalkel{k}{j}}$. Then 
\beq
\pli_{j=1}^\ntau
\E^{-\I \sqrt{\eptau} h_{j} (\ywalkel{k}{j})}
=
\E^{-\I \sqrt{\eptau} (h, \dens{k})_{\XTspace}}
\eeq
Let $\ddens (\yywalk) = \sum_{k=1}^N \dens{k}$. Then
\beq
\begin{split}
\pli_{k=1}^N
\Cov(h)_{(0,\xx_k),(\ntau,\xx_{\sigma(k)})}
&=
\sum_{\ywalk{1}, . . , \ywalk{N}} 
\pli_{k=1}^N
\Pamp{\ywalk{k}}
\;
\E^{\I \sqrt{\eptau} \bili{h}{\ddens (\yywalk)}_\XTspace}
\end{split}
\eeq
The Gaussian integral over $h$ can now be done by the standard formula
\beq
\int \dd\mu_\WWW (h) \;   \E^{\I \sqrt{\eptau} \bili{h}{\ddens (\yywalk)}_\XTspace}
=
\E^{-\frac{\eptau}{2} \; \bili{\ddens(\yywalk)}{\WWW \ddens (\yywalk)}_\XTspace} \; ,
\eeq
and the bilinear form in the exponent then worked out in terms of summations over $k$ and $k' \in \{1, \ldots, N\}$ using the definition of $\ddens$. 
\intremark{%
Details
\beq
\begin{split}
\int \dd\mu_\WWW &(h) \;   \E^{\I \sqrt{\eptau} \bili{h}{\ddens (\yywalk)}_\XTspace}
\\
&=
\E^{-\frac{\eptau}{2} \; \bili{\ddens(\yywalk)}{\WWW \ddens (\yywalk)}_\XTspace}
\\
&=
\E^{-\frac{\eptau}{2} \sum_{k,k'=1}^N 
\bili{\dens{k}}{\WWW \dens{k'}}_\XTspace}
\\
&=
\E^{-\frac{\eptau}{2} \sum_{k,k'=1}^N 
\sum_{\jtau=1}^\ntau
\bili{\dens{k}_\jtau}{\WW\; \dens{k'}_\jtau}_\Xspace}
\\
&=
\E^{-\frac{\eptau}{2} \sum_{k,k'=1}^N 
\sum_{\jtau=1}^\ntau
\WW_{\ywalk{k}_\jtau, \ywalk{k'}_{\jtau}}}
\\
&=
\E^{
- \eptau
\sum_{\jtau=1}^\ntau
\sum_{1 \le k < k' \le N}
\WW_{\ywalk{k}_\jtau, \ywalk{k'}_{\jtau}}
-
\sfrac{\eptau}{2} 
\sum_{k=1}^{N}
\sum_{\jtau=1}^\ntau
\WW_{\ywalk{k}_\jtau, \ywalk{k}_{\jtau}}
}
\end{split}
\eeq
}
With the notations 
\beq
\tau_\jtau = \frac{\jtau}{\ntau}\; \beta\; , 
\qquad 
\int_\tau \; F(\tau)= \eptau \sum_{\jtau=1}^\ntau F(\tau_\jtau)
\eeq 
the canonical partition function then becomes in the translation-invariant case\footnote{a simple generalization of this formula holds for $\WW$ that are not translation-invariant.}
\beq
\canZ^{(N,\beta,\Xspace,\ntau)}
=
\frac{1}{N!} \sum_{\sigma \in \cS_N} \;
\ili_{\xx_1, . . , \xx_N}
\sum_{\ywalk{1}, . . , \ywalk{N}\atop \ywalk{k}: \xx_k \to \xx_{\sigma(k)}} 
\;
\E^{
- \cV (\ywalk{1}, . . , \ywalk{N}) 
}
\;
\pli_{k=1}^N
\Pamp{\ywalk{k}}
\eeq
with 
\beq
\cV (\ywalk{1}, . . , \ywalk{N}) 
=
\frac12
\sum_{1 \le k , k' \le N}
\int_\tau
\WW \left(\ywalk{k}(\tau) - \ywalk{k'} (\tau)\right)
\eeq
Thus the partition function is a sum over collections of $N$ random walks $\ywalk{1}, \ldots , \ywalk{N}$ that have a local-in-time interaction $\WW$.
This is a discrete analogue of the representation as symmetrized interacting Brownian motions of \cite{AdamsKoenig}. The convergence of this partition function for interacting random loops as $\eptau \to 0$ holds as a direct consequence of the Lie product formula (\ref{Lieprod}). While it is clear that the representations are related since they equal the same partition function, this provides an explicit illustration of the role of the factor \Ref{sandprich} in the integral representation of the canonical partition function, and a natural finite-dimensional approximation of the Wiener integral over Brownian motions.

\section{Analysis 1: Covariances}\label{Analysis1}

This section contains the basic estimates for the covariance. They are elementary and the proofs are easy, but they will be crucial because they are uniform in the auxiliary field $h$. 

From now on, the spatial lattice constant will be set to unity, $\speps = 1$. 

\subsection{Stochasticity and bounds for the covariance}


\bigskip\begin{lemma}\label{Quadlemma}
If $\Kin$ generates a stochastic process (see $(\ref{stochcond})$), then for all
$\jtau'\ge\jtau$, all $\xx,\xx'\in \Xspace$, and all $h$
\beq\label{Covbound}
\begin{split}
\abs{
\Cov(h)_{(\jtau,\xx), (\jtau',\xx')}}
&\le 
\abs{
\Cov(0)_{(\jtau,\xx), (\jtau',\xx')}}
\\
&=
\left(
\E^{(\jtau'-\jtau)\eptau\; \Kin}
\right)_{\xx,\xx'} \; .
\end{split}
\eeq	
\end{lemma}

\begin{proof}
By \Ref{termNeum}, and writing out the matrix products in the random-walk notation of \Ref{QNeumann}, 
\beq
\begin{split}
\Cov(h)_{(\jtau,\xx), (\jtau',\xx')}
&=
\sum_{\yywalk: \xx \to \xx'}
\pli_{j=\jtau+1}^{\jtau'}
\left(
\E^{-\eptau \Kin}
\right)_{\yywalk_{j-1},\yywalk_{j}}\;
\E^{\I \sqrt{\eptau} h_{\jtau,\yywalk_{j}}}
\end{split}
\eeq
All $h$-dependent factors have absolute values equal to $1$, and, by hypothesis, the matrix elements of $\E^{-\eptau \Kin}$ are all nonnegative, hence unaffected by taking absolute values. Thus
\beq
\begin{split}
\abs{
\Cov(h)_{(\jtau,\xx), (\jtau',\xx')}}
&\le
\sum_{\yywalk: \xx \to \xx'}
\pli_{j=\jtau+1}^{\jtau'}
\left(
\E^{-\eptau \Kin}
\right)_{\yywalk_{j-1},\yywalk_{j}}
\\
&=
\Cov(0)_{(\jtau,\xx), (\jtau',\xx')}
\\
&=
\left(
\E^{(\jtau'-\jtau)\eptau\; \Kin} 
\right)_{\xx,\xx'} \; .
\end{split}
\eeq
\end{proof}

\subsection{Modified covariances}

In this section, a number of modifications of the operator $\Quad (h)$ are introduced, and their properties, as well as those of their inverses, are discussed. These modifications will be related to different actions that all reproduce the same limit for the canonical partition function as $\eptau \to 0$. 

\subsubsection{$\tQuad$}
Let $\cB $ be the diagonal matrix $\cB_{\jtau\jtau'} = \opB_\jtau \delta_{\jtau,\jtau'}$, i.e.\
\beq
\cB
=
\left[
\begin{array}{cccccc}
\opB_0 & 0 & 0 & 0 & \ldots & 0 \\
0 & \opB_1 & 0 & 0 & \ldots & 0\\[2ex]
 &  & & \ddots & &   \\[2ex]
0 & 0 & 0 & & \opB_{\ntau-1} & 0 \\
0 & 0 & 0 & \ldots & 0 & \opB_\ntau
\end{array}
\right]
\eeq
Then 
\beq
\tQuad(h) 
=
\Quad (h) \; \cB
=
\left[
\begin{array}{cccccc}
\opB_0 & - \opA_1 \opB_1 & 0 & 0 & \ldots & 0 \\
0 & \opB_1 & - \opA_2 \opB_2 & 0 & \ldots & 0\\[3ex]
 &  & \ddots & & \ddots &  \\[3ex]
0 & 0 & 0 & & \opB_{\ntau-1} & -A_\ntau \opB_\ntau \\
0 & 0 & 0 & \ldots & 0 & \opB_\ntau
\end{array}
\right]
\eeq
If $\opB_0, \ldots, \opB_\ntau$ are invertible, then $\tQuad$ is invertible, with the inverse $\tCov(h) = \tQuad(h)^{-1}$ satisfying
\beq
\tCov(h)_{(\jtau,\xx),(\jtau'\xx')}
=
\int_{\xy}
(\opB_\jtau^{-1})_{\xx,\xy} \; 
\Cov(h)_{(\jtau,\xy),(\jtau'\xx')}
\eeq

\begin{lemma}
Let $\opB_0 = \opId$ and $\opB_{\jtau} = \E^{-\I\sqrt{\eptau} h_{\jtau}}$. Then $\opA_\jtau \opB_\jtau = \E^{\eptau \Kin}$ for all $j \in \{1, \ldots, \ntau\}$, and
\beq\label{tQuaddef}
\tQuad (h) _{\jtau,\jtau'}
=
\de_{\jtau,\jtau'} \; \E^{- \I \sqrt{\eptau} \; h_{j}}
-
\de_{\jtau+1,\jtau'} \; \E^{-\eptau \Kin} 
\eeq
is invertible, and 
\beq\label{tildegut}
\tCov(h)_{(0,\xx),(\jtau'\xx')}
=
\Cov(h)_{(0,\xx),(\jtau'\xx')} \; .
\eeq
If $\Kin $ generates a stochastic process, then for all $h$ and all $\jtau,\jtau', \xx, \xx'$
\beq\label{tCovbound}
\abs{\tCov(h)_{(\jtau,\xx), (\jtau', \xx')}}
\le
\abs{\tCov(0)_{(\jtau,\xx), (\jtau', \xx')}} \; .
\eeq
\end{lemma}

\begin{proof}
Eq.\ (\ref{tildegut}) holds because $\opB_0= \opId$. By definition, 
\beq
\tCov(h) = \cB(h)^{-1} \; \Cov(h) \; .
\eeq
Every $\opB_\jtau $ is diagonal in the spatial indices, so $\cB(h)^{-1}$ is diagonal, and 
\beq
\abs{\cB(h)^{-1}_{(\jtau,\xx), (\jtau, \xx)}}
\le 
1
=
\cB(0)^{-1}_{(\jtau,\xx), (\jtau, \xx)}
\eeq
Thus (\ref{tCovbound}) follows by matrix multiplication 
and Lemma \ref{Quadlemma}. 
\end{proof}

\subsubsection{$\ttQuad$ and $\tttQuad$}
Consider furthermore the following two operators motivated by a small-$\eptau$ expansion. 
Let $\uu: \Xspace \to \bR$ and let $\tKin$ be a self-adjoint linear operator on $\bC^\Xspace$, and regard, as with $h$, functions of $\uu$ as multiplication operators, i.e.\ $(\E^{-\eptau \uu})_{\xx,\xx'} = \E^{-\eptau \uu_\xx} \; \delta_{\xx,\xx'}$. 
Define $\ttQuad (h,\uu)$ and $\tttQuad (h)$ by 
\beq\label{ttQuaddef}
\ttQuad (h,\uu)_{\jtau,\jtau'}
=
\de_{\jtau,\jtau'} \; (\opId - \I \sqrt{\eptau} \; h_{j}) \; \E^{-\eptau \uu}
-
\de_{\jtau+1,\jtau'} \; \E^{-\eptau \tKin} 
\eeq
and 
\beq\label{tttQuaddef}
\tttQuad (h)_{\jtau,\jtau'}
=
\de_{\jtau,\jtau'} \; (\opId - \I \sqrt{\eptau} \; h_{j})
-
\de_{\jtau+1,\jtau'} \; (\opId - \eptau \tKin ) \; .
\eeq
Let 
\beq\label{tubedef}
\Tube_\eptau 
=
\{ h \in \bC^\XTspace:  \abs{\Im h_{\jtau,\xx}} < \eptau ^{-1/2} \; \forall (\jtau,\xx) \in \XTspace \}
\eeq

\newpage
\begin{lemma}  
\label{ttQuadlemma}

Let $h \in \bR^\Xspace$ and $\uu \in \bR^\Xspace$. 

\begin{enumerate}

\item
$\ttQuad(h,\uu)$ is invertible and its inverse 
\beq
\ttCov (h,\uu) = \ttQuad(h,\uu)^{-1}
\eeq
has norm  bounded uniformly in $h$. $ \ttCov$ is analytic in $h$ on $\Tube_\eptau$. 
 
\item
If $\E^{-\eptau \tKin } < \E^{-\eptau \uu}$, then
\beq
\begin{split}
\Re \bili{\bar\aphi}{\ttQuad (h,\uu) \; \aphi}_\XTspace
&=
\Re \bili{\bar\aphi}{\ttQuad (0,\uu) \; \aphi}_\XTspace
\\
&\ge 
\sum_{\jtau=0}^\ntau
\bili{\bar b_\jtau}{(1- \tilde K)\; b_\jtau}_\Xspace
> 0 \; .
\end{split}
\eeq
(Here $\tilde K = \E^{\frac{\eptau}{2} \uu}  \E^{-\eptau \tKin }   \E^{\frac{\eptau}{2} \uu}$ and $b_\jtau = \E^{-\frac{\eptau}{2} \uu}\; \aphi_\jtau$.)

\item
If $\tKin $ generates a stochastic process, then 
\beq\label{ttCovbound}
\abs{\ttCov(h,\uu)_{(\jtau,\xx),(\jtau',\xx')}}
\le
\ttCov(0,\uu)_{(\jtau,\xx),(\jtau',\xx')} \; .
\eeq
For any $\delta > 0$, there is $\ntau_0 = \ntau (\delta, \beta, \Vert{\tKin}\Vert, \norm{\uu})$ so that for all $\ntau > \ntau_0$, all $\jtau, \jtau' \in \{ 0, \ldots, \ntau\}$ and all $\xx, \xx' \in \Xspace$
\beq\label{113}
\abs{
\ttCov(0,\uu)_{(\jtau,\xx),(\jtau',\xx')} 
-
\left(\E^{- (\jtau-\jtau') \eptau (\tKin - \uu)}
\right)_{\xx,\xx'}}
<
\delta
\eeq
Consequently, for $\ntau > \ntau_0$, the right hand side of \Ref{ttCovbound} is bounded uniformly in $\jtau,\jtau'$ and $\ntau$:
\beq\label{ttCunif}
\abs{\ttCov(h,\uu)_{(\jtau,\xx),(\jtau',\xx')}}
\le
\E^{\beta E_{min}} + \delta
\eeq
where $E_{min} = \; {\rm infspec} \; (\tKin -\uu)$. 

\end{enumerate}

\end{lemma}

\begin{proof}
(a) Viewed as a matrix in $\jtau,\jtau'$, $\ttQuad$ is upper triangular, and the diagonal part $\cD$ of $\ttQuad$, given by 
\beq\label{cDjjp}
\cD (h,\uu) _{\jtau,\jtau'}
=
\delta_{\jtau,\jtau'} \; 
 (\opId - \I \sqrt{\eptau} \; h_{j}) \; \E^{-\eptau \uu} \; ,
\eeq
is invertible for all $\uu, h \in \bR^\XTspace$. For all $h \in \Tube_\eptau$, the inverse $\cD (h,\uu)^{-1} $ is analytic in $h$ because it only contains factors of $( 1 - \I \sqrt{\eptau} \; h_{j,\xx}) \; \E^{-\eptau \uu_\xx})^{-1}$. The inverse of $\ttQuad$ is given by a Neumann series that terminates after at most $\ntau+1$ terms, hence has the same analyticity properties. The first equality in (b) holds because the $h$-dependent term drops out of the hermitian part of $\Quad$; the inequality the follows by the hypothesis made in $(b)$.  (c) The proof of \Ref{ttCovbound} is a straightforward adaptation of that of Lemma \ref{Quadlemma}, writing out the finite Neumann sum for $\ttCov$ as a matrix product, and using that by (\ref{cDjjp}), $\abs{\cD (h,\uu) _{(\jtau,\xx),(\jtau',\xx')}} \le \cD (0,\uu) _{(\jtau,\xx),(\jtau',\xx')}$. To see \Ref{113}, note that 
an easy variation of the standard proof of the Lie product formula (as given, e.g., in \cite{msbook}) results in
\beq
\norm{
\E^{\eptau (\opA + \opB) \jtau }
-
\left(
\E^{\eptau\;\opA}\;
\E^{\eptau\;\opB}
\right)^\jtau}
\le
\jtau \; 
\E^{\eptau (\jtau - 1) (\norm{\opA} + \norm{\opB})} \;
\norm{
\E^{\eptau (\opA + \opB)}
-
\E^{\eptau\;\opA}\;
\E^{\eptau\;\opB}
}
\eeq
and an estimate of the last factor gives $O(\eptau^2)$, hence
(see Appendix \ref{BCHapp})
\beq
\left(
\E^{\eptau\;\opA}\;
\E^{\eptau\;\opB}
\right)^\jtau
=
\E^{\eptau (4\opA + 2\opB) \jtau }
+ \opR
\eeq
with
\beq
\norm{\opR}
\le
\frac{1}{\ntau} \; \beta\; 
\E^{\beta (\norm{\opA} + \norm{\opB})} \;
(\norm{\opA} + \norm{\opB})^2 \; .
\eeq
Given $\delta > 0$, $\ntau_0$ is defined by the condition $\norm{\opR} \le \delta$. Then use $|\opR_{\xx,\xy}| \le \norm{\opR}$, and $\abs{\jtau- \jtau'} \le \ntau$, so that 
$\abs{\jtau-\jtau'} \eptau \le \beta$, hence $\abs{{\E^{- (\jtau-\jtau') \eptau (\tKin - \uu)}}_{\xx,\xx'}} \le \E^{\beta E_{min}}$. Thus for $\ntau > \ntau_0$, \Ref{ttCunif} holds.
\end{proof}

\begin{lemma}  
\label{tttQuadlemma}

Let $h \in \bR^\Xspace$ and $\uu \in \bR^\Xspace$. 

\begin{enumerate}

\item
$\tttQuad (h)$ is invertible and its inverse $\tttCov (h) = \tttQuad(h)^{-1}$
has norm  bounded uniformly in $h$.  $\tttCov$ is analytic in $h$ on $\Tube_\eptau$.

\item
If $\tKin > 0$ then 
\beq
\begin{split}
\Re \bili{\bar\aphi}{\tttQuad (h) \; \aphi}_\XTspace
&\ge
\Re \bili{\bar\aphi}{\tttQuad (0) \; \aphi}_\XTspace
\\
&\ge 
\eptau
\sum_{\jtau=0}^\ntau
\bili{\bar a_\jtau}{\tKin\; a_\jtau}_\Xspace
> 0 \; .
\end{split}
\eeq

\item
If $\tKin_{\xx,\xy} \le 0$ for all $\xx \ne \xy$ and if $\eptau$ is so small  that $1 - \eptau \tKin_{\xx,\xx} \ge 0$ for all $\xx$, then 
\beq\label{tttCovbound}
\abs{\tttCov(h)_{(\jtau,\xx),(\jtau',\xx')}}
\le
\tttCov(0)_{(\jtau,\xx),(\jtau',\xx')} \; .
\eeq
Moreover, the right hand side of \Ref{tttCovbound} is bounded uniformly in $\ntau$.

\end{enumerate}

\end{lemma}

\begin{proof}
Similar to the proof of Lemma \ref{ttQuadlemma}.
\end{proof}

\subsubsection{$\tttaQuad$}
To place the kinetic term $\tKin$ in the time-local ($\jtau=\jtau'$) term, consider
\beq\label{tttaQuaddef}
\tttaQuad (h)_{\jtau,\jtau'}
=
\de_{\jtau,\jtau'} \; (\opId + \eptau \tKin - \I \sqrt{\eptau} \; h_{j})
-
\de_{\jtau+1,\jtau'} \; \opId  \; .
\eeq
The bounds for this covariance are analogous to the previous ones; 
the argument is slightly different, and will be given in several steps in the following, without collecting it in a single Lemma. 

If $\tKin > 0$, then 
\beq
\Re \bili{\bar \aphi}{ \tttaQuad (h) \; \aphi}_\XTspace
=
\Re \bili{\bar \aphi}{ \tttaQuad (0) \; \aphi}_\XTspace
\ge
\eptau
\sum_{\jtau=0}^\ntau
\bili{\bar\aphi}{\tKin \aphi}_\Xspace
> 0 \; .
\eeq

\bigskip\begin{lemma}\label{posplusim}
Let $\scB $ be the algebra of bounded operators on a Hilbert space and 
$\opQ,\opH \in \scB$. Furthermore, let $\opQ_0$ be a positive operator and $\opH = \opH^\dagger$, and  let
\beq
\opQ 
=
\opQ_0 + \I \opH \; . 
\eeq
Then $\opQ^{-1}$ exists , and $\norm{\opQ^{-1}} \le \norm{\opQ_0^{-1}}$. 
For $\vv \ne 0$
\beq\label{QQ0}
\Re \bracket{\vv}{\opQ \vv}
\ge
\bracket{\vv}{\opQ_0 \vv} > 0 \; .
\eeq

\end{lemma}

\begin{proof}
Because $\opQ_0 > 0$, it is invertible and has a positive square root $\opQ_0^{\frac12} $. 
Set $\opA = \opQ_0^{-1/2} \; H \;  \opQ_0^{-1/2}$. Then 
\beq
\opQ
=
\opQ_0^{\frac12} 
(\opId + \I \opA ) \; 
\opQ_0^{\frac12}  \; .
\eeq
Because $\opA = \opA^\dagger$, 
$\norm{(\opId + \I \opA) \vv}^2  = \norm{\vv}^2+ \norm{\opA \vv}^2 \ge \norm{\vv}^2 $.
Thus $\norm{(1+\I \opA)^{-1}} \le 1$, hence $\opQ^{-1}$ exists and 
\beq
\opQ^{-1} 
\le 
\Vert\opQ_0^{-\frac12}\Vert^2
=
\norm{\opQ_0^{-1}}
\eeq
where the last equality holds since $\opQ_0^{-\frac12}$ is self-adjoint. 
Moreover 
\beq
\bracket{\vv}{\opQ \vv}
=
\bracket{\vv}{\opQ_0 \vv}
+
\I \bracket{\opQ_0^{\frac12}\vv}{\opA \opQ_0^{\frac12}\vv}
\eeq
so \Ref{QQ0} holds. 
\end{proof}

If $\tKin > 0$, and because $h$ is real, $\tttaQuad (h)$ satisfies the hypotheses of Lemma \ref{posplusim} with $\opQ_0 = \tttaQuad (0)$. 
In the estimate of products as above, then use 
\beq
\abs{{\tttaQuad (h)^{-1}}_{\xx,\xy}} 
\le 
\norm{\tttaQuad(h)^{-1}} 
\le 
\norm{\tttaQuad(0)^{-1}} \; .
\eeq
When doing the terminating Neumann series for this inverse, use the spectral theorem for $\tKin$ and the elementary inequality that for all $x \ge 0$ and $\jtau \le \ntau \in \bN_0$
\beq
\left(1+ \frac{x}{\ntau}\right)^\jtau
=
\E^{\jtau \ln (1+\frac{x}{\ntau})}
\le
\E^{\frac{\jtau}{\ntau} x}
\le 
\E^{x}
\eeq
to see that the norm of $\tttaQuad(0)^{-1}$ is bounded uniformly in $\ntau$. 

\begin{lemma}\label{tttaQuadlemma}
If $\tKin \ge 0$ and the off-diagonal matrix elements of $\tKin$ are negative, then 
$\tttaCov (h) = \tttaQuad(h)^{-1}$ satisfies
\beq
\abs{\tttaCov (h)_{(\jtau,\xx),(\jtau',\xx')}}
\le
\tttaCov (0)_{(\jtau,\xx),(\jtau',\xx')}
\eeq
and the right hand side is bounded uniformly in $\ntau$. 
\end{lemma}

\begin{proof}
The diagonal of $\tttaQuad(h)$, as a matrix in the $\jtau$ indices, contains 
\beq
\opD (h_\jtau)
=
\opId + \eptau \tKin + \I \sqrt{\eptau} h_\jtau
\eeq
By hypothesis, $\tKin=\tKin_d+\tKin_{od}$ where the off-diagonal part satisfies $(\tKin_{od})_{\xx,\xx'} < 0$ for all $\xx$ and $\xx'$.  By hypothesis, the geometric series for the inverse
\beq
\opD (h_\jtau)^{-1}
=
\sum_{n=0}^\infty
(\opId + \eptau \tKin_d + \I \sqrt{\eptau} h_\jtau)^{-1}
\left[
-\eptau \tKin_{od} \, 
(\opId + \eptau \tKin_d + \I \sqrt{\eptau} h_\jtau)^{-1}
\right]^n
\eeq
converges. $\tKin_d$ and $h_\jtau$ commute, since both are diagonal in $\xx$. 
(Moreover $\tKin \ge 0$ implies $e_d(\xx) = \tKin_{\xx,\xx} \ge 0$.) Use 
\beq
\abs{\frac{1}{1 + \eptau e_d(\xx) + \I \sqrt{\eptau} h_{\jtau,\xx}}}
\le
\frac{1}{1+  \eptau e_d(\xx) } 
\eeq
and $(\tKin_{od})_{\xx,\xx'} < 0 $ to see that $\abs{\opD (h_\jtau)^{-1}_{\xx,\xx'}} \le 
\abs{\opD (0)^{-1}_{\xx,\xx'}}$. Using the terminating Neumann series for $\tttaQuad^{-1}$ then concludes the proof. 
\end{proof}
 
\section{Analysis 2: Time-Continuum Limit}\label{Analysis2}

The existence of the time-continuum limit $\ntau \to \infty$ for the integrals (\ref{ZcN}), (\ref{aI}), and (\ref{Qhpermanent}) for the canonical partition function holds simply by the Lie product formula, and it is already in this form very useful to do analysis. Sometimes, however, one is interested in simpler actions that resemble more closely those occurring in formal continuum-time functional integrals. Heuristically, they are obtained by dropping higher-order terms in $\eptau$, but proving that these terms can really be dropped is not straightforward. For many-fermion systems, it has been done using tree expansion techniques and determinant bounds, sometimes combined with multiscale analysis \cite{BGPS}. Later \cite{FKT,PSUV}, it was done without multiscale techniques.All these proofs rely on determinant estimates which allow for convergent perturbation expansions under suitable conditions. The bosonic case is harder. In \cite{BFKT3}, existence of the time-continuum limit for a FIR of the many-boson system in the grand-canonical ensemble was proven by multiscale techniques (renormalization group using a decimation transformation in time). Ref.\ \cite{BFKT3} contains much more than just the proof that the limit exists; it also gives a rigorous justification for an effective action with a fixed short-time cutoff, which can be used in the analysis of Bose-Einstein condensation in the thermodynamic limit \cite{BFKT4}. The analysis required to do this involves decompositions in large and small fields, a conceptually clear and analytically powerful field theoretical method which, however, entails some technical overhead. 

In this section, I prove that a variety of different actions give rise to the same limit $\eptau \to 0$ of the partition function and the unnormalized expectation values of the canonical ensemble of bosons. The proof given here does not require any multiscale technique, only the uniform bounds on covariances given in the last section. This makes it possible to avoid a large-field analysis. As mentioned, it is not strictly necessary for the further analysis of the system to consider modified actions, but it serves to illustrate the application of the main observations made above, namely (i) the analyticity of the covariances in the auxiliary field $h$ (ii) their uniform boundedness as functions of $h$, and (iii) the fact that the permanent in the canonical ensemble factorizes easily over covariances and hence allows for a straightforward use of (i) and (ii) in remainder estimates. The factorization can be expressed in terms of the cycles of permutations, and then used to study the thermodynamic limit by convergent decoupling expansions \cite{toappear}.

Recall that $\eptau \to 0$ always means $\ntau \to \infty$ with $\eptau = \beta/\ntau$; this will be important in estimates since convergence or uniform boundedness as $\eptau \to 0$ really means convergence or uniform boundedness  as $\ntau \to \infty$, and statements that constants are independent of $\eptau$ will mean that they are independent of $\ntau$ as well.

\subsection{Convergence of modified partition functions}
For simplicity, I consider only the case $\opF = \opId$ here. 

The following lemma implies that the $h$-dependence can be moved to the diagonal part of the $\Quad$-matrix, without changing the partition function. 

\bigskip\begin{lemma} \label{tildeZrep}
Let $\tQuad$ be as in \Ref{tQuaddef} and $\tCov (h) = \tQuad (h)^{-1}$. Then
\beq\label{tZtQuad}
\tilZ^{(N,\beta,\Xspace,\ntau)}
=
\frac{1}{N!} \; 
\hexpec{
\Prm{N}{\Xspace} \tCov (h)_{0,\ntau}}
\eeq
\end{lemma}

\begin{proof}
This follows immediately from \Ref{tilZshort} and \Ref{tildegut}.
\end{proof}

It is interesting to note that formally, already this gives rise to an integral representation where the interaction term is local in time. Assuming for simplicity that the interaction is local, i.e.\ $\WW_{\xx,\xy} = \WW_{\xx} \delta_{\xx,\xy}$, 
\beq\label{wouldbenice}
\int \dd\mu_\WWW (h) \; 
\E^{-\bili{\bar\aphi}{\tQuad (h) \; \aphi}_\XTspace}
=
\E^{-\bili{\bar\aphi}{\tQuad (0) \; \aphi}_\XTspace
- \eptau \sum_\jtau \int_\xx \cV(\abs{\aphi_{\jtau,\xx}}^2)}
\eeq
where $\cV = - \ln \Phi $, with $\Phi $ given in \Ref{Phidef}. The interaction term then has the properties that it is local in the time index $\jtau$, hence positive, and moreover $\Phi$  can grow at most like $\E^{|\aphi|^2}$ as $|\aphi| \to \infty$. Formally, one can also derive this representation directly from the trace using a modified coherent-state formula. That derivation is, however, formal because integrals that get exchanged do not converge absolutely. This is reflected here in that $\Re \bili{\bar\aphi}{\tQuad (h) \; \aphi}_\XTspace$ is not positive for arbitrary $h$, hence $\tQuad (h)$ cannot be used as a covariance for a Gaussian integral, so again, going from the right hand side of \Ref{tZtQuad} to the right hand side of \Ref{wouldbenice} involves non-convergent integrals. This problem may be bypassed by an additional regularization suppressing large $\aphi$. 

In a heuristic procedure, a further step is to expand the exponential to linear order, 
$\E^{-\I \sqrt{\eptau} h_{\jtau,\xx}} = 1 - \I \sqrt{\eptau} h_{\jtau,\xx} + O( \eptau)$ and drop the order $\eptau$ remainder. The integral over $h$ then becomes Gaussian, and results in a quartic interaction for the $\aphi$ fields. This does not really work, but, as shown in the following, a slight modification does, using integration by parts in the auxiliary field $h$, as given by Lemma \ref{intbypartslemma} in Appendix \ref{intbypartssec}. 

\bigskip\begin{lemma}
\label{ttilZlemma}
Assume that $\Kin$ generates a stochastic process. Let $\ttQuad (h,\uu)$ be defined as in \Ref{ttQuaddef}, with $u_\xx = \frac12 \WW_{\xx,\xx}$ and $\tKin=\Kin$, and 
\beq\label{ttilZdef}
\ttilZ^{(N,\beta,\Xspace,\ntau)}
=
\frac{1}{N!} \; 
\hexpec{
\Prm{N}{\Xspace} \ttCov (h,\uu)_{0,\ntau}
}
\eeq
Then $\ttilZ^{(N,\beta,\Xspace,\ntau)} \to \canZ^{(N,\beta,\Xspace)}$ as $\ntau \to \infty$.  
\end{lemma}

\begin{proof}
By \Ref{Lieprod}, and Lemma \ref{tildeZrep},  $\tilZ^{(N,\beta,\Xspace,\ntau)} (\opH, \opF)  $ converges to $ \canZ^{(N,\beta,\Xspace)} (\opH, \opF) $ as $\ntau \to \infty$, so it suffices to prove 
\beq
\ttilZ^{(N,\beta,\Xspace,\ntau)} 
-
\tilZ^{(N,\beta,\Xspace,\ntau)} (\opH, \opId) 
\gtoas{\ntau \to \infty} 0 \; .
\eeq
Overview of the proof. ---  The permanent $\Prm{N}{\Xspace}$ contains a product over $n \in \{ 1, \ldots, N\}$ of $\tCov (h)_{(0,\xx_n), (\ntau, \xx_{\pi(n)})}$ in $\tilZ$ and a product of $\ttCov (h)_{(0,\xx_n), (\ntau, \xx_{\pi(n)})}$ in $\ttilZ$. By the discrete product rule and the resolvent formula, a difference 
\beq
(\ttQuad (h,\uu) - \tQuad (h) )_{\jtau,\jtau'}
=
\delta_{\jtau,\jtau'} \; 
\left(
\opId - \I \sqrt{\eptau} h_\jtau  - \E^{-\I \sqrt{\eptau} h_\jtau } - \frac{\eptau}{2} \uu
\right) 
+ O (\eptau^{\frac32})
\eeq
will be made explicit. This term is of order $\eptau$ for any $\uu$ (e.g.\ for $\uu =0$), but because a sum over $\itau \in \{ 0, \ldots, \ntau\}$ arises as well and  $\eptau \ntau = \beta$, a straightforward estimate only  implies that the difference of partition functions is finite, but not that it vanishes. To get a better estimate, integrate by parts using Lemma \ref{intbypartslemma}, with a suitable (real) choice of $\uu$. Because $\Kin$ generates a stochastic process and $\uu$ is real, $\tKin$ also generates a stochastic process. Thus the uniform bounds \Ref{tCovbound} and \Ref{ttCovbound} hold, and they imply that the integral of the difference over $h$ is bounded by  $O(\eptau^{\frac32})$, which, even after multiplication with $\ntau$, still vanishes as $\eptau \to 0$. 

Details of the proof. --- The difference of partition functions is
\beq
\ttilZ^{(N,\beta,\Xspace,\ntau)}  
-
\tilZ^{(N,\beta,\Xspace,\ntau)} (\opH, \opId) 
=
\sfrac{1}{N!} \; \hexpec{D(h)}
\eeq
with
\beq
\begin{split}
D(h)
&=
\Prm{N}{\Xspace}  \ttCov (h,\uu)_{0,\ntau}
-
\Prm{N}{\Xspace}  \tCov (h)_{0,\ntau}
\\
&=
\sum_{\pi \in \scS_N} \int_{\xy_1, . . , \xy_N} 
\left[
\pli_{n=1}^N \ttCov (h,\uu)_{(0,\xy_n), (\ntau, \xy_{\pi(n)})}
-
\pli_{n=1}^N \tCov (h)_{(0,\xy_n), (\ntau, \xy_{\pi(n)})}
\right]
\\
&=
\sum_{\pi \in \scS_N} \int_{\xy_1, . . , \xy_N} 
\sum_{m=1}^N
\pli_{n<m} \ttCov (h,\uu)_{(0,\xy_n), (\ntau, \xy_{\pi(n)})}
\; 
\pli_{n>m} \tCov (h)_{(0,\xy_n), (\ntau, \xy_{\pi(n)})}
\\
&\mkern120mu
\left[
\ttCov (h,\uu) 
-
\tCov (h)\right]_{(0,\xy_m), (\ntau, \xy_{\pi(m)})}
\end{split}
\eeq
In the last equality, the discrete product rule was used. Now 
\beq\label{resolventQ}
\begin{split}
\ttCov (h,\uu)
-
\tCov (h)
=
\ttCov (h,\uu) \; 
\left[ \tQuad (h, \uu) - \ttQuad (h) \right] \; 
\tCov(h)
\end{split}
\eeq
The difference of $\Quad$'s is a diagonal matrix, so  
\beq
\left( \ttCov (h,\uu)
-
\tCov (h,\uu) \right)_{(0,\xy_m), (\ntau, \xy_{\pi(m)})}
=
\sum_{\itau = 0}^\ntau
\int_\xx
\ttCov (h)_{(\jtau,\xy_m),(\itau,\xx)} \; 
\tilde \Delta_{\itau,\xx} \;
\tCov (h)_{(\itau,\xx),(\jtau',\xy_{\pi(m)})} 
\eeq
with
\beq
\tilde \Delta_{\itau,\xx}
=
\E^{-\I \sqrt{\eptau} h_{\itau,\xx}} 
-
(\opId - \I \sqrt{\eptau} h_{\itau,\xx}) \;  \E^{-\eptau \uu_\xx} \; . 
\eeq
Set
\beq
\begin{split}
F(h) 
=&
\pli_{n<m} \tCov (h)_{(0,\xy_n), (\ntau, \xy_{\pi(n)})}
\; 
\pli_{n>m} \ttCov (h,\uu)_{(0,\xy_n), (\ntau, \xy_{\pi(n)})}
\\
&\tCov (h)_{(\jtau,\xy_m),(\itau,\xx)} 
\;
\ttCov (h,\uu)_{(\itau,\xx),(\jtau',\xy_{\pi(m)})} \; 
\end{split}
\eeq
($F$ depends on $\eptau$ and on $ \pi, \itau ,\xx, \xy_1, \ldots, \xy_N$, etc.).
Then 
\beq
\begin{split}
\sfrac{1}{N!} \; \hexpec{D(h)}
&=
\sfrac{1}{N!} \; \sum_{\pi \in \scS_N} \int_{\xy_1, . . , \xy_N} 
\sum_{\itau = 0}^\ntau
\int_\xx
\hexpec{
\tilde \Delta_{\itau,\xx} \; F(h)
}
\end{split}
\eeq
The sum over $\itau$ is the one mentioned in the overview. All other summations run over index sets that are independent of $\ntau$. To show that this expression vanishes in the limit $\ntau \to \infty$, it will suffice to use the bound
\beq
\abs{\sfrac{1}{N!} \; \hexpec{D(h)}}
\le
|\Xspace|^N\; (\ntau + 1) \; 
\max\limits_{\pi, \itau\atop \xy_1, \ldots \xy_N}
\abs{
\sum_{\xx \in \Xspace} 
\hexpec{
\tilde \Delta_{\itau,\xx} \; F(h)}
} \; .
\eeq
By Taylor expansion,
\beq\label{TaylorQ}
\tilde \Delta_{\itau,\xx}
=
\Diff_\eptau (h_{\itau,\xx}) + 
\eptau \uu_\xx +
\rho_\eptau
\eeq
with  $\Diff_\eptau $ given in \Ref{Diffeptau} and 
\beq
\abs{\rho_\eptau}
\le
\frac{\uu_\xx^2}{2} \; \eptau^2 + \eptau^{\frac32} \; \abs{\uu_\xx h_{\itau,\xx}} \; .
\eeq
By Lemma \ref{ttQuadlemma} (a) and Lemma \ref{tttQuadlemma} (a), the function $F$ is analytic in $h$ in a neighbourhood of $\bR_\XTspace$, and by  \Ref{tCovbound} and \Ref{ttCovbound}, 
\beq
\abs{F(h)} 
\le 
F(0) \; .
\eeq
Let $F^{\rm max}  =\max F(0)$, where the maximum is taken over $ \pi, \itau ,\xx, \xy_1, \ldots, \xy_N$. By Lemma \ref{ttQuadlemma} (c), $F^{\rm max}$ is bounded uniformly in $\ntau$. The Gaussian measure $\dd\mu_\WWW (h)$is normalized, its covariance satisfies $\WWW_{\jtau,\jtau'} = \delta_{\jtau,\jtau'} \; \WW$, and $\WW$ is  independent of $\eptau$. Thus, for all $r \ge 0$ there is $K_r > 0$ so that for all $\ntau$ and all $\itau,\xx$
\beq
\hexpec{ |h_{\itau,\xx}|^r} \le K_r \; .
\eeq
(The constant $K_r$ depends on $\WW$; $K_0 =1$ by normalization of $\mu_\WWW $.)
\intremark{Assume a local interaction $\WW$, where the integral factorizes into individual integrals over the $h_{\jtau, \xx}$. Then for even exponents $r$, this integral can be rewritten in the standard way with source terms and Laplacians, and one gets a Hafnian, i.e.\ one can expect the $r$-dependence to be $2^{-r} \sqrt{r!}^{-1}$ times a $\WW$-dependent constant to the power $r$. For odd exponents, use a Schwarz inequality to get essentially the same. }%
Thus
\beq
(\ntau + 1)
\abs{
\sum_{\xx \in \Xspace} 
\hexpec{
\rho_\eptau \; F(h)}
}
\le
F^{\rm max} \; (\ntau+1)\;
\left(
\frac{\uu_\xx^2}{2} \; \eptau^2 + \eptau^{\frac32} \; \abs{\uu_\xx} \; K_1
\right) 
\eeq
so that this contribution vanishes in the limit $\ntau \to \infty$ (since $\eptau = \frac\beta\ntau$). 
Because 
\beq
\sfrac{\partial}{\partial h_{\itau,\xx}}\;
\tCov (h)_{(\jtau,\xy),(\jtau',\xy')}
=
\I \; \sqrt{\eptau}\;
\int_{\xx}
\tCov (h)_{(\jtau,\xy),(\itau,\xx)} \; 
\E^{-\I \sqrt{\eptau} \; h_{\itau,\xx}}\; 
\tCov (h)_{(\itau,\xx),(\jtau',\xy')}
\eeq
and
\beq
\sfrac{\partial}{\partial h_{\itau,\xx}}\;
\ttCov (h)_{(\jtau,\xy),(\jtau',\xy')}
=
\I \; \sqrt{\eptau}\; 
\int_{\xx}
\ttCov (h)_{(\jtau,\xy),(\itau,\xx)} \; 
\ttCov (h)_{(\itau,\xx),(\jtau',\xy')}
\eeq
it follows, again from \Ref{tCovbound} and \Ref{ttCovbound} that $\hexpec{\abs{\frac{\partial F}{\partial h_{\itau,\xx}}}} < \infty $ and that \Ref{117} and \Ref{118} hold. Thus Lemma \ref{intbypartslemma} applies. 
\beq
(\ntau + 1) \; 
\max\limits_{\pi, \itau\atop \xy_1, \ldots \xy_N}
\abs{
\sum_{\xx \in \Xspace} 
\hexpec{
\Diff_\eptau (h_{\itau,\xx}) \; F(h)}
} \; .
\le 
\const \eptau^{\frac12}
\gtoas{\eptau \to 0} 0 \; .
\eeq
\end{proof}

\subsection{Recovering the $\aphi$-integral}\label{recoverasec}
Under the stronger condition that the real part of the quadratic form defined by $\ttQuad (0,\uu)$ is strictly positive, one can now go backwards and rewrite the permanent as a Gaussian integral over $\aphi$, up to an additional $h$-dependent factor that arises from the determinant of $\ttQuad$. Because $\ttQuad$ is upper triangular in the time slice index $\jtau$ and because the diagonal part is diagonal in $\xx$ as well, this determinant is
\beq\label{ttQuaddet}
\begin{split}
\det \ttQuad (h, \uu) 
&=
\pli_{\jtau, \xx} 
(1 - \I \sqrt{\eptau} \; h_{\jtau,\xx}) \; \E^{-\eptau \uu_\xx}
\end{split}
\eeq
Assume $\tKin - \uu > 0$ and recall the definition of $\DD \aphi$ in \Ref{DDadef}. By Lemma \ref{ttQuadlemma} (b), $\ttQuad(h,\uu) + \ttQuad(h,\uu)^\dagger > 0$, thus the Gaussian integral given by $\ttQuad$ is absolutely convergent and 
\beq
\int \DD \aphi \; \E^{-\bili{\bar\aphi}{\ttQuad(h,\uu) \aphi}_\XTspace}
=
\det \ttQuad (h, \uu) ^{-1}
\eeq
(see Lemma \ref{gauzz}). 
\intremark{%
so the normalized Gaussian measure is
\beq
\dd\mu_{\ttQuad(h,\uu)} (\aphi)
=
\det \ttQuad (h, \uu) \; \E^{-\bili{\bar\aphi}{\ttQuad(h,\uu) \aphi}_\XTspace} \; \DD \aphi
\eeq
The permanent formula holds for the normalized integral, i.e.\
\beq
\int
\dd\mu_{\ttQuad(h,\uu)} (\aphi)
\;
{\bili{\bar\aphi_\ntau}{\aphi_0}_\Xspace}^N
=
\Prm{N}{\Xspace} \ttCov(h,\uu)_{0,\ntau}
\eeq
Thus the unnormalized Gaussian integral still contains another determinant as a factor. This is important because the determinant of $\ttQuad$ depends on $h$ and $\uu$, hence will change the $h$-integration.
}
By (\ref{Quadperm}), 
\beq\label{unnormperm}
\Prm{N}{\Xspace} \ttCov(h,\uu)_{0,\ntau}
=
\det \ttQuad (h, \uu) \; 
\int \E^{-\bili{\bar\aphi}{\ttQuad(h,\uu) \aphi}_\XTspace} \; 
{\bili{\bar\aphi_\ntau}{\aphi_0}_\Xspace}^N
\;
\DD \aphi
\eeq

\begin{lemma}\label{planB}
Let
\beq
B(h,\uu)
=
\frac{\E^{-\I \sqrt{\eptau} \bili{1}{h}_\XTspace}}{\det \ttQuad(h,\uu)}
\eeq
Under the hypotheses of Lemma  \ref{ttilZlemma},
\beq
\hexpec{
B (h,\uu) \; 
\frac{1}{N!}
\;
\Prm{N}{\Xspace} \ttCov (h,\uu)_{0,\ntau}
}
\gtoas{\ntau \to \infty}
\tilZ^{(N,\beta,\Xspace)}
\eeq
\end{lemma}

\begin{proof}
The proof is an easy generalization of that of Lemma  \ref{ttilZlemma}, the only difference being the presence of the factor 
\beq
B(h,\uu)
=
\pli_{\jtau,\xx} \frac{\E^{-\I \sqrt{\eptau} h_{\jtau,\xx}}}{(1-\I \sqrt{\eptau} h_{\jtau,\xx}) \E^{-\eptau \uu_\xx}} \; .
\eeq
This just gives additional contributions to the discrete product rule, namely
\beq
B(h,\uu)-1
=
\sum_{\itau,\xx} 
\left(
\pli_{\jtau < \itau} \frac{\E^{\eptau \uu_\xx}}{1-\I \sqrt{\eptau} h_{\jtau,\xx}}
\right)\;
\frac{
\Diff_\eptau (h_{\itau,\xx}) + \eptau \uu_\xx + \rho (h_{\itau,\xx}, \uu_\xx)
}{1-\I \sqrt{\eptau} h_{\itau,\xx}}
\eeq
where
\beq
 \rho (h_{\itau,\xx}, \uu_\xx)
 =
(\E^{-\I \sqrt{\eptau} h_{\jtau,\xx}} -1) (\E^{\eptau \uu_\xx}-1)
+
\E^{\eptau \uu_\xx}-1 -\eptau \uu_\xx
\eeq
is analytic in $h_{\itau,\xx}$ and of order $\eptau^{\frac{3}{2}} h_{\itau,\xx} + \eptau^2 \uu_\xx^2$.
Lemma \ref{intbypartslemma} applies to the factor $\Diff_\eptau (h_{\itau,\xx}) + \eptau \uu_\xx$ in the same way as in the proof  of Lemma  \ref{ttilZlemma}. Again, the linear growth in $\ntau$ of the sum over $\itau$ is suppressed by the extra powers of $\eptau$ obtained by the integration by parts.  
\end{proof}

By (\ref{unnormperm}) and Lemma \ref{planB}, 
\beq\label{johey}
\tilZ^{(N,\beta,\Xspace)}
=
\lim_{\ntau \to \infty} 
\hexpec{
\E^{-\I \sqrt{\eptau} \bili{1}{h}_\XTspace}
\int \E^{-\bili{\bar\aphi}{\ttQuad(h,\uu) \aphi}_\XTspace} \; 
\sfrac{1}{N!}\; 
{\bili{\bar\aphi_\ntau}{\aphi_0}_\Xspace}^N
\;
\DD \aphi
}
\eeq
By (\ref{ttQuaddef}) and Fubini's theorem, $\tilZ^{(N,\beta,\Xspace)}
=
\lim_{\ntau \to \infty} 
\ZZ{\ntau}^{(N,\beta,\Xspace)}
$
with
\beq\label{joheyy}
\ZZ{\ntau}^{(N,\beta,\Xspace)}
=
\E^{-\eptau \tilde\uu N} 
\int \DD \aphi \; 
\E^{-\bili{\bar\aphi}{\ttQuad(0,\uu) \aphi}_\XTspace} \;
\sfrac{1}{N!}\; 
{\bili{\bar\aphi_\ntau}{\aphi_0}_\Xspace}^N
\hexpec{
\E^{-\I \sqrt{\eptau} \bili{1+\bar\aphi\aphi \E^{-\eptau \uu}}{h}_\XTspace}
}
\eeq
The prefactor involving $\tilde \uu = |\Xspace|^{-1} \int_\xx \uu_\xx$ is included for later convenience. It makes only an inessential change because its limit as $\ntau \to \infty$ is $1$. To remove the $\uu$-dependence from the term coupling $h$ to $\aphi$, change variables in the integral to $\aphi'_{\jtau,\xx} = \E^{-\frac{\eptau}{2} \uu_\xx} \aphi_{\jtau,\xx}$. (This also cancels the $\uu$-dependence of $\ttQuad(0,\uu)$
in the quadratic form.)
The Gaussian expectation over $h$ can now be performed and gives (dropping the primes from the notation)
\beq\label{joheyyy}
\ZZ{\ntau}^{(N,\beta,\Xspace)}
=
\E^{\ntau\eptau \int_\xx \uu_\xx-\eptau \tilde\uu N} \; 
\int \DD \aphi \; 
\E^{-\bili{\bar\aphi}{\Queps \aphi}_\XTspace 
- \frac12 \eptau \bili{(1+\bar\aphi\aphi}{\WWW (1+\bar\aphi\aphi)}_\XTspace
} \;
\sfrac{1}{N!}\; 
{\bili{\bar\aphi_\ntau}{\E^{\eptau \uu} \aphi_0}_\Xspace}^N
\eeq
with 
\beq\label{Quepsdef}
\Queps_{\jtau,\jtau'}
=
\delta_{\jtau,\jtau'} \opId 
-
\delta_{\jtau,\jtau'} \E^{- \eptau \Kin_\eptau}
\; ,
\qquad
\E^{- \eptau \Kin_\eptau}
=
\E^{\frac{\eptau\uu}{2}} \; \E^{-\eptau \tKin} \; \E^{\frac{\eptau\uu}{2}}  \; .
\eeq
Expanding the quadratic form of the interaction term then gives the integral representation with a positive interaction term, as follows.

\begin{satz}\label{a4thmj}
Let $\Queps$ be given by (\ref{Quepsdef}), with $\tKin=\Kin$ and $\uu_\xx = \frac12 \WW_{\xx\xx}$. If $\tKin - \uu > 0$, then
\beq
\ZZ{\ntau}^{(N,\beta,\Xspace)}
=
\E^{\frac{\beta}{2} \tilde v -\eptau \tilde\uu N} \; 
\int \DD \aphi \; 
\E^{-\bili{\bar\aphi}{(\Queps + \eptau \mm) \aphi}_\XTspace 
- \frac12 \eptau \bili{\bar\aphi\aphi}{\WWW \bar\aphi\aphi}_\XTspace
} \;
\sfrac{1}{N!}\; 
{\bili{\bar\aphi_\ntau}{\E^{\eptau \uu} \aphi_0}_\Xspace}^N
\eeq
with $\tilde v =\int_\xx (\WW_{\xx,\xx} - \int_\xy \WW_{\xx,\xy})$ and $\mm_\xx = \int_\xy \WW_{\xx,\xy}$. 
\end{satz}

In the translation invariant case where $\WW_{\xx,\xy} = \WW(\xx-\xy)$, $\tilde v = \WW(0) - \hat \WW (0)$, where the hat denotes the Fourier transform, i.e.\ $\hat \WW (0) = \int_\xx \WW (\xx)$, and $\mm_\xx = \hat \WW (0)$ is independent of $\xx$. Moreover, $\uu_\xx =\frac12 \WW (0) = \tilde \uu$, so  $\bili{\bar\aphi_\ntau}{\E^{\eptau \uu} \aphi_0}_\Xspace^N = \E^{\eptau \tilde \uu N}\bili{\bar\aphi_\ntau}{ \aphi_0}_\Xspace^N$, so that the prefactor involving $\tilde \uu$ gets cancelled. It is now convenient to use the notations $\aphi (\tau,\xx) = \aphi_{\jtau,\xx}$, relabelling to a discrete time index $\tau = \eptau \jtau= \frac{\beta}{\ntau} \jtau \in \eptau \{0, \ldots, \ntau\}$, 
\beq
(\del_\tau \aphi)(\tau,\xx) 
=
\frac{1}{\eptau} (\aphi (\tau +\eptau, \xx) - \aphi (\tau,\xx))
=
\frac{1}{\eptau} (\aphi_{\jtau+1,\xx} - \aphi_{\jtau,\xx})
\eeq
and (as already introduced) $\eptau\sum_\jtau = \int_\tau$. Rewriting all terms in this notation then gives Theorem \ref{a4thm}.

The integral representations involving the variants of the kinetic term in the action stated in Theorem \ref{a4thm} right after (\ref{31}) are obtained similarly, by using $\tttQuad$ and $\tttaQuad$ instead of $\ttQuad$. By Lemmas \ref{tttQuadlemma}, \ref{posplusim} and \ref{tttaQuadlemma},  the proofs that the limit $\ntau \to \infty$ exists and is unchanged for these modified actions are straightforward adaptations of the ones given for $\ttQuad$. 

\section{Remarks about the Grand-Canonical Ensemble}\label{grandsec}

Here I briefly outline the connection to the grand canonical ensemble. By definition (see (\ref{grancanZ})), the grand canonical partition function is the Laplace transform of $N \mapsto \canZ^{(N \beta, \Xspace)}$. 
The setup of the canonical ensemble has a few technical advantages, in particular, on a lattice with nonzero spacing, the Hamiltonian, as well as its kinetic and interaction parts, become bounded operators when restricted to the $N$-particle space, so that one can regard fixing particle number as a natural regularization for defining bosonic models mathematically. One can then also try to study the grand canonical ensemble based on results for the canonical one. In the following, I discuss a few significant points where the ensembles differ, and how things fit together. 

\subsection{The periodic boundary condition in time}
The integral over the $\aphi$ fields in the canonical ensemble does {\em not} have a periodic boundary condition in the Gaussian: there is no matrix element in $\Quad$ that couples $\jtau=\ntau$ and $\jtau=0$. The reason for this is explicit from \ref{sandprich} -- the projection $\Proj_N$ removes the exponential, replacing it by the $N$'th order term in its expansion. It is therefore plausible that the exponential, and hence the periodic boundary condition, will be reinstated by the sum over $N$ in (\ref{grancanZ}). But there are of course convergence issues when going beyond formal considerations. When (\ref{HSFI}) is inserted into (\ref{grancanZ}) and the sum over $N$ is formally exchanged with the limit and the integral, one gets 
\beq\label{periseri}
\sum_{N=0}^\infty \frac{\E^{N\beta\mu}}{N!} \; {\bili{\bar\aphi_\ntau}{\aphi_0}_\Xspace}^N
=
\E^{\E^{\beta\mu}\; \bili{\bar\aphi_\ntau}{\aphi_0}_\Xspace}
\eeq
This series converges for all values of $\beta\mu$, $\aphi_\ntau$ and $\aphi_0$. But the additional term created in the exponent spoils positivity of the real part of the quadratic form, hence leads to a divergent Gaussian integral, if $\mu > 0$. This is explained in more detail in Appendix \ref{Covarapp}.

The appearance of a factor $\E^{\beta\mu}$ only in the term coupling the time slices $\jtau$ and $0$ is also unusual. The standard combination $\Kin - \mu \opId$ can be obtained simply by rearranging the product 
\beq
\E^{\beta \mu N} \tilZ^{(N,\beta,\Xspace)}
=
\Tr_{\FockB} \left[
\E^{-\beta (\opH -\mu\opN)}\right]
\eeq
similarly to what is done in (\ref{grancanZ}), which changes $\Kin$ to $\Kin_\mu = \Kin - \mu \opId$ everywhere. That the sum over $N$ in (\ref{grancanZ}) converges for $\mu < 0$ and $\Kin \ge 0$ is then a simple consequence of the nonnegativity of the interaction term $\WW$. If $\WW > 0$, the condition $\mu < 0$ can be dropped. 

\subsection{The covariance} 
The covariance of the canonical ensemble is upper triangular because of the absence of a periodic boundary condition in time in the operator $\Quad$.
With the just described absorption of $\mu$ into the kinetic term $\Kin_\mu$, the summation over $N$ gives 
\beq\label{perisereins}
\sum_{N=0}^\infty \frac{1}{N!} \; {\bili{\bar\aphi_\ntau}{\aphi_0}_\Xspace}^N
=
\E^{\bili{\bar\aphi_\ntau}{\aphi_0}_\Xspace} 
\eeq
and (\ref{HSFI}) gets replaced by 
\beq\label{granHSFI}
\begin{split}
Z_g^{(N,\beta,\Xspace,\ntau)} 
&=
\int \dd \mu_{\WWW} (h) 
\int_{\bC^{\XTspace}} \DD \aphi \; 
\E^{-\bili{\bar\aphi}{\Kuad (h) \; \aphi}_\XTspace}\;
\end{split}
\eeq
with 
\beq\label{Kuaddef}
\Kuad (h)_{\jtau,\jtau'} = \Quad (h)_{\jtau,\jtau'}  + \delta_{\jtau,\ntau} \delta_{0,\jtau'} \opId \;.
\eeq 
(Here $\E^{-\eptau \Kin_\mu}$ appears in $\Quad$.) 
As shown in Appendix \ref{Covarapp}, the covariance is given for $h=0$ by the standard time-ordered Green function for free bosons, eq.\ (\ref{granolcov}).

\begin{satz}\label{grancanthm2}
Assume that $\Kin_\mu > 0$. Then the integral $(\Ref{granHSFI})$ is absolutely convergent. The grand-canonical covariance $\Gov (h) = \Kuad(h)^{-1}$ exists for all $h$, and has a norm bounded uniformly in $h$. It is analytic in $h$ if $|\Im h_{\jtau,\xx}| < \sqrt{\eptau} e_{min}$, where $e_{min}$ is the smallest eigenvalue of $\Kin_\mu$. If $\Kin_\mu$ generates a stochastic process, then for all $\jtau'\ge\jtau$, all $\xx,\xx'\in \Xspace$, and all $h$, the covariance $\Gov (h) = \Kuad(h)^{-1}$ for the grand canonical ensemble satisfies the bound 
\beq\label{grand-uniformed}
\forall h \in \bR^\XTspace: \quad
\abs{\Gov (h) _{(\jtau,\xx), (\jtau',\xx')}}
\le
\Gov (0) _{(\jtau,\xx), (\jtau',\xx')} \; .
\eeq
\end{satz}

\begin{proof}
This proof is an application of Lemma \ref{InversesLemma}, and I will use the notations introduced in Appendix \ref{Covarapp}. Let $\opA_\jtau = \E^{-\eptau \Kin_\mu} \E^{\I \sqrt{\eptau} h_\jtau}$. Because $\Kin_\mu > 0$, 
$\norm{\opA_\jtau} = \norm{\E^{-\eptau \Kin_\mu}} < 1$ uniformly in $h$, and $\opA_0 = \opId$, hence 
$\opP= \opA_0 \ldots \opA_\ntau$ has norm $\norm{\opP} < 1$.
By Lemma \ref{InversesLemma}, the inverse $\Gov$ exists, is given by (\ref{Kinverse1}), and has norm bounded uniformly in $h\in \bR^\XTspace$ because the same holds for all factors appearing in (\ref{Kinverse1}). In particular, the inverse $(\opId - \opP)^{-1}$ is given by a convergent geometric series. Expanding out this series gives a linear combination of products of $\opA_\jtau$'s, with nonnegative coefficients. Thus the bound (\ref{grand-uniformed}) follows by straightforward adaptation of the proof of Lemma \ref{Quadlemma}. Analyticity in the $\eptau$-dependent neighbourhood of $\bR^\XTspace$ holds because the norm of 
$\E^{-\eptau \Kin} \E^{\I \sqrt{\eptau} h}$ remains strictly smaller than $1$ on that set. 
\end{proof}

Thus, all arguments based on the uniform bounds for the covariance, integration by parts, etc., have their analogues for the grand canonical ensemble, and similar theorems about convergence of integrals with modified actions hold, in particular, there is a representation with an $|\aphi|^4$ interaction term in the action. (I omit the statements and proofs here to avoid further repetition.)  As mentioned in Section \ref{resultssec} for the the canonical ensemble, positivity of the interaction, i.e.\ $\WW> 0$, implies that the integral (\ref{a4integral}) converges absolutely also beyond $\mu =0$ and it defines an analytic function of $\mu$. By the identity theorem, the continuation represents the partition function for all $\mu$. A similar statement can be proven for the grand canonical ensemble. 

\subsection{The determinant} 
A further interesting and striking difference between the ensembles is the role played by the (inverse of the) determinant arising from the Gaussian $\aphi$-integral. In the canonical case, the determinant is simply equal to one. In the grand canonical case, the determinant differs from one, and the expansion of $\ln \det \Kuad (h)$ in powers of $h$ creates the well-known expansion in terms of loops corresponding to virtual particle-antiparticle pairs with higher and higher order interactions. These virtual pair creation processes do not happen in the canonical ensemble because there is a restriction to fixed $N$. Thus the constraint of fixed particle number shows up explicitly as the determinant being one. 

Moreover, there is no backward-in-time propagation in the canonical ensemble because the covariance is upper triangular in time. By contrast, the grand canonical covariance $\Gov (h)$ is not upper triangular. It is simply the time ordered boson two-point function in the external field $h$, hence allows forward and backward propagation. All these features make explicit that fluctuations are different in the canonical and grand canonical ensemble. 

Finally, there is another, more technical point. The inverse of the determinant in the grand canonical ensemble does not factorize. The permanent in the canonical ensemble factorizes over cycles of permutations, and this makes a rather simple polymer expansion, in which the bounds on the covariances can be used directly, in analogy to the proofs given in Section \ref{Analysis2} possible \cite{toappear}. 

\section{The Effect of a Condensate in the Canonical Ensemble}\label{physsec}

In his famous work \cite{Bog}, Bogoliubov treated condensation of the weakly interating Bose gas by an ansatz in which the creation and annihilation operators of the kinetic-energy zero mode are replaced by commuting operators, and diagonalizing the resulting effective Hamiltonian by a canonical transformation, the Bogoliubov transformation. He was thus able to derive the linear spectrum of low-energy excitations of the condensate, which is known to be experimentally accurate. Even the elements of the Bogoliubov transformation could be identified in experiments \cite{Ketterle}. 

At sufficiently low temperatures $\beta^{-1}$, weakly interacting bosons are thus expected to condense into a state resembling the lowest-energy state of the free boson system. The latter is determined by the zero mode of the one-particle kinetic term which, in a translation-invariant system, is spatially constant. In the FIR, this corresponds to splitting the integration variable into a `condensate field' which is integrated over a low-dimensional subspace of field space, and all modes orthogonal to it. Condensation as a phase transition then corresponds to the condensate field acquiring a nonvanishing expectation value, which spontaneously breaks the $U(1)$ symmetry of particle number conservation. In the FIR, this is the symmetry of action and integration measure under $\aphi_{\tau,\xx} \mapsto \E^{\I \theta} \; \aphi_{\tau,\xx} $ (where $\theta$ is independent of $\tau$ and $\xx$). In the following, I do this field decomposition as a change of variables in the $\aphi$ integral, without, for the moment, making the assumption that the zero mode field is also constant in time (which is the usual zeroth order approximation for expansions). This decomposition can be applied irrespective of what ensemble is studied, but again, I will focus on the canonical ensemble here. It also makes only a technical difference which of the discretizations discussed above one uses; for convenience of presentation, I take the action with the time-local $|\aphi|^4$ term and the quadratic form $\tttaQuad$ for the kinetic term.

Assume that $\Kin$ is given by the discrete Laplacian. Then $\Kin$ has a nontrivial kernel, namely the constant functions on $\Xspace$. By (\ref{fixy}), $\Kin \ge 0$ can be changed to $\Kin + \hebeps \opId$ with $\hebeps > 0$, at the expense of having an explicit factor $\E^{\beta N \hebeps}$ in front of the partition function. Assume, moreover, that $\WW$ is strictly positive, $\WW > 0$. Then the integral (\ref{a4integral}) converges absolutely for any $\hebeps$, even $\hebeps < 0$, and consequently, the limit $\hebeps \to 0$ can be taken. It is thus possible to proceed by introducing $\hebeps$. Because $\Kin_\hebeps > 0$, Theorem \ref{a4thm} applies, and one can also postpone taking the limit $\hebeps \to 0$ to a convenient point. Thus the integral is
\beq
\tttaZ{\ntau}^{(N,\beta,\Xspace)}
=
\int \DD \aphi \; 
\E^{-\tttaS_\XTspace (\aphi)} \; 
\sfrac{1}{N!} \; 
{\bili{\bar\aphi(\beta)}{\aphi (0)}_\Xspace}^N
\eeq
with
\beq
\tttaS_\XTspace (\aphi)
=
\int_{\tau,\xx}
\bar\aphi (\tau,\xx) \; 
\lbrack (-\del_\tau + \mm + \Kin_\hebeps)\aphi\rbrack (\tau,\xx)
+
\frac12
\int_\tau
\int_{\xx,\xy} 
\abs{\aphi(\tau,\xx)}^2 \; \WW(\xx-\xy) \; \abs{\aphi(\tau,\xy)}^2 
\eeq
and $\mm = \hat \WW (0)$. For simplicity of presentation, $\WW_{\xx,\xy} = \WW \delta_{\xx,\xy}$ (at $\speps =1$) is taken in the following. Then $\hat \WW (0) = \WW$ as well.

\subsection{Orthogonal decomposition}
The fundamental equation of it all is 
\beq
\aphi = \bphi + \cphi
\eeq
where $c = (c_\tau)_{\tau \in \bT}$ is independent of $\xx$ and 
\beq
\bili{1}{\bpht{\tau}}_\Xspace
=
\int_{\xx} \bphi(\tau,\xx) 
=
0 \; . 
\eeq
This is an orthogonal decomposition: $\bili{\cpht{\tau}}{\bpht{\tau}}_\Xspace = 0$ for all $\tau$. 
With this, 
\beq
\DD \aphi 
=
\DD^\bT \cphi \; 
\DD' \bphi
\qquad
\DD^\bT \cphi 
=
\pli_{\tau\in \bT} \sfrac{\dd \overline{\cpht{\tau}} \wedge \dd \cpht{\tau}}{2\pi\I} 
\qquad
\DD' \bphi 
=
\DD \bphi\; 
\pli_{\tau \in \bT} \delta \left(\int_\xx \bphi (\tau,\xx)\right) 
\eeq
and 
\beq
\tttaS (\aphi)
=
\tttaA_0 (\cphi)
+
\tttaA_2 (\bphi,\cphi)
+
\tttaA_3 (\bphi,\cphi)
+
\tttaA_4 (\bphi)
\eeq
with
\beq
\begin{split}
\tttaA_0 (\cphi)
&=
\abs{\Xspace} 
\int_\tau
\lbrack
\overline{\cpht{\tau}} (-\del_\tau + \hebeps + \mm) \cpht{\tau}
+
\frac{\WW}{2} \abs{\cpht{\tau}}^4
\rbrack
\\
\tttaA_2 (\bphi,\cphi)
&=
\frac{\WW}{2}
\int_{\tau,\xx}
\left\lbrack
4 \abs{\cpht{\tau}}^2 \abs{\bphi (\tau,\xx)}^2 
-
\cpht{\tau}^2 \overline{\bphi (\tau,\xx)}^2 
-
\overline{\cpht{\tau}}^2 \bphi (\tau,\xx)^2 
\right\rbrack
\\
&+
\bili{\bar\bphi}{(-\del_\tau+ \Kin_\hebeps + \mm)\bphi}_\XTspace
\\
\tttaA_3 (\bphi,\cphi)
&=
\WW
\int_{\tau,\xx}
\abs{\bphi (\tau,\xx)}^2\;
\left(
\cpht{\tau}\; \overline{\bphi(\tau,\xx)}
+
\overline{\cpht{\tau}}\;  \bphi(\tau,\xx)
\right)
\\
\tttaA_4 (\bphi)
&=
\frac{\WW}{2}
\int_{\tau,\xx} \abs{\bphi (\tau,\xx)}^4 \; . 
\end{split}
\eeq
The contribution to the action $\tttaS$ that is linear in $\bphi$ vanishes because of the orthogonality of $\bphi$ and $\cphi$. For the same reason, 
\beq
\bili{\bar\aphi(\beta)}{\aphi (0)}_\Xspace
=
\bili{\overline{\cpht{\beta}}}{\cpht{0}}_\Xspace
+
\bili{\bar\bphi_\beta}{\bphi_0}_\Xspace
=
\abs{\Xspace} \overline{\cpht{\beta}} \; \cpht{0} 
+
\bili{\bar\bphi_\beta}{\bphi_0}_\Xspace
\eeq
and a binomial expansion gives 
\beq\label{cangrancan}
\tttaZ{\ntau}
=
\abs{\Xspace}^N
\sum_{K=0}^N 
\sfrac{1}{(N-K)!} \; 
\int \DD^\bT \cphi\; \E^{-\tttaA_0(\cphi)} \; 
\left(\overline{\cpht{\beta}} \; \cpht{0} \right)^{N-K} \; 
\tttaY{\ntau}^{(\beta,K,\Xspace)} (\cphi)
\eeq
with
\beq
\tttaY{\ntau}^{(\beta,K,\Xspace)} (\cphi)
=
\sfrac{1}{K!}
\int \DD'\bphi \; 
\E^{- \tttaA_2 (\bphi,\cphi) - \tttaA_3 (\bphi,\cphi) - \tttaA_4 (\bphi)} \; 
\left(
\frac{1}{|\Xspace|}
\int_\xx \overline{\bphi(\beta,\xx)} \bphi (0,\xx)
\right)^K
\eeq
The sum over $K$ in (\ref{cangrancan}) suggests that the condensate now plays the role of a particle reservoir --- at fixed $\cphi$, the ensemble described by the $\bphi$ integral is no longer at fixed $N$, but rather, all particle numbers $K \in \{ 0, \ldots, N\}$ contribute to the partition function. To be a true reservoir in a thermodynamical sense, the condensate fluctuations would have to be so small that it becomes approximately constant also in time $\tau$, and that the likeliest values of $K$ are small compared to $N$. This is expected in the thermodynamic limit at sufficiently large $\beta$ and with $N$ growing together with $|\Xspace|$ so that $\rho=\frac{N}{|\Xspace|}$ is fixed and large enough. However, in the present setup, $\cphi$ may still receive important contributions from the $\bphi$-integral, especially because of the slow decay of the $\bphi$-covariance (see Section \ref{constcond} below). 

\subsection{Quadratic Form for the $\bphi$ Fields}
By definition of the time derivative and by its boundary condition (see Theorem \ref{a4thm} and (\ref{tttaQuaddef})), the operator $\boQ(\cphi)$ defining the action $\tttaA_2$ as a quadratic form in $\bphi$ is still upper triangular in the time slice indices $\jtau$ and $\jtau'$. Explicitly, in the notation where times $\tau= \eptau \jtau$ are labelled by $\jtau \in \{ 0, \ldots, \ntau\}$,
\beq
\tttaA_2 (\bphi,\cphi)
=
\sum_{\jtau=0}^\ntau
\bili{(\bphj{\jtau},\overline{\bphj{\jtau}})}{%
\boQloc (\cphj{\jtau}) \;
(\overline{\bphj{\jtau}},\bphj{\jtau})^\top}_\Xspace 
- \sum_{\jtau=0}^{\ntau-1}
\bili{\overline{\bphj{\jtau}}}{\bphj{\jtau+1}}
\eeq
with
\beq\label{boQlocdef}
\boQloc (\cphi)
=
1_2 \otimes
(\opId+ \eptau \mm + \eptau \Kin)
+ \eptau \WW
\;
\left[
\begin{array}{cc}
|\cphi|^2 & \cphi^2 \\
\overline{\cphi}^2 &  |\cphi|^2
\end{array}
\right]
\otimes
\opId
\eeq
The $\cphi$-dependent matrix appearing in (\ref{boQlocdef}) is obviously nonnegative (but does have an eigenvalue zero). The determinant of $\boQloc (\cphj{\jtau})$ is no longer equal to one. However, and this is again peculiar to the canonical ensemble, $\det \boQloc(\cphi) = 1 + O(\eptau)$, and the matrix $\boQ$ corresponding to the quadratic form $\tttaA_2$ in $\bphi$ is therefore invertible if $\eptau$ is small enough.

\subsection{Time-independent condensate field}\label{constcond}
In the canonical partition function, there is no chemical potential that can be tuned to give a mexican hat potential. Instead it is the power (\ref{sandprich}) that introduces a tendency towards a nonzero value of $\cphi$. A first understanding of this can be gained by assuming that $\cphi$ is constant in time, $\cpht{\tau} = \cpht{0} =\cphi_0$ for all $\tau$, and taking the contribution from $K=0$ only. Then $\overline{\cpht{\beta}} \cpht{0} = |\cphi_0|^2$. Fixing $\rho = \frac{N}{|\Xspace|}$ means that $|\Xspace| = \sigma N$ with $\sigma = \frac1\rho$. Then for large $N$, the factor
\beq
\frac{|\Xspace|^N |\cphi_0|^{2N}}{N!}
\approx
\frac{(\sigma|\cphi_0|^{2}\E)^N}{\sqrt{2\pi N}}
=
\frac{1}{\sqrt{2\pi N}} \; \E^{N (1+ \ln (\sigma|\cphi_0|^{2}))}
\eeq
in the integrand drives the maximum of the integrand in $\cphi_0$, hence also the most probable value of $\cphi_0$, away from zero. 
The above replacements amount to considering the simplified function 
\beq
\tttaZ{\ntau}
=
\sfrac{1}{\sqrt{2\pi N}} 
\int \sfrac{\dd \bar\cphi_0 \wedge \dd \cphi_0}{2\pi\I} \;
\E^{-N F(|\cphi_0|^2)} 
\;
\tttaY{\ntau}^{(\beta,0,\Xspace)} (\cphi_0) 
\eeq
with
\beq
F(\gamma)
=
\sigma \left[ (1+ \beta \mm) \gamma + \frac{\WW\beta}{2} \gamma^2\right]
- 1 - \ln (\sigma \gamma)
\eeq
The $1$ term in the square brackets is due to the boundary condition for the discretized time derivative which implies that its integral does not evaluate to zero on a constant field. Explicitly, when inserting a $\jtau$-independent field into (\ref{jbouc}) (with $\Kin =0$ and $h=0$): the difference of the two sums does not vanish, but gives $|\Xspace| |\cphi_0|^2 = N \sigma \gamma$. As discussed, the logarithm in $F$ comes from the power in the integral for the canonical partition function, and it implies that $F$ takes its minimum at a $\gamma_0 \ne 0$. As mentioned, the true ensemble average is more complicated than this. In particular, $\rho_0 = \frac{N-K}{|\Xspace|}$ is the condensate density, and whether condensation really takes place at a given temperature is a question of the distribution of $K$'s. 

\subsection{Positivity in the continuum limit}
If $\cphi_0 \ne 0$, the eigenvalues of the matrix $\boQ$ exhibit the typical Bogoliubov spectrum, leading to a kinetic term $E_B (\xp) = |\xp| \sqrt{\bogv^2 + \xp^2}$
where $\bogv^2 = 2\WW |\cphi_0|^2$. (Here I have taken a spatial continuum limit as well.)

\begin{lemma}
For all $t \ge 0$, the $x$-space kernel of $\E^{ - t |\xp| \sqrt{2\mu + \xp^2}}$ is positive. 
\end{lemma}

\begin{proof}
Let 
\begin{equation}
\hat F (\xp) = |\xp| \sqrt{\bogv^2 + \xp^2} - \xp^2 - \frac{\bogv^2}{2} ,
\end{equation}
Then
\begin{equation}
\hat F(\xp) = 
- \frac{v^4}{2(|\xp| + \sqrt{\bogv^2 + \xp^2})^2} 
= 
- \frac{v^4}{2}
\int_0^\infty  t \; \dd t\; \E^{-t (|\xp| + \sqrt{\bogv^2 + \xp^2})}
\end{equation}
The Fourier transform $G_t$ of $\E^{-t (|\xp| + \sqrt{\bogv^2 + \xp^2})}$
is the convolution of that of $\E^{-t |\xp|}$ and that of 
$\E^{-t \sqrt{\bogv^2 + \xp^2}}$, which are both strictly positive functions, 
hence strictly positive, so that 
\begin{equation}
F(\xx-\xy) = - \int_0^\infty t \; \dd t\; G_t(\xx-\xy) < 0 
\end{equation}
for all $\xx$ and $\xy$. By iterated convolution, $\E^{-\tau F} (\xx,\xy)$
is positive for all $\xx$ and $\xy$,
and the same holds, again by convolution,
for $\E^{-\tau E_B} (\xx,\xy) $.
\end{proof}

The decay is, however, slow since the kernels that get convolved only have power law 
decays. Specifically \cite{SteiWei},
\begin{equation}
\E^{-t|\xp|} (\xx,\xy)
=
\left\lbrace
\begin{array}{cc}
\sfrac{1}{\pi^2}\;
\frac{t}{(t^2+|\xx-\xy|^2)^2}
& \mbox{ for } d=3 \\[2ex]
\Gamma(\sfrac{d+1}{2})\, \pi^{-\frac{d+1}{2}} \; 
\frac{t}{(t^2+|\xx-\xy|^2)^{\frac{d+1}{2}}}
& \mbox{ for } d > 3\; .
\end{array}
\right.
\end{equation}
This slow decay suggests that a stochastic representation as a random walk in a fluctuating condensate background has long jumps, hence differs from the one corresponding to Brownian motion. The cubic term $\tttaA_3$ also leads to branching and coalescence, corresponding to processes where one of two scattering particles emerges from, or gets absorbed in, the condensate.

\intremark{At fixed $t$, this function is integrable in $\xx-\xy$, but 
\beq
\int_0^\infty \dd t\; \E^{-t|\xp|} (0,\xx)
\sim
\frac{1}{|\xx|^{d-1}}
\eeq
is not integrable in $x$. It is in $L^p(\bR_+ \times \bR^3)$ for all $p > \frac{d+1}{d}$.}

\section{Conclusion and Outlook}\label{conclusec}

The coherent-state and auxiliary-field integrals for the boson system derived here provide a simple regularized version of the standard formal time-continuum functional integrals used in condensed-matter physics. The integral for the canonical partition function and expectation values exhibits a few significant differences to the usually considered grand-canonical case: there is no periodic boundary condition in time in the quadratic form defining the Gaussian integral, propagation is only in forward time direction, and the inverse determinant from Gaussian integration is equal to one. 
The integrals also allow to derive a straightforward random-walk expansion, which can be viewed as a regularized version of the Brownian motion representation used in \cite{AdamsKoenig, AdamsBruKoenig, AdamsCollevecchioKoenig}. 

It is often stated that the Wick rule fails for the canonical expectation values of products of field operators. The integral representations for the canonical expectation values provide an easy way to understand this -- in Gaussian integrals, the Wick rule is derived by integration by parts, and the factor $\bili{\bar \aphi_{\ntau}}{\aphi_0}_\Xspace$ that arose from the projection to the $N$-particle space $\FckoB{N}$, see Lemma \ref{sandprichlemma}, gives rise to additional terms in that procedure.
In presence of an operator insertion $\opF$, this factor gets modified, see, e.g.\ (\ref{fandprick}), but it remains a polynomial in the fields.  It is easy to derive a modified Wick rule using the standard  formalism of field Laplacians for representing  Gaussian expectations (see Appendix \ref{detpermlapapp} and \cite{msbook}). 

At the heart of the technical results of this paper are the pointwise bounds, uniformly in $h$, for the regularized covariances, which are similar to those used in \cite{MagRiv} and subsequent works on the loop vertex expansion. Although it turns out that they are very simple to prove (see Lemma \ref{Quadlemma} for an example), it was not completely obvious a priori that they would be available in the discrete-time regularization. Above, I have used them to show that there is a considerable freedom in chosing the discretized action without changing the time-continuum limit of the partition function. Although this is expected, the point here is that the proofs require only very elementary bounds, such as the combination of resolvent equations like (\ref{resolventQ}), Taylor expansion like in (\ref{TaylorQ}) and integration by parts as in (\ref{intbypartsWWW}. Moreover, these tehniques can be used to derive convergent expansions for expectation values, which allow to take the infinite-volume limit \cite{toappear}. 

One should not jump to the conclusion that this already makes more involved machinery, such as a decomposition in large and small fields, unnecessary in the analysis of bosonic many-body problems, however. If renormalization becomes necessary, e.g.\ when regularizing determinants or performing any operation in which expansions in the fields are required, it is not yet clear whether uniform bounds will apply.

The bulk of this paper has dealt with the canonical ensemble, but (as discussed in some detail in Sections \ref{resultssec} and \ref{grandsec}), many essential bounds carry over to the grand-canonical case. Moreover, one can also take the approach to use results for the canonical ensemble to make statements about the grand canonical one, by summation over $N$, provided one has good enough control over the $N$-dependence. 

Finally, I have given a brief discussion coming back to the original motivation, namely the connection to stochastic representations. The energy of excitations in a Bose condensate is no longer quadratic, but linear in momentum, $E(p) = c |p|$, with $c$ the velocity of sound. Nevertheless, in the formal time-continuum limit, this kinetic term still generates a stochastic process (albeit with long-range jumps). It will be interesting to see if the uniform bounds derived in this paper, for which this positivity was crucial, continues to hold in a regularized version of the  functional integral. 

\appendix

\section{Coherent-State Formulas} 
\label{cohstateapp}
In terms of the orthonormal basis (\ref{FockONB}), the expansion of the coherent state is 
\beq
\coh{\aphi}
=
\sli_{\nu \in \bN_0^\Xspace} 
\ee_\nu \; 
\pli_{\xx \in \Xspace}
\frac{\aphi_\xx^{\nu_\xx}}{\sqrt{\nu_\xx!}} 
\eeq
which is norm convergent for any $\aphi \in \bC^\Xspace$ by inspection. Eq.\ (\ref{cohinp}) then follows by inserting the expansion. 

By the commutation relations, $\aann_\xx {\acre_\xx}^n \Omega = n \speps^{-d} {\acre_\xx}^{n-1} \Omega$. Thus the vector $\ww_{n,\xx} (\aphi)  = \sum_{\nu =0}^n \frac{\speps^d \aphi^\nu}{\nu!} {\acre_\xx}^\nu \Omega$ satisfies 
\beq
\aann_\xx \ww_{n,\xx} (\aphi)
=
\aphi_\xx \; \ww_{n-1,\xx} (\aphi) \; ,
\eeq
so $\coh{\aphi}^{(N)} = \sli_{n=0}^N
\frac{{\bili{\aphi}{\acre}_\Xspace}^n}{n!} \Omega$ satisfies
\beq
\lim_{N \to \infty}
\aann_\xx 
\coh{\aphi}^{(N)} 
=
\aphi_\xx \; \coh{\aphi} \; .
\eeq
Since the exponential series for $\coh{\aphi}$ is norm convergent, $\Vert\coh{\aphi}^{(N)} - \coh{\aphi}\Vert \to 0$ for $N \to \infty$. For all $\xx \in \Xspace$, the operators $\aann_\xx$ and $\acre_\xx$ have the common dense domain $\cD =\cD_{\sqrt{\opn_\xx}}$, and $\coh{\aphi} \in \cD$. Let $\ww \in \cD$. Then 
\beq
\begin{split}
\bracket{\ww}{\aann_\xx \coh{\aphi}}
&=
\bracket{\acre_\xx \; \ww}{\coh{\aphi}}
=
\lim_{N\to \infty}
\bracket{\acre_\xx \; \ww}{\coh{\aphi}^{(N)}}
=
\lim_{N\to \infty}
\bracket{\ww}{\aann_\xx \coh{\aphi}^{(N)}}
\\
&=
\bracket{\ww}{\aphi_\xx \coh{\aphi}}
\end{split}
\eeq
Since $\cD$ is dense in $\FockB$, (\ref{coheig}) holds. 

To see (\ref{resid}), let $\psi\in \FockB$. Then the basis coefficients 
$\psi_\nu = \bracket{\ee_\nu}{\psi}$ are square summable over $\nu \in \bN_0^\Xspace$, and the sum 
\beq
\ketbra{\unicoh{\aphi}}{\unicoh{\aphi}}\psi \rangle
=
\sum_{\nu,\nu' \in \bN_0^\Xspace}
\E_\nu \; 
\psi_{\nu'} 
\pli_\xx \frac{\bar\aphi_\xx^{\nu'_\xx}\aphi_\xx^{\nu_\xx}}{\sqrt{\nu_\xx! \nu'_\xx!}}
\eeq
converges absolutely and uniformly on $\bC_\Rad^{|\Xspace|}$, hence can be exchanged with the $\aphi$-integral. 
Writing $\aphi_\xx = \rho_\xx \E^{\I \theta_\xx}$, 
\beq
\frac{1}{2\pi\I} \dd \bar\aphi \wedge \dd \aphi 
=
\dd (\rho_\xx^2) \; 
\frac{\dd \theta_\xx }{2\pi}
\eeq
The integral on the right hand side of (\ref{spurli}) factorizes into a product over $\xx$ of 
\beq
\int_0^\Rad  \dd (\rho_\xx^2) \E^{-\rho_\xx^2} \; \rho_\xx^{\nu_\xx+\nu'_\xx}
\int\frac{\dd\theta_\xx}{2\pi}\; \E^{\I \theta (\nu_\xx - \nu'_\xx)}
\eeq
The $\theta$-integral vanishes unless $\nu_\xx = \nu'_\xx$. The $\rho_\xx$-integral is an incomplete Gamma function $\Gamma (\nu_\xx+1, \Rad^2)$, which increases to $\nu_\xx!$ as $\Rad \to \infty$. Thus 
\beq
P_\Rad (\psi)
=
\ili_{{\bC_\Rad}^\Xspace}
\dd^\Xspace\aphi \; 
\ketbra{\unicoh{\aphi}}{\unicoh{\aphi}}\psi \rangle
=
\sum_{\nu}
\psi_\nu f_\nu(\Rad) \; \E_\nu 
\eeq
where 
\beq
f_\nu (\Rad)
=
\pli_{\xx}
\frac{\Gamma(\nu_\xx+1, \Rad^2)}{\nu_\xx!} \in (0,1)
\eeq
Thus $S_\Rad = \sum_\nu \abs{\psi_\nu (1-f_\nu (\Rad))}^2 $ converges, and
and $f_\nu (\Rad) \to 1$ as $\Rad \to \infty$. By the dominated convergence theorem, applied to the sum over $\nu$, $S_\Rad \to 0$ as $\Rad \to \infty$. 
Since the $\E_\nu$ are an ONB, $S_\Rad = \norm{\psi  - P_\Rad (\psi)}^2$, so (\ref{resid}) holds. 

To see (\ref{spurli}), use that the series expansion
\beq
\bracket{\unicoh{\aphi}}{\opA \; \unicoh{\aphi}}
=
\sum_{\nu,\nu'}
\bracket{\E_{\nu'}}{\opA \; \E_{\nu}}
\pli_\xx \frac{\bar\aphi_\xx^{\nu'_\xx}\aphi_\xx^{\nu_\xx}}{\sqrt{\nu_\xx! \nu'_\xx!}}
\eeq
converges absolutely and uniformly for $\aphi \in \bC_\Rad^\Xspace$, because $\opA$, as a trace class operator, is bounded hence $|\bracket{\E_{\nu'}}{\opA \; \E_{\nu}}| \le \norm{\opA}$. So the integral over $\aphi$ can be exchanged with the summation. By the same arguments as above, the sum reduces to a single summation over $\nu$, and 
\beq
\ili_{{\bC_\Rad}^\Xspace}
\dd^\Xspace\aphi \; 
\bracket{\unicoh{\aphi}}{\opA  \; \unicoh{\aphi}}
=
\sum_\nu
f_\nu (\Rad) \; 
\bracket{\E_{\nu}}{\opA \; \E_{\nu}}
\eeq
Because $\opA$ is trace class, the sum $\sum_\nu \abs{\bracket{\E_{\nu}}{\opA \; \E_{\nu}}}$ converges. Thus, again by the dominated convergence theorem, (\ref{spurli}) follows. 

Eq.\ (\ref{expexp}) is an immediate consequence of the following Lemma. 

\begin{lemma}
Let $\tau \ge 0$ and $\opH_0 = \bili{\acre}{\Kin \aan}_\Xspace$ with $\Kin \ge 0$. Then 
\beq
\E^{-\tau \opH_0} \coh{\aphi}
=
\coh{\E^{-\tau \Kin} \aphi}
\eeq
\end{lemma}

\begin{proof}
Because $\Kin^\dagger = \Kin$, there is a unitary transformation $\cU$ so that $\Kin = \cU^\dagger \cD \cU$, where the diagonal matrix $\cD$ has the eigenvalues $E_\alpha$ of $\Kin$ as diagonal entries. By hypothesis, $E_\alpha \ge 0$ for all $\alpha$.  Set $\tilde \opa = \cU \opa$, then the $\tilde \opa_\alpha$ also have canonical commutation relations, $[\tilde\opa_\alpha, \tilde\opa^\dagger_{\alpha'}] = \speps^{-d} \delta_{\alpha,\alpha'}$ and $[\tilde\opa_\alpha, \tilde\opa_{\alpha'}] = 0$, and $\opH_0 = \speps^d \sum_\alpha E_\alpha \tilde \opn_\alpha$ with $\tilde \opn_\alpha = \tilde\opa^\dagger_\alpha \tilde\opa_\alpha$. The $\tilde \opn_\alpha $ all commute, and thus $\E^{-\tau \opH_0} = \prod_\alpha \E^{-\tau  \speps^d E_\alpha \tilde \opn_\alpha}$. Every factor in this product is a bounded operator, since $\opn_\alpha \ge 0$, $E_\alpha \ge 0$, and $\tau \ge 0$. Moreover, $\bili{\aphi}{\acre}_\Xspace = \sum_\alpha \tilde \aphi_\alpha \tilde \opa_\alpha^\dagger$, where $\tilde \aphi = \speps^d \bar\cU \aphi$, and thus $\coh{\aphi} = \prod_\alpha \E^{\tilde \aphi_\alpha \tilde \opa^\dagger_\alpha} \Omega$. Because $\opn_\alpha$ commutes with $\acre_{\alpha'}$ if $\alpha \ne \alpha'$, it suffices to consider the action of  
$\opB_\alpha=\E^{-\tau  \speps^d E_\alpha \tilde \opn_\alpha}$ on 
$\ww_\alpha (\tilde \aphi_\alpha) = \E^{\tilde \aphi_\alpha \tilde \opa^\dagger_\alpha} \Omega$.
$B_\alpha$ is bounded, so its application can be exchanged with the exponential series, and 
\beq
\begin{split}
\opB_\alpha \ww_\alpha (\tilde \aphi_\alpha)
&=
\sum_{n=0}^\infty
\frac{\tilde \aphi_\alpha^n}{n!}\;
B_\alpha \;
{\opa^\dagger_\alpha}^n \Omega
\\
&=
\sum_{n=0}^\infty
\frac{\tilde \aphi_\alpha^n}{n!}\;
\E^{-\tau E_\alpha  n}  \; {\opa^\dagger_\alpha}^n \Omega
=
\ww_\alpha (\E^{-\tau E_\alpha}\;\tilde \aphi_\alpha)
\end{split}
\eeq
In going from the first to the second line, it was used that ${\opa^\dagger_\alpha}^n \Omega$ is an eigenvector of $\tilde \opn_\alpha$ with eigenvalue $\speps^{-d}\; n$. 
Transforming back to the original basis completes the proof.
\end{proof}

\section{Inversion Formulas}
\label{Covarapp}

Let $\Bana$ be a Banach algebra with identity $\opId$ and $\opA_0, \ldots, \opA_\ntau \in \Bana$. Define $\Quad$ and $\Kuad$ in $M_{\ntau+1} (\bC) \otimes \Bana$ by 
\beq
\Quad_{\jtau,\jtau'}
=
\delta_{\jtau,\jtau'} \; 
\opId
-
\delta_{\jtau+1,\jtau'} 
\opA_\jtau  \; ,
\eeq
and 
\beq
\Kuad_{\jtau,\jtau'}
=
\Quad_{\jtau,\jtau'}
-
\delta_{\jtau,\ntau} \; \delta_{0,\jtau'} \; 
\opA_0
\eeq
where $\jtau, \jtau' \in \{ 0, \ldots, \ntau\}$.  The following lemma states the formulas for the inverses
\beq
\Quad^{-1} = \Cov
\qquad\mbox{and}\qquad
\Kuad^{-1} = \Gov \; , 
\eeq
for the general case where the $\opA_\jtau$ need not commute. 

\bigskip\begin{lemma}\label{InversesLemma}
For all $\opA_1, \ldots, \opA_\ntau \in \Bana$, $\Quad$ is invertible with inverse
\beq\label{Qinverse}
\Cov_{\jtau,\jtau'}
=
\True{\jtau' \ge \jtau} 
\;
\pli_{\itau=\jtau+1}^{\jtau'}
\opA_\itau
\eeq
(where the empty product equals $\opId$).
Let $\opP = \pli_{\itau=0}^\ntau \opA_\itau$. If $\opId - \opP$ is invertible, then
$\Kuad$ is invertible, with inverse
\beq\label{Kinverse1}
\Gov_{\jtau,\jtau'}
=
\True{\jtau' \ge \jtau} 
\;
\pli_{\itau=\jtau+1}^{\jtau'}
\opA_\itau
+
\pli_{\itau=\jtau+1}^\ntau
\opA_\itau\;\;
(\opId - \opP)^{-1}\;
\pli_{\itau=0}^{\jtau'}
\opA_\itau
\eeq
\end{lemma}

\begin{proof}
Considered as a matrix in the indices $\jtau$ and $\jtau'$, $\Quad$ is upper triangular, more precisely of the form identity plus a nilpotent term $R$. So it is invertible, and the inverse given is by a terminating Neumann series in powers of $R$. This gives \Ref{Qinverse}. 
Let $\cD_{\jtau,\jtau'} = \delta_{\jtau,\ntau} \; \delta_{0,\jtau'} \; \opA_0$ so that $\Kuad = \Quad - \cD$. Due to the special form of $\cD$, the resolvent equation
\beq
\Gov
=
\Cov
+
\Gov\; \cD \; \Cov
\eeq
reads for the matrix elements
\beq
\Gov_{\jtau,\jtau'}
=
\Cov_{\jtau,\jtau'}
+
\Gov_{\jtau,\ntau}\;
\opA_0\;
\Cov_{0,\jtau'}\; .
\eeq
Because $\opA_0\Cov_{0,\ntau} = \opP$, the equation for $\jtau' = \ntau$ reads
\beq
\Gov_{\jtau,\ntau}
=
\Cov_{\jtau,\ntau}
+
\Gov_{\jtau,\ntau}\;
\opP
\eeq
By hypothesis, $\opId - \opP$ is invertible, so 
\beq
\Gov_{\jtau,\ntau}
=
\Cov_{\jtau,\ntau}
\; 
(\opId - \opP)^{-1} 
\eeq
and
\beq
\Gov_{\jtau,\jtau'}
=
\Cov_{\jtau,\ntau}
+
\Cov_{\jtau,\ntau}\;
\; 
(\opId - \opP)^{-1} \; 
\opA_0\;
\Cov_{0,\jtau'}
\eeq
Inserting \Ref{Qinverse}, and noting that the indicator functions in that expression evaluate to $1$ in $\Cov_{\jtau,\ntau} $ for all $\jtau $, and in  $\Cov_{0,\jtau'}$, for all $\jtau'$, one gets \Ref{Kinverse1}.
\end{proof}

By inserting $1= \True{\jtau' \ge \jtau} + \True{\jtau' < \jtau}$, 
\Ref{Kinverse1} can be rewritten in the form
\beq\label{Kinverse2}
\begin{split}
\Gov_{\jtau,\jtau'}
& =
\True{\jtau' \ge \jtau} 
\;
\left[
\opId
+
\pli_{\itau=\jtau+1}^\ntau
\opA_\itau\;\;
(\opId - \opP)^{-1}\;
\pli_{\itau=0}^{\jtau}
\opA_\itau
\right]
\;
\pli_{\itau=\jtau+1}^{\jtau'}
\opA_\itau
\\
&+
\True{\jtau' < \jtau} \;
\pli_{\itau=\jtau+1}^\ntau
\opA_\itau\;\;
(\opId - \opP)^{-1}\;
\pli_{\itau=0}^{\jtau'}
\opA_\itau
\end{split}
\eeq
If $\opA_\jtau \opA_{\jtau'} = \opA_{\jtau'} \opA_\jtau$
for all $\jtau, \jtau' \in \{ 0, \ldots, \ntau\}$, 
this simplifies to
\beq\label{Govcommute}
\begin{split}
\Gov_{\jtau,\jtau'}
&=
\True{\jtau' \ge \jtau} \; 
(\opId - \opP)^{-1} \;
\pli_{\itau \in \{\jtau+1, \ldots , \jtau'\}} \opA_\itau
\\
&+
\True{\jtau' < \jtau} \; 
(\opId - \opP)^{-1} \;
\pli_{\itau \not\in \{\jtau'+1, \ldots , \jtau\}} \opA_\itau \; .
\end{split}
\eeq
If (as in the case where all auxiliary fields vanish) 
$\opA_\itau = \E^{-\eptau\Kin}$ with $\beta = \eptau \ntau$, then (denoting $\tau_\jtau = \jtau \eptau = \beta \frac{\jtau}{\ntau}$)
\beq\label{granolcov}
\begin{split}
\Gov_{\jtau,\jtau'}
=
(\opId - \E^{-\beta\Kin})^{-1}
\left[
\True{\tau_{\jtau'} \ge \tau_\jtau} \; 
\E^{- (\tau_{\jtau'} - \tau_\jtau) \Kin} \; 
+
\True{\tau_{\jtau'} < \tau_\jtau} \; 
\E^{- (\beta - (\tau_{\jtau'} - \tau_\jtau)) \Kin} 
\right]
\end{split}
\eeq
i.e.\ $\Gov_{\jtau,\jtau'}$ is  the standard time-ordered two-point function (``propagator'') for free bosons with kinetic term $\Kin$ at inverse temperature $\beta$, evaluated at the discrete times $\tau_\jtau$ and $\tau_{\jtau'}$. 

A sufficient condition for $\opId - \opP$ to be invertible is that $\norm{\opP} < 1$, which in turn holds if $\norm{\opA_\itau} \le 1$ for all $\itau \in \{0, \ldots, \ntau\}$ and $\norm{\opA_{\itau_0} } < 1$ for some $\itau_0$. 
For $\opA_\jtau = \E^{-\eptau \Kin} \; \E^{\I \sqrt{\eptau} h_\jtau}$ with real $h_\jtau$, a sufficient condition is that $\Kin > 0$. These conditions are not necessary for invertibility of  $\opId - \opP$. But they are necessary for the convergence of the coherent-state functional integral of the grand-canonical ensemble. The series 
\beq
\sum_{N=0}^\infty \frac{\E^{N\beta\mu}}{N!} \; {\bili{\bar\aphi_\ntau}{\aphi_0}_\Xspace}^N
=
\E^{\E^{\beta\mu}\; \bili{\bar\aphi_\ntau}{\aphi_0}_\Xspace}
\eeq
converges for all values of $\beta\mu$, $\aphi_\ntau$ and $\aphi_0$. Upon formal exchange of the summation and integration in \Ref{grancanZ}, the resulting term in the exponent is just the difference $\cD$ of $\Quad$ and $\Kuad$, with $\opA_0 = \E^{\beta\mu} \opId$ of norm $\norm{\opA_0} = \E^{\beta\mu}$. If $\mu > 0$, this is bigger than one, and invertibility of $\opP$ can fail for particular $\Xspace$, and even if it holds, it will typically not hold with a norm of the inverse that is uniform in $\Xspace$. Moreover, the quadratic form does not have a positive real part, and therefore the Gaussian integral diverges, which in turn invalidates the exchange of summation and integration. Similarly, if one wants to use the complex contour integral for enforcing that the particle number is $N$, one should take integration radius $|z| < 1$ (see (\ref{contourli})).

\section{Baker-Campbell-Hausdorff Formulas}
\label{BCHapp}

The  Baker-Campbell-Hausdorff expansion is well-known and has been considered in many works, often with the aim to simplify the recursion for generating more and more terms in the expansion \cite{Wilcox,Gilmore,VanBruntVisser}. The point here is, on the other hand, to avoid combinatorics by deriving and bounding remainder terms in integral form. I assume in the following that $\opA$ and $\opB$ are elements of some Banach algebra, and show that 
\beq
\begin{split}
\E^{t (\opA+\opB)}
&=
\E^{t\opA} \E^{t \opB}\;
( \opId + \opH_2 (t))
\\
&=
\E^{t\opA} \; \E^{t \opB} \; \E^{\frac{t^2}{2} [\opB,\opA]} \; 
(\opId + \opH_3(t) ) 
\end{split}
\eeq
with $\opH_2 (t)  = O(t^2)$ and $\opH_3 = O(t^3)$ analytic in $\opA$ and $\opB$, and both having explicit bounds in terms of the norms of $\opA$ and $\opB$. The application will be with $t = \eptau \sim \ntau^{-1}$ very small. It is then of course irrelevant whether the error term is in the exponent or `downstairs', as it is small in norm compared to $1$, so one can put it in the exponent simply by taking the logarithm. Indeed it is in general an advantage if the error term is not in the exponent. The way the expansion is generated also allows to go to any higher order, with an according remainder term. Variants of the Lie product formula including higher order terms in the exponent, and with accordingly improved convergence rate as $\ntau \to \infty$, can be proven straightforwardly using these results. They are useful for generalizations of the systems discussed here, e.g.\ for certain Fermi-Bose mixtures.  

Let
\beq
\opV (t) = \E^{-t\opA} \; \E^{t(\opA+\opB)} \; .
\eeq
Then $\opV (0) = \opId$ and 
\beq
\sfrac{\dd}{\dd s} \opV(s) 
=
\E^{-s \opA} \; \opB \; \E^{s \opA} \; \opV (s) 
\eeq
hence 
\beq
\opV(t) 
=
\opId
+
\int_0^t \dd s\; \E^{-s \opA} \; \opB \; \E^{s \opA} \; \opV (s) 
\eeq
Let  $\opW (t) = \E^{-t\opB} \; \opV(t)$. Then $\opW(0) = \opId$ and 
\beq
\sfrac{\dd}{\dd s} \opW(s) 
=
\E^{-s \opB} \; 
\left( 
\E^{-s \opA} \; \opB \; \E^{s \opA} - \opB
\right)
\; \E^{s \opB} \; \opW (s) 
\eeq
The usual interpolation
\beq
\E^{-s \opA} \; \opB \; \E^{s \opA} - \opB
=
\int_0^s \dd r \; \E^{-r\opA} \; [\opB,\opA]\; \E^{r \opA}
\eeq
gives
\beq\label{Wstern}
\sfrac{\dd}{\dd s} \opW(s) 
=
\E^{-s \opB} \; 
\int_0^s \dd r \; \E^{-r\opA} \; [\opB,\opA]\; \E^{r \opA}
\;
\E^{s \opB} \; \opW (s) 
\eeq
Since $\E^{t(\opA + \opB)} = \E^{t\opA} \E^{t \opB} \opW(t)$, integration of \Ref{Wstern} gives
\beq
\E^{t(\opA + \opB)} 
= 
\E^{t\opA} \E^{t \opB}\;
( 1 + \opH_2 (t))
\eeq
with
\beq
\opH_2 (t) 
=
\int_0^t \dd s \; \E^{-s\opB} \;
\int_0^s \dd r \; \E^{-r\opA} \; 
[\opB,\opA]\; \E^{-(s-r) \opA}\;
\E^{s(\opA + \opB)}
\eeq
For $\eptau \ge 0$, a straightforward (and in general wasteful) bound for $\opH_2$ leads to
\beq
\norm{
\E^{-\eptau(\opA + \opB)} 
- 
\E^{-\eptau\opA} \E^{-\eptau \opB}}
\le
\eptau^2 \; 
\E^{\eptau (2 \norm{\opB} + 4 \norm{\opA})}\;
\frac12 \norm{[A,B]} \; .
\eeq
Better bounds can be obtained if one or both of $\opA$ and $\opB$ are positive.)

To go one order higher, go back to \Ref{Wstern} and expand further, by writing
\beq
\E^{-r\opA} \; [\opB,\opA]\; \E^{r \opA}
=
[\opB,\opA]\
+
\int_0^r \dd q \; \E^{-q\opA} \; [[\opB,\opA], \opA]\; \E^{q \opA}
\eeq
and iterating, to get
\beq
\begin{split}
&\E^{-s \opB} \; 
\int_0^s \dd r \; \E^{-r\opA} \; [\opB,\opA]\; \E^{r \opA}
\;
\E^{s \opB} 
\\
=&
s\; \E^{-s \opB} \; 
[\opB,\opA] \; 
\E^{s \opB} 
+
\E^{-s \opB} \; 
\int_0^s \dd r \; 
\int_0^r \dd q \; \E^{-q\opA} \; [[\opB,\opA], \opA]\; \E^{q \opA}\;
\;
\E^{s \opB} 
\\
=&
s\; [\opB,\opA]
+
\tilde \opH_3 (s)
\end{split}
\eeq
with
\beq
\begin{split}
\tilde \opH_3 (s)
&=
- s \int_0^s \dd r \; \E^{-r\opB} [\opB , [\opB , \opA ]] \; \E^{r \opB} 
\\
&+
\E^{-s \opB} \; 
\int_0^s \dd r \; 
\int_0^r \dd q \; \E^{-q\opA} \; [[\opB,\opA], \opA]\; \E^{q \opA}\;
\;
\E^{s \opB} 
\end{split}
\eeq
involving double commutators. The equation for $\opW$ then implies
\beq
\sfrac{\dd}{\dd s} 
\left(
\E^{-\frac{s^2}{2} [\opB,\opA]} \; \opW(s) 
\right)
=
\E^{-\frac{s^2}{2} [\opB,\opA]} \; \tilde \opH_3 (s)\; \opW(s) 
\eeq
so integration over $s$ gives
\beq
\E^{t (\opA+\opB)}
=
\E^{t\opA} \; \E^{t \opB} \; \E^{\frac{t^2}{2} [\opB,\opA]} \; 
(\opId + \opH_3(t) ) 
\eeq
with 
\beq
\opH_3(t) 
=
\int_0^t \dd s \; 
\E^{-\frac{s^2}{2} [\opB,\opA]} \; \tilde \opH_3 (s)\;
\E^{-\frac{s^2}{2} [\opB,\opA]} \; \E^{-s\opB} \; \E^{-s\opA} \; \E^{s (\opA+\opB)} \; .
\eeq
Obviously this iteration can be continued to arbitrary orders and it gets more and more tedious as one goes on. But at each step one has an integral formula for the remainder which is useful for bounds. 

\section{Integration by Parts}\label{intbypartssec}
For $\oh \in \bR$ let
\beq\label{Diffeptau}
\Diff_\eptau (\oh)
=
\E^{-\I \sqrt{\eptau} \oh}
-
1 + \I \sqrt{\eptau} \oh \; .
\eeq
By Taylor expansion of the exponential, $\Diff_\eptau (\oh) = - \frac{\eptau}{2} \oh^2 + O(\eptau^{\frac32} \oh^3)$. The following elementary statements are used in the sequel: 
\beq\label{hilf1}
\Diff_\eptau (\oh)
=
\I \sqrt{\eptau} \oh \;  \Theta_\eptau (\oh) 
\quad
\mbox{with}
\quad 
\Theta_\eptau (\oh)  = \int_0^1 \left(
1 - \E^{-\I \sqrt{\eptau} s \, \oh} 
\right) \; \dd s \; 
\eeq
and
\beq\label{hilf2}
\I \sqrt{\eptau} \; 
\frac{\partial \Theta_\eptau}{\partial \oh}
=
- \frac{\eptau}{2}
+
\I \oh \eptau^{\frac32} R_\eptau (\oh)
\eeq
with $R_\eptau (\oh) = \int_0^1 s^2 \dd s \int_0^1 \dd t \; \E^{-\I \sqrt{\eptau} \oh s t}$, 

 r
Recall the notation \Ref{hexpecdef}. Let $\Bana$ be a Banach algebra, and let $F: \bR^\XTspace \to \Bana$ be in $L^1 ( \bR^\XTspace, \dd \mu_\WWW)$, i.e.\ $\hexpec{|F(h)|} < \infty$.

\bigskip\begin{lemma}\label{intbypartslemma}
Let $F: \bR^\XTspace \to \Bana$, $h \mapsto F(h)$ be differentiable in $h$ and assume that $\hexpec{|F(h)|} < \infty$, and  for all $(\jtau,\xx) \in \XTspace$, $\hexpec{ \abs{\frac{\partial F}{\partial h_{\jtau,\xx}}}}  < \infty$. 

(a) For all $(\jtau,\xx) \in \XTspace$, 
\beq
\begin{split}\label{anice1}
\hexpec{\Diff_\eptau (h_{\jtau,\xx} ) \; F(h)}
=
&- 
\sfrac{\eptau}{2}\WW_{\xx,\xx} \; \hexpec{F(h)}
\\
&+
\eptau^{\frac32} \; \I \WW_{\xx,\xx} \; \hexpec{F(h)\; h_{\jtau,\xx} \; R_\eptau (h_{\jtau,\xx})}
\\
&+
\I \sqrt{\eptau}
\hexpec{ h_{\jtau,\xx}\; \Theta_\eptau (h_{\jtau,\xx}) \; 
\int_\xy
\WW_{\xx,\xy}\;
 \frac{\partial F}{\partial h_{\jtau,\xy}}\; .
 }
\end{split}
\eeq
(b) If, in addition, the function $F$ depends on $\eptau=\beta/\ntau$ and if there is $M > 0$ such that for all $\ntau  \ge 1$, all $\jtau \in \{ 0, \ldots ,\ntau\}$ and all $\xx$ and $\xy \in \Xspace$ 
\beq\label{117}
\hexpec{
\abs{h_{\jtau,\xx}} \; 
\norm{F(h)}
}
\le M 
\eeq
and
\beq\label{118}
\hexpec{
\abs{h_{\jtau,\xx}}^2 \; 
\norm{ \frac{\partial F}{\partial h_{\jtau,\xy}}}
} \le M \; \sqrt{\eptau} 
\eeq
then 
\beq\label{3halbe}
\norm{
\hexpec{\left[
\Diff_\eptau (h_{\jtau,\xx} )
+ 
\sfrac{\eptau}{2}\WW_{\xx,\xx} 
\right] \; F(h)}}
\le
\eptau^{\frac32} \; M \; \left(
\abs{\WW_{\xx,\xx}} + \int_\xy \abs{\WW_{\xx,\xy}}
\right) \; .
\eeq
In particular, for all $\delta > 0$
\beq\label{3halbe-delta}
\eptau^{-\frac32 + \delta} \; 
\hexpec{\left[
\Diff_\eptau (h_{\jtau,\xx} )
+ 
\sfrac{\eptau}{2}\WW_{\xx,\xx} 
\right] \; F(h)}
\gtoas{\eptau \to 0} 0 \;  .
\eeq
\end{lemma}

\begin{proof}
By \Ref{hilf1}, 
\beq
\hexpec{
\Diff_\eptau (h_{\jtau,\xx} ) \; F(h)} 
=
\I \sqrt{\eptau}\; 
\hexpec{
h_{\jtau,\xx} \; G_\eptau(h)} 
\eeq
with $G_\eptau (h) 
=
F(h) \;
\Theta_\eptau (h_{\jtau,\xx})$ . 
Integration by parts w.r.t.\ the Gaussian measure $\dd\mu_\WWW$ gives  
\beq\label{intbypartsWWW}
\hexpec{
h_{\jtau,\xx} \; G_\eptau (h) }
=
\hexpec{
\int_\xy
\WW_{\xx,\xy}\;
 \frac{\partial G_\eptau}{\partial h_{\jtau,\xy}}}\; .
\eeq
When the derivative with respect to $h_{\jtau,\xx}$ acts on $\Theta_\eptau$, the sum over $\xy$ reduces to the term $\xy=\xx$, and \Ref{hilf2} gives the first two summands on the right hand side of \Ref{anice1}. The action of the derivative on $F$ gives the third summand. 

To prove (b), move the first term on the right hand side of \Ref{anice1} to the left hand side.
Since $h_{\jtau,\xx}$ is real, $\abs{R_\eptau ( h_{\jtau,\xx})} \le \frac13$, hence
\beq
\abs{\WW_{\xx,\xx} \; \hexpec{F(h)\; h_{\jtau,\xx} \; R_\eptau (h_{\jtau,\xx})}}
\le 
\frac{M}{3} \; \WW_{\xx,\xx} 
\eeq
by hypothesis of part (b). Similarly, noting that $\abs{\Theta_\eptau (h_{\jtau,\xx}} \le \frac12  \sqrt{\eptau} \abs{h_{\jtau,\xx}}$, one obtains by (\ref{118}) 
\beq
\begin{split}
\abs{
\hexpec{ h_{\jtau,\xx}\; \Theta_\eptau (h_{\jtau,\xx}) \; 
\int_\xy
\WW_{\xx,\xy}\;
\frac{\partial F}{\partial h_{\jtau,\xy}}
}}
\le
\frac{M}{2} \eptau \; \int_\xy \abs{\WW_{\xx,\xy}}
\end{split}
\eeq
Thus both remainder terms are bounded by $\eptau^{\frac32}$ times a constant, which implies \Ref{3halbe} and \Ref{3halbe-delta}.
\end{proof}

Note that because $\WWW_{(\jtau,\xx), (\jtau', \xx')} = \delta_{\jtau,\jtau'} \; \WW_{\xx, \xx'}$ is local in `time" $\jtau$, the summation on the right hand side of \Ref{intbypartsWWW} involves only the spatial variable $\xy$, but no extra summation over a `time-slice' index $\jtau'$. This is essential since every $\jtau$ summation goes over $\ntau + 1 \sim \eptau^{-1}$ variables, hence potentially creates inverse powers of $\eptau$ in remainder estimates. 

In the standard representation of the Gaussian expectation $\hexpec{\cdot}$ as a sum over Feynman graphs, the first summand in \Ref{anice1} corresponds to the tadpole (self-contraction) graph.

\section{Determinants, Permanents, and Laplacians}\label{detpermlapapp}

\begin{lemma}\label{abqlemma}
Let $a,b,q \in \bC$ and $\Re q > 0$. Then
\beq\label{qpi}
\sfrac{q}{2 \pi \I}  
\int _\bC 
\E^{-q |z|^2 + \bar a z + \bar z b}\;
\dd \bar z \wedge \dd z
=
\E^{\frac{\bar a b}{q}} \; .
\eeq
\end{lemma}

\begin{proof}
Call the left hand side of \Ref{qpi} $I_q (a,b)$. Because $\Re q > 0$ this integral exists, and it is absolutely convergent, hence can be taken as the limit $\Rad \to \infty$ of the integral over $\bC_\Rad = \{z \in \bC : |z| \le \Rad\}$, which will be denoted by $I_q^\Rad (a,b)$.
At fixed $\Rad$, expand in $\bar a $ and $b$. The compactness of the integration region justifies the exchange of summations and integral. This gives
\beq
\begin{split}
I_q^\Rad (a,b)
&=
\sfrac{q}{2 \pi \I}  
\sum_{m,n=0}^\infty \frac{\bar a^m b^n}{m!n!} 
\int \E^{-q |z|^2 } \; z^m \bar z^n \; 
{\dd \bar z \wedge \dd z}
\\
&=
\sum_{m,n=0}^\infty \frac{\bar a^m b^n}{m!n!} 
q \int_0^\Rad 2r \dd r\; \E^{-q r^2} \; r^{m+n} \; 
\int_0^{2\pi}
\sfrac{\dd\theta}{2\pi} \E^{\I \theta (m-n)}
\\
&=
\sum_{m=0}^\infty \frac{(\bar a b)^m}{m!^2} 
q \int_0^{\Rad^2} s^m\ \E^{-qs} \; \dd s 
\end{split}
\eeq
The last integral is bounded in absolute value by $(\Re q)^{-m-1} m!$, so the series converges absolutely and $\Rad \to \infty$ can be taken under the sum. Finally, 
\beq\label{iliR}
\lli_{\Rad \to \infty}
q \int_0^{\Rad^2} s^m\ \E^{-qs} \; \dd s 
=
q^{-m} m! 
\eeq 
and the sum over $m$ then gives the exponential. To see \Ref{iliR}, first transform to $t=qs$ (calling $\Rad^2 =L$):
\beq
q \int_0^{L} s^m\ \E^{-qs} \; \dd s 
=
q^{-m} \int_{\gamma_1} \E^{-t} t^{m} \dd t
\eeq
where $\gamma_1$ is the straight line from $0$ to $q L \in \bC$. 
The  integrand is an entire function of $s$, so the contour can be deformed to 
$\gamma_2 + \gamma_3$, where $\gamma_2$ is the straight line from $0$ to $L \Re q$ and $\gamma_3$ the straight line parallel to the imaginary axis from $L \Re q$ to $Lq$. The latter can be parametrized by $t=Lu + \I Lv r$, $u=\Re q$, $v =\Im q$
and $r\in [0,1]$. Then 
\beq
\abs{
\int_{\gamma_3} \E^{-t} t^{m} \dd t}
\le
L^{m+1} \; 
\E^{-u L} \; |v| \; (u^2+v^2)^{m/2}
\eeq
which vanishes as $L \to \infty$.
\end{proof}

\newcommand{\NN}{{\rm N}}
\begin{lemma}\label{gauzz}
Let $Q \in M_\NN(\bC)$ with $Q+Q^\dagger > 0$, and $a,b \in \bC^\NN$. 
Then
\beq\label{Qdetformula}
\det Q
\int_{\bC^\NN} 
\pli_{n=1}^\NN \sfrac{\dd \bar z_n \wedge \dd z_n}{2\pi\I} \; 
\E^{\bracket{z}{ Qz} + \bracket{a}{z} + \bracket{z}{b}}\;
=
\E^{\bracket{a}{Q^{-1} b}}
\; .
\eeq
Here $\bracket{a}{b}=\bar a^\top b$ denotes the standard inner product on $\bC^\NN$.
\end{lemma}

\begin{proof}
The proof will by reduction to Lemma \ref{abqlemma}, using a diagonalization argument. 
Let $Q_0 = \frac12 (Q+Q^\dagger)$ and $H = \frac{1}{2\I} (Q-Q^\dagger)$
then $Q = Q_0 + \I H$, $H=H^\dagger$, $Q_0 = Q_0^\dagger > 0$. 
Thus $Q$ satisfies the hypotheses of Lemma \ref{posplusim}, and the same decomposition as in the proof of that lemma will be used here: let $B = Q_0^{\frac12}$ be the positive square root of $Q_0$ and $A = B^{-\frac12} H B^{-\frac12}$.  Then $A = A^\dagger$, so $1+ \I A$ is normal, hence diagonalized by a unitary transformation: $1 + \I A = U^\dagger D U$ with diagonal $D$ and $U^\dagger = U^{-1}$. $Q = B (1 + \I A) B$, so transforming to $w=U B z$ gives 
\beq
\begin{split}
J 
&=
\int_{\bC^\NN} 
\pli_{n=1}^\NN \sfrac{\dd \bar z_n \wedge \dd z_n}{2\pi\I} \; 
\E^{\bracket{z}{ Qz} + \bracket{a}{z} + \bracket{z}{b}}\;
\\
&=
(\det B)^{-2}
\int_{\bC^\NN} 
\pli_{n=1}^\NN \sfrac{\dd \bar w_n \wedge \dd w_n}{2\pi\I} \; 
\E^{\bracket{w}{ Dw} + \bracket{U B^{-1} a}{w} + \bracket{w}{U B^{-1}  b}}
\end{split}
\eeq
The Jacobian from the transformation with $U$ is one because the determinants $U$ and $U^\dagger$ appear and multiply to $1$. 
The $w$-integral factorizes because $D$ is diagonal. The real part of every eigenvalue of $D$ equals $1$, so Lemma \ref{abqlemma} applies. The product of eigenvalues is $\det D^{-1}$, which combines with $\det B^{-2}$ to $\det Q^{-1}$. 
The product of exponentials gives a quadratic form in the exponent as
\beq
\bracket{UB^{-1} a}{D^{-1} U B^{-1} b}
=
\bracket{a}{B^{-1} U^\dagger D^{-1} U B^{-1} b}
=
\bracket{a}{Q^{-1} b} 
\eeq
so  \Ref{Qdetformula} follows. 
\end{proof}

Permanent formulas, such as (\ref{perm1}) and (\ref{Qhpermanent}), follow from Lemma \ref{gauzz} by differentiation. Denoting 
\beq
\langle f \rangle_Q
=
\det Q
\int_{\bC^\NN} 
\pli_{n=1}^\NN \sfrac{\dd \bar z_n \wedge \dd z_n}{2\pi\I} \; 
\E^{\bracket{z}{ Qz} }\;
f(z) \; ,
\eeq
Let $a,b \in \bR^\NN$. 
For $m_1, \ldots, m_K$ and $n_1, \ldots , n_K$ in $\{1, \ldots, \NN\}$
\beq\label{228}
\begin{split}
\langle \pli_{k=1}^K \bar z_{m_k} z_{n_k}\rangle_Q
&=
\left\lbrack
\pli_{k=1}^K \sfrac{\del}{\del b_{m_k}} \; \sfrac{\del}{\del a_{n_k}} \;
\langle 
\E^{\bracket{a}{z} + \bracket{z}{b}}\;
\rangle_Q
\right\rbrack_{a=b=0}
\\
&=
\left\lbrack
\pli_{k=1}^K \sfrac{\del}{\del b_{m_k}} \; \sfrac{\del}{\del a_{n_k}} \;
\E^{\bracket{a}{Q^{-1} b}}
\right\rbrack_{a=b=0}
\\
&=
\left\lbrack
\sfrac{1}{K!}
\pli_{k=1}^K \sfrac{\del}{\del b_{m_k}} \; \sfrac{\del}{\del a_{n_k}} \;
\bracket{a}{Q^{-1} b}^{K}
\right\rbrack_{a=b=0}
\\
&=
\sum_{\pi \in \Perm{K}} \pli_{k=1}^K Q^{-1}_{n_k, m_{\pi(k)}} \; .
\end{split}
\eeq
As in Section 2.3 of \cite{msbook}, the generating function identity can be used to rewrite the Gaussian integral of polynomials in terms of the field space Laplacian.
The result is stated in the following Lemma. (The operator $\Quad$ in that lemma may be $h$-dependent.)

\begin{lemma} 
Let $\Quad$ be a linear operator on $\bC^\XTspace$ and $\Quad + \Quad^\dagger > 0$. Then 
\beq
\det \Quad \int \DD \aphi \; 
\E^{-\bili{\bar\aphi}{\Quad \aphi}_\XTspace} \;
\pli_{k=1}^K \bar \aphi_{\jtau_k,\xx_k} \aphi_{\jtau'_k,\xx'_k}
=
\left\lbrack
\E^{\Delta_{\Cov}} \; 
\pli_{k=1}^K \bar \aphi_{\jtau_k,\xx_k} \aphi_{\jtau'_k,\xx'_k}
\right\rbrack_{\aphi =0}
\eeq
where $\Cov =\Quad^{-1}$ and 
\beq
\Delta_\Cov 
=
\sum_{\jtau,\jtau',\xx,\xx'} \Cov_{(\jtau_k,\xx_k),(\jtau'_k,\xx'_k)} \; 
\sfrac{\del}{\del\aphi_{\jtau,\xx}}
\sfrac{\del}{\del\bar\aphi_{\jtau',\xx'}}
\eeq
With the notation $\sfrac{\de}{\de\aphi_{\jtau,\xx}} = \eptau^{-1} \speps^{-d} \sfrac{\del}{\del\aphi_{\jtau,\xx}}$, the Laplacian can be rewritten as 
\beq
\Delta_\Cov 
=
\bili{\sfrac{\de}{\de\aphi} }{ \Cov \sfrac{\de}{\de\bar\aphi}}_\XTspace
=
\bili{ \Cov \sfrac{\de}{\de\bar\aphi}}{\sfrac{\de}{\de\aphi} }_\XTspace
\eeq
\end{lemma}

Techniques using these Laplacians will be useful in doing convergent expansions, in analogy to \cite{SW}.

By (\ref{228}), 
\beq\label{Quadperm}
\det \Quad \int \DD \aphi \; 
\E^{-\bili{\bar\aphi}{\Quad \aphi}_\XTspace} \;
\pli_{k=1}^K \bar \aphi_{\jtau_k,\xx_k} \aphi_{\jtau'_k,\xx'_k}
=
\sum_{\pi \in \Perm{K}} \pli_{k=1}^K \Cov_{(\jtau'_k,\xx'_k),(\jtau_{\pi(k)},\xx_{\pi(k)})} \; .
\eeq

\noindent
{\bf Acknowledgement. } This work started within DFG collaborative research group FOR 718: {\em Analysis and Stochastics in Complex Physical Systems}. I would like to thank Stefan Adams and Wolfgang K\"onig for discussions about Brownian motion representations of Bose systems, which motivated me to consider the canonical ensemble in the coherent-state integral representation discussed here. Much of this work was done during visits to the Mathematics Department at UBC, and I would like to thank Joel Feldman for his hospitality and many discussions, especially during the last two years. This work is supported by Deutsche Forschungsgemeinschaft (DFG, German Research Foundation) under Germany's Excellence Strategy  EXC-2181/1 - 390900948 (the Heidelberg STRUCTURES Cluster of Excellence).

\end{document}